\documentclass[11pt,english]{article}
\usepackage[T1]{fontenc}
\usepackage[latin9]{inputenc}
\usepackage{geometry}
\usepackage{verbatim}
\usepackage{bm}
\usepackage{amsmath}
\usepackage{amsthm}
\usepackage{amssymb}
\usepackage{authblk}
\usepackage[hidelinks]{hyperref}
\usepackage{float}

\date{}

\usepackage{arydshln}
\usepackage{array,booktabs}
\newcolumntype{C}[1]{>{\centering\let\newline\\\arraybackslash\hspace{0pt}}m{#1}}

\usepackage{lipsum}
\newcommand\blfootnote[1]{%
  \begingroup
  \renewcommand\thefootnote{}\footnote{#1}%
  \addtocounter{footnote}{-1}%
  \endgroup
}

\theoremstyle{definition}
 \newtheorem{example}{Example}[section]
 \newtheorem{defn}[example]{Definition}
 \newtheorem{notation}[example]{Notation}
\theoremstyle{remark} 
 \newtheorem{rem}[example]{Remark}
\theoremstyle{plain} 
 \newtheorem{prop}[example]{Proposition}
 \newtheorem{lem}[example]{Lemma}
 \newtheorem{cor}[example]{Corollary}
  
\theoremstyle{plain} 
 \newtheorem{thm}{Theorem}
 \newtheorem{conjecture}{Conjecture}

\begin{document}

\title{The quantum Witten-Kontsevich series and one-part double Hurwitz
numbers}

\author{Xavier Blot\footnote{Institut Math\'{e}matique de Jussieu, Sorbonne Universit\'e, 4 place Jussieu, 75005 Paris, France; email: xavier.blot@imj-prg.fr}} 

\maketitle


\begin{abstract}
We study the quantum Witten-Kontsevich series introduced by Buryak,
Dubrovin, Gu\'er\'e and Rossi in~\cite{Buryak_2019} as the
logarithm of a quantum tau function for the quantum KdV hierarchy.
This series depends on a genus parameter $\epsilon$ and a quantum
parameter $\hbar$. When $\hbar=0$, this series restricts to the
Witten-Kontsevich generating series for intersection numbers of psi
classes on moduli spaces of stable curves.

We establish a link between the $\epsilon=0$ part of the quantum
Witten-Kontsevich series and one-part double Hurwitz numbers. These
numbers count the number of non-equivalent holomorphic maps from a Riemann
surface of genus $g$ to $\mathbb{P}^{1}$ with a complete ramification over $0$, a prescribed ramification profile over $\infty$ and a given number of simple ramifications elsewhere. Goulden, Jackson and Vakil
proved in~\cite{goulden2005towards} that these numbers have the
property to be polynomial in the orders of ramification over $\infty$.
We prove that the coefficients of these polynomials are the coefficients
of the quantum Witten-Kontsevich series.

We also present some partial results about the full quantum Witten-Kontsevich power series.
\end{abstract}

\tableofcontents{}

\section{Introduction}

\subsection{Overview}

The discovery of deep connections between enumerative geometry and
integrable systems was initiated by the Witten conjecture \cite{witten1990two}
proved by Kontsevich \cite{kontsevich1992intersection}. The Witten conjecture states that a generating series of intersection
numbers of the so-called $\psi$-classes on the moduli spaces of curves is the logarithm of a tau function of the Korteweg\textendash de
Vries (KdV) hierarchy. This result also revealed the important role
played by tau functions at the interplay of these two fields. Many
examples of such connections formulated in terms of tau function followed,
see \cite{witten1993algebraic,faber2010tautological,Okounkov_2000,Okounkov_2006}
for various examples. The statement has always the same formulation;
the logarithm of a tau function of some integrable hierarchy
is a generating series of numbers with a geometrical content. This
paper is the first step in the extension of these results to quantum integrable hierarchies.

In a more systematic approach, Dubrovin and Zhang \cite{dubrovin2001normal}
constructed integrable hierarchies associaciated to semisimple cohomological
field theories (CohFT). In their construction, the potential of the
CohFT is the logarithm of a tau function of the corresponding Dubrovin-Zhang
(DZ) hierarchy. For example, when the CohFT is trivial, the DZ hierarchy
is KdV and its potential is the Witten-Kontsevich series.

In 2014, Buryak \cite{buryak2015double} constructed another type
of integrable hierarchies associated to CohFTs called the double ramification (DR) hiearchies. It was conjectured
\cite{buryak2015double,BGR_2019} that the DR hierarchies can be obtained from the DZ hierarchies by a change of coordinates. A stronger version
of the DR/DZ conjecture states that a tau function of the DR hierarchy
is given by the potential of the CohFT up to a term related to this
change of variable. 

Remarkably, the DR hierarchies were quantized by Buryak and Rossi
in \cite{BuryakRossi2016}. Moreover, Buryak, Dubrovin, Gu\'er\'e and
Rossi introduced their quantum tau functions in \cite{Buryak_2019}.
In the classical and quantum settings, the tau functions of the DR
hierarchies were constructed indirectly using the tau symmetric hamiltonian
structure of the DR hierarchies \cite{buryak2018tau,Buryak_2019}.
In the classical setting, the DR/DZ conjecture gives an interpretation
of a DR tau function as the potential of the corresponding CohFT.
However, in the quantum setting, nothing was known about the tau functions
(except that we recover the classical DR tau function in the classical
limit $\hbar=0$), in particular an interpretation of their coefficients
was missing. 

In this paper, we study the first example of a quantum tau function
by investigating the logarithm of a quantum tau function associated to the trivial CohFT, we denote this series by $F^{\left(q\right)}$. In this case, the quantum DR hierarchy corresponds
to a quantization of KdV and the classical limit of the logarithm of the tau function,
$F^{\left(q\right)}\vert_{\hbar=0}$, is given by the Witten-Kontsevich
series. The series $F^{\left(q\right)}$ is then called the quantum
Witten-Kontevich series.

\begin{table}[H]
\begin{tabular}{c|c|c|c|c}
 & $\hbar^{0}$ & $\hbar^{1}$ & $\hbar^{2}$ & \tabularnewline
\cline{1-4} 
$\epsilon^{0}$ & {\renewcommand{\arraystretch}{5}
\begin{tabular}{C{4,5cm}}
$\frac{t_{0}^{3}}{6}+\frac{t_{0}^{3}t_{1}}{6}+\frac{t_{0}^{4}t_{2}}{24}+\cdots$
\end{tabular}} & {\renewcommand{\arraystretch}{3}
\begin{tabular}{C{4,5cm}}
{\renewcommand{\arraystretch}{1,5}
$\begin{array} {c}
  \,\\
  \frac{t_{2}}{24}+\frac{t_{0}t_{3}}{24}+\frac{t_{1}t_{2}}{24}+\\
  \, \frac{t_{1}^{2}t_{2}}{24}+\frac{t_{0}t_{2}^{2}}{24}+\cdots\\
 \,
\end{array}$} \\

\cdashline{1-1}

{\renewcommand{\arraystretch}{1,5}
$\begin{array} {c}
\, \\
 \frac{t_{0}}{24}+\frac{t_{0}t_{1}}{24}+\frac{t_{0}^{2}t_{2}}{48}+ \\
 \, \frac{t_{0}t_{1}^{2}}{24}+\cdots\\
\,
\end{array}$} \\
\end{tabular}} & {\renewcommand{\arraystretch}{2,5}
\begin{tabular}{C{4,5cm}}
$\frac{1}{1920}t_{6} \,\, + \,\, \frac{1}{1920}t_{0}t_{7} \,\, + \,\, \frac{1}{480}t_{1}t_{6}+\cdots$ \\

\cdashline{1-1}
$\frac{1}{576}t_{4} \,\, + \,\, \frac{1}{576}t_{0}t_{5} \,\, + \,\, \frac{1}{192}t_{1}t_{4}+\cdots$\\

\cdashline{1-1}
$\frac{7}{5760}t_{2}+\frac{7}{5760}t_{0}t_{3}+\frac{7}{1920}t_{1}t_{2}+\cdots$ \\

\end{tabular}} & \tabularnewline
\cline{1-4} 
$\epsilon^{2}$ & {\renewcommand{\arraystretch}{5}
\begin{tabular}{C{4,5cm}}
$\frac{t_{1}}{24}+\frac{t_{0}t_{2}}{24}+\frac{t_{1}^{2}}{48}+\cdots$
\end{tabular}} & {\renewcommand{\arraystretch}{4}
\begin{tabular}{C{4,5cm}}
$\frac{1}{720}t_{5}+\frac{1}{720}t_{0}t_{6}+\frac{1}{240}t_{1}t_{5}+\frac{1}{120}t_{2}t_{4}+\frac{31}{5760}t_{3}^2+\cdots$ \\

\cdashline{1-1}

$\frac{1}{576}t_{3}+\frac{1}{576}t_0t_{4}+\frac{1}{192}t_{1}t_{3}+\frac{1}{288}t_{2}^2+\cdots$\\

\cdashline{1-1}

$\frac{1}{5760}+\frac{1}{2880}t_{1}+\frac{1}{2880}t_0t_{2}+\frac{t_{1}^{2}}{1920}+\cdots$\\
\end{tabular}} &  & \tabularnewline
\cline{1-4} 
$\epsilon^{4}$ & {\renewcommand{\arraystretch}{2,5}
\begin{tabular}{C{4,5cm}}
$\begin{array} {c}
\frac{t_{4}}{1152}+\frac{t_{0}t_5}{1152}+\frac{1}{384}t_{1}t_{4}+ \\
\frac{29}{5760}t_{2}t_{3}+\cdots \\
\end{array}$
\end{tabular}} &  &  & \tabularnewline
\cline{1-4} 
\multicolumn{1}{c}{} & \multicolumn{1}{c}{} & \multicolumn{1}{c}{} & \multicolumn{1}{c}{} & \tabularnewline
\end{tabular}\caption{\label{tab: first terms qWK}First terms of the quantum Witten-Kontsevich.}

\end{table}

In Table~\ref{tab: first terms qWK}, we give the first terms of $F^{\left(q\right)}$.
This series depends on two parameters $\epsilon$ and $\hbar$. The coefficients of $\epsilon^{2l}\left(-i\hbar\right)^{k}$ appear in line $l$ and column $k$.
The left column $k=0$ contains the coefficients of $F^{\left(q\right)}\vert_{\hbar=0}$
and one can recognize the coefficients of the Witten-Kontsevich series. In this column, the box of line $l$ corresponds to the genus $l$ intersection
numbers of $\psi$-classes on the moduli space of curves. Starting
from the second column, the boxes are divided by dashed lines into levels. This subdivision
corresponds to some vanishings of the coefficients of $F^{\left(q\right)}$
and will be explained in Section~\ref{subsec: Correlators of qWK}.
The reader familiar with intersection numbers on the moduli space of curves
will recognize typical genus $g$ intersection numbers in the boxes of the diagonal $l+k=g$ of the array.

The main result of the paper is a combinatorial interpretation of
the coefficients of $F^{\left(q\right)}\vert_{\epsilon=0}$, corresponding
to the top line $l=0$ of the table, in terms of Hurwitz numbers. More
precisely, consider the one-part double Hurwitz numbers studied in \cite{goulden2005towards} and subject to a conjectural ELSV formula. These numbers count the number of non-equivalent holomorphic maps from a genus $g$ Riemann surface to the sphere with a complete ramification over $0$, a prescribed ramification profile over $\infty$ and a fixed number of simple ramification elsewhere. Goulden, Jackson and Vakil showed that these numbers depends polynomialy in the orders of ramification over $\infty$. The coefficients of these polynomials are the coefficient  $F^{\left(q\right)}\vert_{\epsilon=0}$. This relation was completly unexpected since the coefficients
of $F^{\left(q\right)}\vert_{\epsilon=0}$ are built from a combination
of a certain type of intersection numbers on moduli spaces of curves
and their relation with Hurwitz theory is non trivial.%
{} Various results relating Hurwitz theory and intersection theory on
the moduli space are known see \cite{Ekedahl_2001,Okounkov_2006},
however this one appears in the quantum context and is completely new.

An interpretation of the rest of $F^{\left(q\right)}$ is still missing.
However, we point out a conjecture which may interest the
reader. Let $g,l,n \geq 0$ such that $g \geq l$ and let $\left(d_1,\dots ,d_n \right)$ be a list of nonnegative integers such that $\sum_{i=1}^{n}d_{i}=2g-3+n-l$. The coefficient of $\epsilon^{2l}\hbar^{g-l}$ in $\left.\frac{\partial^{n}F^{\left(q\right)}}{\partial t_{d_{1}}\cdots\partial t_{d_{n}}}\right|_{t_{*}=0}$ is the Hodge integral
\[
\int_{\overline{\mathcal{M}}_{g,n}}\lambda_{g}\lambda_{l}\psi_{1}^{d_{1}}\cdots\psi_{n}^{d_{n}},
\]
where we denoted by $\overline{\mathcal{M}}_{g,n}$
the Deligne-Mumford compactification of the moduli space of curves
and $\lambda_{j}$ is the $j$-th Chern class of the Hodge bundle,
complete definitions will be given below. The form of this expression suggests a hidden localization formula.

We also prove a deformed version of the string equation
\[
\frac{\partial}{\partial t_{0}}F^{\left(q\right)}=\sum_{i\geq0}t_{i+1}\frac{\partial}{\partial t_{i}}F^{\left(q\right)}+\frac{t_{0}^{2}}{2}-\frac{i\hbar}{24}
\]
and conjecture a dilaton equation
\[
\frac{\partial}{\partial t_{1}}F^{\left(q\right)}=\sum_{i\geq0}t_{i}\frac{\partial}{\partial t_{i}}F^{\left(q\right)}+\epsilon\frac{\partial}{\partial\epsilon}F^{\left(q\right)}+2\hbar\frac{\partial}{\partial\hbar}F^{\left(q\right)}-2F^{\left(q\right)}+\frac{\epsilon^2}{24}.
\]
Generally, when the coefficients of the power series have a simple
geometrical interpretation, the string and dilaton equations follow
from elementary geometrical properties of the $\psi$-classes. However, the definition of the quantum Witten-Kontsevich series $F^{(q)}$ is non geometric in nature and quite elaborate. This makes the proof of the string relation much more complicated. And we have not been able to prove the dilaton relation so far. We can nevertheless remark that
the dilaton equation is proved for $F^{\left(q\right)}\vert_{\epsilon=0}$
as a consequence of our main theorem since it is proved in \cite{goulden2005towards}
for the coefficients of one-part double Hurwitz numbers.

This work suggests%
{} that quantum tau functions introduced by Buryak, Dubrovin, Gu\'er\'e
and Rossi \cite{Buryak_2019} have an interesting 
 interpretation in terms of enumerative geometry. The interpratation
of the rest of $F^{\left(q\right)}$ and other quantum tau functions
will be the object of future research. Deformed version Virasoro constraints
for $F^{\left(q\right)}$ will also be the subject of further studies.

\subsection{The quantum Witten-Kontsevich series}

We introduce a quantum extension of the Witten-Kontsevich power series.
We use the quantum deformation of the Korteweg\textendash de Vries
(KdV) hierarchy constructed by Buryak and Rossi~\cite{BuryakRossi2016}.
Based on the construction of this quantum integrable hierarchy, the
quantum Witten-Kontsevich series was introduced in~\cite{Buryak_2019}. 

\subsubsection{The classical Witten-Kontsevich series\label{subsec:The-classical-Witten-Kontsevich}}

Let $\overline{\mathcal{M}}_{g,n}$ be the moduli space of stable
curves of genus $g$ with $n$ marked points. Let $\pi:\overline{\mathcal{C}}_{g,n}\rightarrow\overline{\mathcal{M}}_{g,n}$
be the universal curve. Denote by $\omega_{rel}$ the relative cotangent
line bundle over $\overline{\mathcal{C}}_{g,n}$ and let $\psi_{i}=c_{1}\left(\sigma_{i}^{*}\left(\omega_{rel}\right)\right)$,
where $\sigma_{i}:\overline{\mathcal{M}}_{g,n}\rightarrow\overline{\mathcal{C}}_{g,n}$
is the $i$-th section of the universal curve. We also define the
classes $\lambda_{i}=c_{i}\left(\pi_{*}\omega_{rel}\right)$ which
will appear in Section~\ref{subsec:Quantum-hamiltonian-densities}.
The Witten-Kontsevich series is

\[
F\left(\epsilon,t_{0},t_{1},\ldots\right)=\sum_{\underset{2g-2+n>0}{g,n\geq0}}\frac{\epsilon^{2g}}{n!}\sum_{d_{1},\ldots,d_{n}\geq0}\langle\tau_{d_{1}}\ldots\tau_{d_{n}}\rangle_{g}\,t_{d_{1}}\cdots t_{d_{n}},
\]
where $\langle\tau_{d_{1}}\ldots\tau_{d_{n}}\rangle_{g}=\int_{\overline{\mathcal{M}}_{g,n}}\psi_{1}^{d_{1}}\cdots\psi_{n}^{d_{n}}$.
Note that specifying the genus in the notation $\langle\tau_{d_{1}}\ldots\tau_{d_{n}}\rangle_{g}$
is redundant since this number is non-zero only if $\sum d_{i}=3g-3+n$.
We use this notation in view of its quantum generalization.

An alternative definition of $F$ is given by the famous Witten-Kontsevich
theorem~\cite{witten1990two,kontsevich1992intersection} : $F$ is
the logarithm of the tau function of the KdV hierarchy associated
to the solution $u\left(x,t_{0},t_{1},\ldots\right)$ with the initial
condition $u\left(x,0,0,\dots\right)=x$. This particular solution
of the KdV hierarchy is called the \emph{string solution}. The definition
of the quantum Witten-Kontsevich series is a generalization of this
point of view.

\subsubsection{A formal Poisson structure of the KdV hierarchy\label{subsec: A formal Poisson structure of the KdV hierarchy}}

We define an algebra of power series and a Poisson structure on it,
we use them to describe each equation of the KdV hierarchy as a Hamiltonian equation. To motivate these definitions, we introduce the infinite
dimensional space of periodic functions $P:=\left\{ u:S^{1}\rightarrow\mathbb{\mathbb{C}}\right\} $.
Suppose these periodic functions have a Fourier transform $u\left(x\right)=\sum_{a\in\mathbb{Z}}p_{a}e^{iax}$,
this gives a system of coordinates $\left\{ p_{a},a\in\mathbb{Z}\right\} $
on $P$. We define an algebra of power series in the indeterminates
$p_{a},\,a\in\mathbb{Z},$ and interpret it as the algebra of functions
on $P$.
\begin{defn}
Let $\mathcal{F}\left(P\right)$ be the algebra $\mathbb{C}\left[p_{>0}\right]\left[\left[p_{\leq0},\epsilon\right]\right]$,
where the indeterminates $p_{>0}$ (resp. $p_{\leq0}$) stands for
$p_{a}$, with $a\in\mathbb{Z}_{>0}$ (resp. $a\in\mathbb{Z}_{\leq0}$).
\end{defn}
\begin{defn}
The Poisson structure on $\mathcal{F}\left(P\right)$ is given by
\[
\left\{ p_{a},p_{b}\right\} =ia\delta_{a+b,0},
\]
and we extend it to $\mathcal{F}\left(P\right)$ by the Leibniz rule.
\end{defn}
The Hamiltonians of the KdV hierarchy are elements of $\mathcal{F}\left(P\right)$,
they will be introduced in Section~\ref{subsec:Quantum-hamiltonian-densities}
as the classical limit $\hbar=0$ of their quantum counterparts. We
denote by $\overline{h}_{d}$ with $d\geq-1$ these classical Hamiltonians.
The $d$-th equation of the KdV hierarchy is then given by
\[
\frac{\partial u}{\partial t_{d}}=\left\{ u,\overline{h}_{d}\right\} .
\]
The whole system of equations forms the KdV hierarchy.
\begin{example}
From the Hamiltonian $\overline{h}_{1}=\frac{1}{3!}\sum_{a+b+c=0}p_{a}p_{b}p_{c}+\frac{\epsilon^{2}}{24}\sum_{a\in\mathbb{Z}}\left(ia\right)^{2}p_{a}p_{-a}$,
we recover the KdV equation which is the first equation of the KdV
hierarchy
\[
\frac{\partial u}{\partial t_{1}}=\sum_{a\in\mathbb{Z}}\left\{ p_{a},\overline{h}_{1}\right\} e^{iax}=u\partial_{x}u+\frac{\epsilon^{2}}{12}\partial_{x}^{3}u.
\]
\end{example}

In the next sections, we construct the quantum KdV hierarchy. It is
a quantum deformation of the KdV hierarchy. In Section~\ref{subsec:Star-product-and},
we deform the product of $\mathcal{F}\left(P\right)$ in the direction
given by the Poisson bracket to get a noncommutative star product. The space
of functions endowed with this star product will be denoted by $\mathcal{F}^{\hbar}\left(P\right)$.
Then, in Section~\ref{subsec:The-double-ramification} we introduce
the so-called double ramification cycle in $\overline{\mathcal{M}}_{g,n}$
and use intersection numbers with this cycle to introduce the quantum
Hamiltonians of the quantum KdV hierarchy. These quantum Hamiltonians
are elements of $\mathcal{F}^{\hbar}\left(P\right)$. They commute
with respect to the star product as mentioned in Section~\ref{subsec:Quantum-integrability-and}.
Hence, they form a quantum integrable hierarchy. We finally present
the quantum KdV equations in Section~\ref{subsec:The-quantum-KdV}.

Once the quantum KdV equations introduced, we are able to define the
quantum Witten-Kontsevich series, this is done in Section \ref{subsec: Construction of the quantum Witten-Kontsevich series}.

\subsubsection{The star product\label{subsec:Star-product-and}}

In the quantization deformation setting, we enlarge the space of functions
to $\mathcal{F}\left(P\right)\left[\left[\hbar\right]\right]$ and
endow it with a new product, the star product.
\begin{defn}
Let $W$ be the free algebra generated by the $p_{a}$ modulo the commutations
relations $\left[p_{a},p_{b}\right]=i\hbar a\delta_{a+b,0}$. The
\emph{normal ordering} of $f\in\mathcal{F}\left(P\right)\left[\left[\hbar\right]\right]$
is the element of $W\left[\left[\epsilon,\hbar\right]\right]$ obtained
by first sorting each monomial of $f$ with the $p_{<0}$ on the left
and then replacing the product of the $p_{a}$ by the non-commutative
product of $W\left[\left[\epsilon,\hbar\right]\right]$. We denote
by $:f:$ the normal ordering of $f$.
\end{defn}
\begin{defn}
\label{Definition star}
Let $f,g\in \mathcal{F}\left(P\right)\left[\left[\hbar\right]\right]$.
The \emph{star product} $f\star g$ is an element of $\mathcal{F}\left(P\right)\left[\left[\hbar\right]\right]$
defined in the following way. Organize the product $:f::g:$ in $W\left[\left[\epsilon,\hbar\right]\right]$
with the $p_{<0}$ on the left using the commutation relation $\left[p_{a},p_{b}\right]=i\hbar a\delta_{a+b,0}$
of $W\left[\left[\epsilon,\hbar\right]\right]$. This process is well
defined according to the polynomiality in the $p_{>0}$ of $f$. This
organization of $:f::g:$ is the normal ordering of a unique element
of $\mathcal{F}\left(P\right)\left[\left[\hbar\right]\right]$. This
element is the star product $f\star g$.

We denote by $\mathcal{F}^{\hbar}\left(P\right)$ the deformed algebra
obtained by endowing $\mathcal{F}\left(P\right)\left[\left[\hbar\right]\right]$
with the star product.
\end{defn}

\begin{rem}
Let $f,g\in\mathcal{F}^{\hbar}\left(P\right)$. The star product is
a quantum deformation of the usual product on $\mathcal{F}\left(P\right)$
in the direction given by the Poisson bracket, that is
\begin{align*}
f\star g & =fg+O\left(\hbar\right)
\end{align*}
and 
\[
\left[f,g\right]=\hbar\left\{ f,g\right\} +O\left(\hbar^{2}\right).
\]
In particular, when we substitute $\hbar=0$ in $\mathcal{F}^{\hbar}\left(P\right)$
we obtain $\mathcal{F}\left(P\right)$.
\end{rem}
Certain special elements of $\mathcal{F}^{\hbar}\left(P\right)$ will
be of particular interest to us. We define them now.

\subsubsection{Differential polynomials}

We give two equivalent definitions of differential polynomials and
explain how to identify them. Differential polynomials appear in the
construction of the quantum Witten-Kontsevich series, in particular
this identification will be necessary. 

\begin{notation}
\label{Notation. u_i}From now on, we denote by $u_{i}$ 
the $i$-th derivative of the function $u:S^{1}\rightarrow\mathbb{C}$, i.e. $u_{s}=\partial_{x}^{s}u$ with $s\geq0$.
\end{notation}

\begin{defn}
A\emph{ differential polynomial} is an element of $\mathcal{A}:=\mathbb{C}\left[u_{0},u_{1},\dots\right]\left[\left[\epsilon,\hbar\right]\right]$.
\end{defn}

\begin{defn}
Let $d$ be a positive integer. Let $\left(\phi_{0},\dots,\phi_{d}\right)$
be a list where\\
$\phi_{k}\left(a_{1},\dots,a_{k}\right)\in\mathbb{C}\left[a_{1},\dots,a_{k}\right]\left[\left[\epsilon,\hbar\right]\right]$
is a symmetric polynomial in its $k$ indeterminates $a_{1},\dots,a_{k}$
for $0\leq k\leq d$. The \emph{formal Fourier series associated to
}$\left(\phi_{0},\dots,\phi_{d}\right)$ is
\[
\phi\left(x\right)=\sum_{A\in\mathbb{Z}}\left(\sum_{k=0}^{d}\sum_{\underset{\sum a_{i}=A}{a_{1},...,a_{k}\in\mathbb{Z}}}\phi_{k}\left(a_{1},\dots,a_{k}\right)p_{a_{1}}\cdots p_{a_{k}}\right)e^{ixA}\in\mathcal{F}^{\hbar}\left(P\right)\left[\left[e^{-ix},e^{ix}\right]\right].
\]
The set of formal Fourier series associated to any $d\in\mathbb{N}$
and any $\left(\phi_{0},\dots,\phi_{d}\right)$ is an algebra that
we denote by $\tilde{\mathcal{A}}$.
\end{defn}
\begin{lem}
The algebras $\mathcal{A}$ and $\tilde{\mathcal{A}}$ are isomorphic.
\end{lem}
Indeed, by substituting the formal Fourier series $u_{s}\left(x\right)=\sum_{a\in\mathbb{Z}}\left(ia\right)^{s}p_{a}e^{iax},\text{ with }s\in\mathbb{N},$
of $u$ and its derivative in a differential polynomial, we obtain
an element of $\tilde{\mathcal{A}}$. By this application, the differential
monomial $u_{s_{1}}\cdots u_{s_{n}}$ yields the formal Fourier series
associated to $\phi_{n}\left(a_{1},\dots,a_{n}\right)=\frac{1}{n!}\sum_{\sigma\in S_{n}}a_{\sigma\left(1\right)}^{s_{1}}\cdots a_{\sigma\left(n\right)}^{s_{n}}$
and $\phi_{i}=0$ if $i\neq n$.

The elements of $\mathcal{A}$ and $\tilde{\mathcal{A}}$ will be called
differential polynomials. However, we will keep the notations $\mathcal{A}$
and $\tilde{\mathcal{A}}$ in order to indicate our point of view. 
\begin{rem}
The Fourier mode of frequency $A\in\mathbb{Z}$ in a differential
polynomial is an element of $\mathcal{F}^{\hbar}\left(P\right)$.
In this text, we will only use the elements of $\mathcal{F}^{\hbar}\left(P\right)$
obtained in this way.
\end{rem}
\begin{defn}
The derivative $\partial_{x}$ of a differential polynomial $\phi$
is the differential polynomial obtained by multiplying the $A$-th
mode of $\phi$ by $A$. 

The integration along $S^{1}$ of a differential polynomial $\phi$
is the $0$-th mode of $\phi$. We denote by $\int_{S^{1}}\phi\left(x\right)dx=\overline{\phi}$
this integral.

The \emph{primitive of a differential polynomial} $\phi$ is a differential
polynomial $\psi$ such that $\partial_{x}\psi=\phi$.
\end{defn}
\begin{prop}
\label{prop: Stability commutator}Let $\phi$ and $\psi$ be two
differential polynomials. The commutator $\left[\phi,\overline{\psi}\right]$
is a differential polynomial.
\end{prop}
A proof can be found in~\cite{BuryakRossi2016} where the authors
give an expression for this commutator in terms of $u$ and its derivative.
It is clear from their expression that it is a differential polynomial. 

\subsubsection{The double ramification cycle\label{subsec:The-double-ramification}}

Fix a list of $n$ integers $A=\left(a_{1},\dots,a_{n}\right)$ which
add up to zero. To this list we associate a space of rubber stable
maps to $\mathbb{P}^{1}$ relative to $0$ and $\infty$ in the following
way.

Let $n_{+}$ be the number of positive $a_{i}$'s in $A$. Let $\mu=\left(\mu_{1},\dots,\mu_{n_{+}}\right)$
be the partition made from these positive $a_{i}$'s. Similarly, let
$n_{-}$ be the number of negative $a_{i}$'s, after changing their
signs we make the partition $\nu=\left(\nu_{1},\dots,\nu_{n_{-}}\right)$.
Note that $\mu$ and $\nu$ are two partitions of $\frac{1}{2}\sum_{i=1}^{n}\left|a_{i}\right|$. Let $n_0$ be the number of vanishing $a_i$'s. 
We denote by 

\[
\overline{\mathcal{M}}_{g}\left(a_{1},\dots,a_{n}\right):=\overline{\mathcal{M}}_{g,n_0}^{\sim}\left(\mathbb{P}^{1},\mu,\nu\right)
\]
the moduli space of rubber stable maps to $\mathbb{P}^{1}$ relative
to $0$ and $\infty$ with profile given by $\mu$ and $\nu$, where the genus $g$ source curve has $n_0$ additional marked points with image in $\mathbb{P}^{1}\backslash\left\{ 0,\infty\right\} $. Rubber
means that two relative stable maps are identified is they differ
by a $\mathbb{C}^{*}$ action in the target $\mathbb{P}^{1}$ (see
e.g. \cite{faber2005relative} where these maps are defined as unparametrized
relative stable maps). This space is endowed with the map $st:\overline{\mathcal{M}}_{g}\left(a_{1},\dots,a_{n}\right)\rightarrow\overline{\mathcal{M}}_{g,n}$
that forgets everything except the source curve and stabilizes it.
Moreover, $\overline{\mathcal{M}}_{g}\left(a_{1},\dots,a_{n}\right)$
has a virtual fundamental class of virtual dimension $2g-3+n$. 
\begin{defn}
The \emph{double ramification cycle} ${\rm DR}_{g}\left(a_{1},\dots,a_{n}\right)$
is defined by 
\[
{\rm DR}_{g}\left(a_{1},\dots,a_{n}\right):=st_{*}\left(\left[\overline{\mathcal{M}}_{g,n}\left(a_{1},\dots,a_{n}\right)\right]^{virt}\right)\in H_{2(2g-3+n)}\left(\overline{\mathcal{M}}_{g,n}\right).
\]
\end{defn}

It is conjectured that the double ramification cycle is a polynomial
in the $a_{i}$'s and a proof was announced by Pixton and Zagier.
However, we will only need that the intersection number of the double
ramification cycle with a tautological class (see e.g. \cite{FaberPandha15_tautologica}
for a definition) is a polynomial in the $a_{i}$'s.
\begin{prop}
[\cite{BuryakRossi2016}, Proposition B.1]\label{prop: DR polynomial}Let
$\alpha$ be a tautological class in $H^{*}\left(\overline{\mathcal{M}}_{g,n}\right)$.
There exists a polynomial $P_{g,n}\left(x_{1},\dots,x_{n}\right)\in\mathbb{C}\left[x_{1},\dots,x_{n}\right]$
such that 
\[
\int_{{\rm DR}_{g}\left(a_{1},\dots,a_{n}\right)}\alpha=P_{g,n}\left(a_{1},\dots,a_{n}\right)
\]
for all $a_{1},\dots,a_{n}$ such that $\sum_{i=1}^{n}a_{i}=0$. Moreover
$P_{g,n}$ is even and of degree $2g$. %
\end{prop}

\subsubsection{Quantum Hamiltonian densities and quantum Hamiltonians\label{subsec:Quantum-hamiltonian-densities}}

\begin{defn}
Fix $d\geq-1$. The \emph{quantum Hamiltonians density $H_{d}$ of
the quantum KdV hierarchy }is
\begin{equation}
H_{d}\left(x\right)=\sum_{\overset{g\geq0,m\geq0}{2g+m>0}}\frac{\left(i\hbar\right)^{g}}{m!}\sum_{a_{1},\dots,a_{m}\in\mathbb{Z}}\left(\int_{{\rm DR}_{g}\left(0,a_{1},\dots,a_{m},-\sum a_{i}\right)}\psi_{1}^{d+1}\Lambda\left(\frac{-\epsilon^{2}}{i\hbar}\right)\right)p_{a_{1}}\cdots p_{a_{m}}e^{ix\sum a_{i}}\in\tilde{\mathcal{A}},\label{eq: Def Hamiltoniens H_d}
\end{equation}
 where $\Lambda\left(\frac{-\epsilon^{2}}{i\hbar}\right):=1+\left(\frac{-\epsilon^{2}}{i\hbar}\right)\lambda_{1}+\dots+\left(\frac{-\epsilon^{2}}{i\hbar}\right)^{g}\lambda_{g}$.\end{defn}
The Hamiltonian density $H_{d}\left(x\right)$ is an element of $\tilde{\mathcal{A}}$. Indeed, the degree of the class ${\rm DR}_{g}\left(0,a_{1},\dots,a_{m},-\sum a_{i}\right)$ is $2\left(2g-1+n\right)$, hence the summation over $g$ and $n$ is finite. Moreover the $\psi$- and $\lambda$- classes are tautologicals \cite{FaberPandha15_tautologica}, then Proposition \ref{prop: DR polynomial} implies that the integral is a polynomial in the $a_i$'s.
\begin{defn}
The \emph{quantum Hamiltonians} of the quantum KdV hierarchy are obtained
by the $x$-integration of the Hamiltonian densities:
\[
\overline{H}_{d}=\int_{S^{1}}H_{d}\left(x\right)dx\in\mathcal{F}^{\hbar}\left(P\right),
\]
for $d\geq-1$.
\end{defn}
\begin{rem}
When we substitute $\hbar=0$, we obtain the classical Hamiltonians
densities $h_{d}\left(x\right):=H_{d}\left(x\right)\big\vert_{\hbar=0}$
and the classical Hamiltonians $\overline{h}_{d}:=\overline{H}_{d}\big\vert_{\hbar=0}$
of the KdV hierarchy, see \cite{buryak2015double} for a proof.
\end{rem}

\subsubsection{Integrability and tau symmetry\label{subsec:Quantum-integrability-and}}

Two properties are needed for the construction of the quantum Witten-Kontsevich
series: the commutativity of the Hamiltonians and the tau-symmetry.
These two properties are proved in~\cite{Buryak_2019},
however the authors defined the quantum Hamiltonian densities and
quantum Hamiltonians in a slightly different way from ours. We explain
the equivalence between these definitions in Appendix~\ref{Appendix A}.
\begin{prop}
[Quantum integrability]\label{prop: quantum integrability}We have
\[
\left[\overline{H}_{d_{1}},\overline{H}_{d_{2}}\right]=0,\text{ for }d_{1},d_{2}\geq-1.
\]
\end{prop}
We say that the quantum hierarchy is \emph{integrable}.
\begin{rem}
The substitution $\hbar=0$ in $\frac{1}{\hbar}\left[\overline{H}_{d_{1}},\overline{H}_{d_{2}}\right]=0$
gives the integrability condition of the classical KdV hierarchy.
\end{rem}
\begin{prop}
[Tau symmetry]\label{prop: tau symmetry}We have
\[
\left[H_{d_{1}-1}\left(x\right),\overline{H}_{d_{2}}\right]=\left[H_{d_{2}-1}\left(x\right),\overline{H}_{d_{1}}\right],\text{ for }d_{1},d_{2}\geq-1.
\]
\end{prop}
\begin{rem}
\label{rem: tau symmetry KdV}The substitution $\hbar=0$ in $\frac{1}{\hbar}\left[H_{d_{1}-1}\left(x\right),\overline{H}_{d_{2}}\right]=\frac{1}{\hbar}\left[H_{d_{2}-1}\left(x\right),\overline{H}_{d_{1}}\right]$
gives the tau symmetry of the KdV hierarchy.
\end{rem}

\subsubsection{The quantum KdV equations \label{subsec:The-quantum-KdV}}
\begin{defn}
The time-dependent function $f^{\bm{t}}\in\mathcal{F}^{\hbar}\left(P\right)\left[\left[t_{0},t_{1},\dots\right]\right]$
is a \emph{solution of the quantum KdV hierarchy} with initial condition
$f\in\mathcal{F}^{\hbar}\left(P\right)$ if $f^{\bm{t}}\big\vert_{\bm{t}=0}=f$
and
\[
\frac{df^{\bm{t}}}{dt_{d}}=\frac{1}{\hbar}\left[f^{t},\overline{H}_{d}\right]\text{ for }d\geq0.
\]
\end{defn}
\begin{prop}
A solution of the quantum KdV hierarchy with initial condition $f\in\mathcal{F}^{\hbar}\left(P\right)$
is given by
\[
f^{\bm{t}}=\exp\left(\sum_{k\geq0}\frac{t_{k}}{\hbar}\left[\cdot,\overline{H}_{k}\right]\right)f.
\]
\end{prop}
\begin{rem}
\label{rem: ft(u)=00003Df(ut)}A classical Hamiltonian flow can be
viewed in two different ways: $\left(i\right)$ a flow $\bm{t}\rightarrow u^{\bm{t}}$
interpreted as a flow on the phase space $P$; $\left(ii\right)$
a flow $\bm{t}\rightarrow f^{\bm{t}}$ on the space $\mathcal{F}\left(P\right)$
of functions on $P$, satisfying $f^{\bm{t}}\left(u\right)=f\left(u^{\bm{t}}\right).$
The standard way of writing the classical KdV hierarchy uses the first
point of view, while the classical limit ($\hbar=0$) of the quantum
KdV hierarchy uses the second. Because of this, the substitution $\hbar=0$
in the quantum KdV equations does not directly yield the standard
presentation of the KdV equations. To obtain the classical KdV equations,
we substitute $\hbar=0$ in the quantum KdV equations to $f=p_{a}$,
getting $\frac{\partial p_{a}^{\bm{t}}}{\partial t_{d}}=\left\{ p_{a}^{\bm{t}},\overline{h}_{d}\right\} $.
A solution of the KdV then has the form $u\left(x,\bm{t}\right)=\sum_{a\in\mathbb{Z}}p_{a}^{\bm{t}}e^{iax}$.

We emphasize that in the quantum setting we have $f^{\bm{t}}\left(u\right)\neq f\left(u^{\bm{t}}\right)$,
the two points of view $\left(i\right)$ and $\left(ii\right)$ are
no more equivalent, a solution is a trajectory on $\mathcal{F}^{\hbar}\left(P\right)$.
\end{rem}
Given a differential polynomial $\phi\in\tilde{\mathcal{A}}$, we
denote by $\phi^{\bm{t}}:=\exp\left(\sum_{k\geq0}\frac{t_{k}}{\hbar}\left[\cdot,\overline{H}_{k}\right]\right)\phi$.
In this notation, the operator $\exp\left(\sum_{k\geq0}\frac{t_{k}}{\hbar}\left[\cdot,\overline{H}_{k}\right]\right)$
acts on each mode of $\phi$ so that each mode of $\phi^{\bm{t}}$
is a solution of the quantum KdV hierarchy.
\begin{prop}
\label{rem: evolved poly diff} The time-dependent differential polynomial
$\phi^{\bm{t}}$ is an element of $\tilde{\mathcal{A}}\left[\left[t_{0},t_{1},\dots\right]\right]$.
\end{prop}
This follows from Proposition~\ref{prop: Stability commutator}.

\subsubsection{The quantum Witten-Kontsevich series \label{subsec: Construction of the quantum Witten-Kontsevich series}}

\paragraph{Overview.}

In the classical settings, Dubrovin and Zhang~\cite{dubrovin2001normal,buryak2018tau}
constructed tau functions from tau symmetric Hamiltonian hierarchies.
In this construction, one associates a tau function to any solution
of the hierarchy. The KdV hierarchy is a tau symmetric Hamiltonian
hierarchy (see Remark~\ref{rem: tau symmetry KdV}). Thus, one way
to build the Witten-Kontsevich series is to construct the tau functions
of KdV following Dubrovin and Zhang and then taking the logarithm
of the tau function associated to the string solution, that is the
solution starting at $u\left(x\right)=x$. 

The construction of Dubrovin and Zhang was generalized in~\cite{Buryak_2019}.
In this work, the authors defined quantum tau functions from tau symmetric
quantum Hamiltonian hierarchies. In this context, a quantum tau function
is associated to a point of the phase space $P$. Of course, the restriction
$\hbar=0$ of a quantum tau function is a classical tau function of
the related classical hierarchy. The point of $P$ associated to this
quantum tau function is the starting point of the classical solution
associated to the classical tau function. In particular, the quantum
Witten-Kontsevich series it the logarithm of the quantum tau function
of the quantum KdV hierarchy associated to the point $u\left(x\right)=x$
of $P$.

\paragraph{Construction of quantum tau functions.}

We now give the construction of quantum tau functions of the quantum
KdV hierarchy and their logarithms. We follow the construction of~\cite{Buryak_2019}.
This construction is based on two properties of the quantum KdV hierarchy: its integrability and the tau symmetry. 

Let $d_{1},d_{2}$ be two positive integers. Start from the differential
polynomial $\left[H_{d_{1}},\overline{H}_{d_{2}}\right]$. Due to
the commutativity of the Hamiltonians $\int\left[H_{d_{1}}\left(x\right),\overline{H}_{d_{2}}\right]dx=0$,
there exists primitives. Define the \emph{two-point function} $\Omega_{d_{1},d_{2}}^{\hbar}\in\tilde{\mathcal{A}}$
by
\[
\partial_{x}\Omega_{d_{1},d_{2}}^{\hbar}:=\frac{1}{\hbar}\left[H_{d_{1}-1},\overline{H}_{d_{2}}\right],
\]
where we fix the constant using the recursive formula $\left.\frac{\partial\Omega_{d_{1},d_{2}}^{\hbar}}{\partial p_{0}}\right|_{p_{*}=0}=\Omega_{d_{1}-1,d_{2}}^{\hbar}\Bigr|_{p_{*}=0}+\Omega_{d_{1},d_{2}-1}^{\hbar}\Bigr|_{p_{*}=0}$
with the initial conditions $\left.\Omega_{0,d}^{\hbar}\right|_{p_{*}=0}=\left.\Omega_{d,0}^{\hbar}\right|_{p_{*}=0}=\left.H_{d-1}\right|_{p_{*}=0}$,
where $d\geq0$. This convention differs from the one used in~\cite{Buryak_2019}.
We made this choice so that the quantum Witten-Kontsevich series satisfies
the string equation.

The time-dependent differential polynomial $\Omega_{d_{1},d_{2}}^{\hbar,\bm{t}}$
with initial condition $\Omega_{d_{1},d_{2}}^{\hbar}$ is an element
of \\
$\tilde{\mathcal{A}}\left[\left[t_{0},t_{1},\dots\right]\right]$
according to Proposition~\ref{rem: evolved poly diff}. By the commutativity
of the Hamiltonians and the tau-symmetry, $\Omega_{d_{1},d_{2}}^{\hbar,\bm{t}}$
and $\frac{\partial}{\partial t_{d_{3}}}\Omega_{d_{1},d_{2}}^{\hbar,\bm{t}}$
are invariants under the permutations of their indices $d_{1},d_{2},d_{3}$.
Thus, using twice the Poincar\'e lemma, we conclude that there exists
a power series $\mathcal{F}\in\tilde{\mathcal{A}}\left[\left[t_{0},t_{1},\dots\right]\right]$
such that
\[
\frac{\partial^{2}\mathcal{F}}{\partial t_{d_{1}}\partial t_{d_{2}}}=\Omega_{d_{1},d_{2}}^{\hbar,\bm{t}}.
\]

\begin{defn}
The \emph{logarithm of a} \emph{quantum tau function} \emph{of the
quantum KdV hierarchy} is obtained by first evaluating $\mathcal{F}\in\mathcal{A}\left[\left[t_{0},t_{1},\dots\right]\right]$
at a point $u\in\mathbb{C}\left[\left[x,\epsilon,\hbar\right]\right]$
(interpreted as a point of $P$) and then at $x=0$. It is an element
of $\mathbb{C}\left[\left[\epsilon,\hbar,t_{0},t_{1},\dots\right]\right]$.
It is uniquely defined up to constant and linear terms.
\end{defn}
\begin{rem}
\label{rem: tau solution vs tau point}We can use the same construction
with $\hbar=0$ in order to define the logarithm of the classical
tau functions of the KdV hierarchy (up to constant and linear terms
in $t_{*}$). We then associate a classical tau function to any point
$u\in\mathbb{C}\left[\left[x,\epsilon\right]\right]$. However, the
evaluation $\Omega_{d_{1},d_{2}}^{\hbar=0,\bm{t}}\left(u\right)$
of the time dependant differential polynomial $\Omega_{d_{1},d_{2}}^{\hbar=0,\bm{t}}$
at $u\in\mathbb{C}\left[\left[x,\epsilon\right]\right]$ is equal
to $\Omega_{d_{1},d_{2}}^{\hbar=0,\bm{t}=0}\left(u^{\bm{t}}\right)$,
where $\bm{t}\rightarrow u^{\bm{t}}$ is the solution of the KdV hierarchy
starting at $u\in\mathbb{C}\left[\left[x,\epsilon\right]\right]$.
This equality is due to the equivalence of the points of views $\left(i\right)$
and $\left(ii\right)$ of Remark~\ref{rem: ft(u)=00003Df(ut)}. Thanks
to this change of perspective, the classical limit of this construction
associates a classical tau function to any solution of the hierarchy;
we recover the definition of classical tau functions of Dubrovin and
Zhang (see \cite{dubrovin2001normal,buryak2018tau}).%
\end{rem}
\begin{rem}
\label{rem: d_t_0 =00003D d_x}One may wonder why we forget about
the $x$ dependancy. Let $\phi^{\bm{t}}$ be the time-dependent differential
polynomial with initial condition the differential polynomial $\phi$.
One can verify using Lemma~\ref{lem: calcul H0} that the $0$-th
quantum KdV equation is
\[
\frac{\partial\phi^{\bm{t}}}{\partial t_{0}}=\frac{1}{\hbar}\left[\phi^{\bm{t}},\overline{H}_{0}\right]=\partial_{x}\phi^{\bm{t}},
\]
that is the evolutions along $x$ and $t_{0}$ are the same. Hence
we recover the $x$ dependancy in the quantum tau functions by setting
$t_{0}:=t_{0}+x$. 
\end{rem}

\paragraph{The quantum Witten-Kontsevich series.}
\begin{defn}
The \emph{quantum Witten-Kontsevich series} $F^{\left(q\right)}$
is obtained by first evaluating $\mathcal{F}$ at the point $u\left(x\right)=x$
of $P$ and then $x=0$. Moreover we impose that the coefficient of
$\epsilon^{2l}\hbar^{g-l}t_{d}$ is the coefficient of $\epsilon^{2l}\hbar^{g-l}t_{0}t_{d+1}$
for any $0\leq l\leq g$ and $d\geq0$. We also impose that the constant
coefficient of $\epsilon^{2l}\hbar^{g-l}$ is given by $\frac{1}{2g-2}$
times the coefficient of $\epsilon^{2l}\hbar^{g-l}t_{1}$ when $g\geq2$ and $0$ otherwise.

Let $k,g$ be two non negative integers, and a list $\left(d_{1},\dots,d_{n}\right)$
of non negative integers. The \emph{quantum correlators $\langle\tau_{d_{1}}\dots\tau_{d_{n}}\rangle_{l,g-l}$}
is the coefficient of the quantum Witten-Kontsevich series written
as
\[
F^{\left(q\right)}=\sum\frac{\langle\tau_{d_{1}}\dots\tau_{d_{n}}\rangle_{l,g-l}}{n!}\epsilon^{2l}\left(-i\hbar\right)^{g-l}t_{d_{1}}\cdots t_{d_{n}}\in\mathbb{C}\left[\left[\epsilon,\hbar,t_{0},t_{1},\dots\right]\right].
\]
\end{defn}
\begin{rem}
We made this particular choice of constant and linear terms so that
$F^{\left(q\right)}$ satisfies the string and dilaton equations,
see Section~\ref{subsec : String-and-dilaton}.
\end{rem}
\begin{rem}
\label{rem: evaluation u_i=00003Dd_i,1} The power series $\mathcal{F}$
is an element of $\mathcal{A}\left[\left[t_{0},t_{1},\dots\right]\right]$.
In order to define the quantum Witten-Kontsevich series, we have to
evaluate $\mathcal{F}$ at the point $u\left(x\right)=x$ and then
substitute $x=0$. This is equivalent to evaluate $\mathcal{F}$ at
$u_{0}=0$, $u_{1}=1$ and $u_{i}=0$, when $i\geq2$.
\end{rem}
\begin{prop}
We have
\[
F^{\left(q\right)}\Bigg\vert_{\hbar=0}=F.
\]
\end{prop}
Indeed, as mentioned in Remark~\ref{rem: tau solution vs tau point},
$F^{\left(q\right)}\vert_{\hbar=0}$ is the logarithm of the tau function
of the KdV hierarchy associated to the string solution (the solution
starting at $u\left(x\right)=x$) with our particular choice of constant
and linear terms in $t_{*}$. Thus, $F^{\left(q\right)}\vert_{\hbar=0}$
is the Witten-Kontsevich series $F$ up to constant and linear terms.
Moreover, $F$ satisfies the string and dilaton equations, hence its
constant and linear terms in $t_{*}$ are the same as those of $F^{\left(q\right)}\vert_{\hbar=0}$.

\subsection{One-part double Hurwitz numbers}

Fix three nonnegative integers $g,n,d$ and a partition $\mu=\left(\mu_{1},\dots,\mu_{n}\right)$
of $d$. Consider the degree $d$ branched covers of the Riemann sphere
$f:C\rightarrow\mathbb{P}^{1}$ by a Riemann surface $C$ of genus
$g$ such that
\begin{itemize}
\item $f$ is completely ramified over $0$,
\item the ramification profile of $f$ over $\infty$ is given by $\mu$,
\item all the ramifications of $f$ over $\mathbb{P}^{1}\backslash\left\{ 0,\infty\right\} $
are simple.
\end{itemize}
By the Riemann-Hurwitz formula, we count $r=2g-1+n$ simple ramifications
over $\mathbb{P}^{1}\backslash\left\{ 0,\infty\right\} $. There is
a finite number of isomorphism classes of such covers. 
\begin{defn}
The \emph{one-part double Hurwitz} \emph{number} is the weighted sum
\[
H_{\left(d\right),\mu}^{g}=\sum_{\left[f\right]}\frac{1}{\left|\mbox{Aut}f\right|}
\]
where the summation is over isomorphism classes of such covers.
\end{defn}
\begin{prop}
[Goulden-Jackson-Vakil \cite{goulden2005towards}]For a fixed genus
$g$, the Hurwitz number $H_{\left(d\right),\mu}^{g}$ depends polynomially
on $\mu_{1},\dots,\mu_{n}$; this polynomial is divisible by $d=\mu_{1}+\dots+\mu_{n}$. 
\end{prop}

\begin{defn}
Let $\left(d_{1},\dots,d_{n}\right)$ be a list of non negative integers.
A \emph{Hurwitz correlator} is the number
\begin{equation}
\left\langle \left\langle \tau_{d_{1}}\dots\tau_{d_{n}}\right\rangle \right\rangle _{g}=\left(-1\right)^{\frac{4g-3+n-\sum d_{i}}{2}}\left[\mu_{1}^{d_{1}}\cdots\mu_{n}^{d_{n}}\right]\left(\frac{H_{\left(d\right),\mu}^{g}}{r!d}\right),\label{eq: Hurwitz Correlators}
\end{equation}
where $\left[\mu_{1}^{d_{1}}\cdots\mu_{n}^{d_{n}}\right]$ denotes
the coefficient of $\mu_{1}^{d_{1}}\cdots\mu_{n}^{d_{n}}$. 
\end{defn}

\begin{prop}
[Goulden-Jackson-Vakil \cite{goulden2005towards}]\label{prop: Hurwitz correlators, level structure}The
Hurwitz correlators $\left\langle \left\langle \tau_{d_{1}}\dots\tau_{d_{n}}\right\rangle \right\rangle _{g}$
vanishes if $\sum d_{i}$ is outside the interval 
\[
\left[2g-3+n,4g-3+n\right]
\]
or if $\sum d_{i}$ has the parity of $n$.
\end{prop}
In~\cite{goulden2005towards}, Goulden, Jackson and Vakil prove various
properties of the Hurwitz correlators: they satisfy the string and
dilaton equations. The polynomial $H_{\left(d\right),\mu}^{g}$ is
actually divisible by $\left(\sum_{i}\mu_{i}\right)^{r-1}$. They
also conjecture an ELSV type formula for the Hurwitz correlators.

\subsection{Statement of the results}

\subsubsection{Correlators of the Witten-Kontsevich series\label{subsec: Correlators of qWK}}

The quantum Witten-Kontsevich series $F^{\left(q\right)}$ is an element
of $\mathbb{C}\left[\left[\epsilon,\hbar,t_{0},t_{1},\dots\right]\right]$.
Its classical part, obtained by plugging $\hbar=0$, is the Witten-Kontsevich
series $F$ of Section~\ref{subsec:The-classical-Witten-Kontsevich}. 

The following theorem concerns the restriction $\epsilon=0$ of the
quantum Witten-Kontsevich series. The coefficients of $F^{\left(q\right)}\big\vert_{\epsilon=0}$
are expressed in terms of one-part double Hurwitz numbers.
\begin{thm}
\label{thm: main thm}Fix two nonnegative integers $g,n$ and a list
of nonnegative integers $\left(d_{1},\dots,d_{n}\right)$. We have
\[
\langle\tau_{d_{1}}\dots\tau_{d_{n}}\rangle_{0,g}=\left\langle \left\langle \tau_{d_{1}}\dots\tau_{d_{n}}\right\rangle \right\rangle _{g}.
\]
\end{thm}
Thus we have a geometric interpretation for the coefficients of $\epsilon^{0}\hbar^{g}$
and $\epsilon^{2g}\hbar^{0}$ of the quantum Witten-Kontsevich series.
So far there is no such interpretation for the other coefficients,
but we have a conjecture for some of them. Let us first explain some
vanishing properties of the correlators.

\paragraph{Level structure. }

The correlators satisfy some vanishing properties similar to the Hurwitz
correlators (see Proposition~\ref{prop: Hurwitz correlators, level structure}). 
\begin{prop}
\label{prop: top vanishing correlators }Fix three nonnegative integers
$g,n,l$ such that $l\leq g$. The correlator $\langle\tau_{d_{1}}\dots\tau_{d_{n}}\rangle_{l,g-l}$
vanishes if
\[
\sum_{i=1}^{n}d_{i}>4g-3+n-l\,\,\,\,\text{ or if }\,\,\,\,\sum d_{i}\equiv n-l\left(\text{mod 2}\right),
\]
where $\left(d_{1},\dots,d_{n}\right)$ is a list of nonnegative integers. 
\end{prop}
\begin{conjecture}
\label{conj: minimal vanishing correlators}Fix three nonnegative
integers $g,n,l$ such that $l\leq g$. The correlator $\langle\tau_{d_{1}}\dots\tau_{d_{n}}\rangle_{l,g-l}$
vanishes if
\[
\sum_{i=1}^{n}d_{i}<2g-3+n-l,
\]
where $\left(d_{1},\dots,d_{n}\right)$ is a list of nonnegative integers.
\end{conjecture}
Hence, the correlators are possibly nonzero only when $\sum d_{i}$
takes the $g+1$ values of the interval $\left[2g-3+n-l,4g-3+n-l\right]$
with the parity of its maximum (or minimum). We say that the correlators
$\langle\tau_{d_{1}}\dots\tau_{d_{n}}\rangle_{l,g-l}$ are structured
in $g+1$ levels.

\paragraph{Minimal level.}

We have the following geometrical interpretation for the correlators
in the minimal levels.
\begin{conjecture}
\label{conj: expression minimal correlator}Fix three nonnegative
integers $g,n,l$ such that $l\leq g$. When $\sum d_{i}=2g-3+n-l$,
the correlators are given by
\[
\langle\tau_{d_{1}}\dots\tau_{d_{n}}\rangle_{l,g-l}=\int_{\overline{\mathcal{M}}_{g,n}}\lambda_{g}\lambda_{l}\psi_{1}^{d_{1}}\cdots\psi_{n}^{d_{n}}.
\]
\end{conjecture}
\begin{rem}
\label{rem: thm check conjecture}Let us check the level structure
and the minimal level property when $\epsilon=0$. In this case, the
correlators $\langle\tau_{d_{1}}\dots\tau_{d_{n}}\rangle_{0,g},\,\text{with }g\geq0$
are equal to $\left\langle \left\langle \tau_{d_{1}}\dots\tau_{d_{n}}\right\rangle \right\rangle _{g}$
according to the main theorem. The level structure of the correlators
follows from the similar level structure described in Proposition~\ref{prop: Hurwitz correlators, level structure}.
The minimal level property in this case reads
\[
\langle\tau_{d_{1}}\dots\tau_{d_{n}}\rangle_{0,g}=\int_{\overline{\mathcal{M}}_{g,n}}\lambda_{g}\lambda_{0}\psi_{1}^{d_{1}}\cdots\psi_{n}^{d_{n}}=\int_{\overline{\mathcal{M}}_{g,n}}\lambda_{g}\psi_{1}^{d_{1}}\cdots\psi_{n}^{d_{n}},
\]
which follows from the analogous equality for Hurwitz correlators
\cite[Proposition 3.12]{goulden2005towards}. 
\end{rem}
\begin{rem}
\label{rem: one level}When $l=g$, the correlators $\langle\tau_{d_{1}}\dots\tau_{d_{n}}\rangle_{g,0}$
are the coefficients of the classical Witten-Kontsevich series. They
are given by
\[
\langle\tau_{d_{1}}\dots\tau_{d_{n}}\rangle_{g,0}=\int_{\overline{\mathcal{M}}_{g,n}}\psi_{1}^{d_{1}}\cdots\psi_{n}^{d_{n}}.
\]
They correspond to the top level of the level structure : $\sum d_{i}=3g-3+n$.
All the other levels vanish. In particular, the bottom level is given
by 
\[
\int_{\overline{\mathcal{M}}_{g,n}}\lambda_{g}^{2}\psi_{1}^{d_{1}}\cdots\psi_{n}^{d_{n}}=0.
\]
\end{rem}

\paragraph{Recap table.}

The following table presents the structure of the correlators of the
quantum Witten-Kontsevich series coming from Theorem~\ref{thm: main thm},
Proposition~\ref{prop: top vanishing correlators }, Conjecture~\ref{conj: minimal vanishing correlators}
and Conjecture~\ref{conj: expression minimal correlator}. Fix three
nonnegative integers $l,k$ and $n$. In the box corresponding to
the $l$-th line and the $k$-th column we store the correlators
\[
\langle\tau_{d_{1}}\dots\tau_{d_{n}}\rangle_{l,k}
\]
where $\left(d_{1},\dots,d_{n}\right)$ is a list of nonnegative integers,
coming from the coefficients of $\epsilon^{2l}$$\hbar^{k}$ of $F^{\left(q\right)}$.
We set $g:=l+k$ in the table.

The first row and column in the table present proved facts, while
the boxes $k,l\geq1$ present conjectures.
\begin{center}
\begin{tabular}{c|c|c|c|c|c}
 & $\hbar^{0}$ & $\hbar^{1}$  & $\cdots$ & $\hbar^{k}$ & \tabularnewline
\cline{1-3} \cline{5-5} 
$\epsilon^{0}$ & $\int_{\overline{\mathcal{M}}_{0,n}}\psi_{1}^{d_{1}}\cdots\psi_{n}^{d_{n}}$ & {\renewcommand{\arraystretch}{7}
\begin{tabular}{C{4cm}}
$\left\langle \left\langle \tau_{d_{1}}\dots\tau_{d_{n}}\right\rangle \right\rangle _{1}$ {\fontsize{7}{8} s.t. $\sum d_{i}=4-3+n$}\\
\cdashline{1-1}
$\left\langle \left\langle \tau_{d_{1}}\dots\tau_{d_{n}}\right\rangle \right\rangle _{1}=\int_{\overline{\mathcal{M}}_{1,n}}\lambda_{0}\lambda_{1}\psi_{1}^{d_{1}}\cdots\psi_{n}^{d_{n}}$ \\

\end{tabular}} & $\cdots$ & {\renewcommand{\arraystretch}{3} \begin{tabular}{C{4,3cm}} 
$\left\langle \left\langle \tau_{d_{1}}\dots\tau_{d_{n}}\right\rangle \right\rangle _{g}$ {\fontsize{7}{7} s.t. $\sum d_{i}=4g-3+n$}\\
\cdashline{1-1} 
$\left\langle \left\langle \tau_{d_{1}}\dots\tau_{d_{n}}\right\rangle \right\rangle _{g}$ {\fontsize{7}{7} s.t. $\sum d_{i}=4g-5+n$}\\
\cdashline{1-1} 
$\vdots$ \\ 
\cdashline{1-1}
$\left\langle \left\langle \tau_{d_{1}}\dots\tau_{d_{n}}\right\rangle \right\rangle _{g}$ {\fontsize{7}{7} s.t. $\sum d_{i}=2g-1+n$}\\
\cdashline{1-1} 
$\left\langle \left\langle \tau_{d_{1}}\dots\tau_{d_{n}}\right\rangle \right\rangle _{g}=\int_{\overline{\mathcal{M}}_{g,n}}\lambda_{g}\psi_{1}^{d_{1}}\cdots\psi_{n}^{d_{n}}$ 
\end{tabular}} & $\left.\rule{0cm}{3,7cm}\right\}\underset{=g}{\underbrace{k}}+1$  \qquad \qquad\tabularnewline
\cline{1-3} \cline{5-5} 
$\epsilon^{2}$ & {\renewcommand{\arraystretch}{7}
\begin{tabular}{C{2,5cm}}
$\int_{\overline{\mathcal{M}}_{1,n}}\psi_{1}^{d_{1}}\cdots\psi_{n}^{d_{n}}$\\
\cdashline{1-1}
$0$ \\
\end{tabular}} & {\renewcommand{\arraystretch}{4}
\begin{tabular}{C{4cm}}
? \\
\cdashline{1-1}
?\\
\cdashline{1-1}
$\int_{\overline{\mathcal{M}}_{2,n}}\lambda_{1}\lambda_{2}\psi_{1}^{d_{1}}\cdots\psi_{n}^{d_{n}}$ \\
\end{tabular}} & $\cdots$ & {\renewcommand{\arraystretch}{3}
\begin{tabular}{C{4,3cm}}
? \\
\cdashline{1-1}
$\vdots$\\
\cdashline{1-1}
?\\
\cdashline{1-1}
$\int_{\overline{\mathcal{M}}_{g,n}}\lambda_{1}\lambda_{g}\psi_{1}^{d_{1}}\cdots\psi_{n}^{d_{n}}$ \\
\end{tabular}} & $\left.\rule{0cm}{2,8cm}\right\}\underset{=g}{\underbrace{k+1}}+1$ \qquad \qquad\tabularnewline
\cline{1-3} \cline{5-5} 
\multicolumn{1}{c}{$\vdots$} & \multicolumn{1}{c}{$\vdots$} & \multicolumn{1}{c}{$\vdots$} & \multicolumn{1}{c}{$\ddots$} & \multicolumn{1}{c}{$\vdots$} & \tabularnewline
\cline{1-3} \cline{5-5} 
$\epsilon^{2l}$ & {\renewcommand{\arraystretch}{3}
\begin{tabular}{C{2,5cm}}
$\int_{\overline{\mathcal{M}}_{g,n}}\psi_{1}^{d_{1}}\cdots\psi_{n}^{d_{n}}$ \\
\cdashline{1-1}
$0$\\
\cdashline{1-1}
$\vdots$\\
\cdashline{1-1}
$0$ \\
\end{tabular}} & {\renewcommand{\arraystretch}{3}
\begin{tabular}{C{4cm}}
? \\
\cdashline{1-1}
$\vdots$\\
\cdashline{1-1}
?\\
\cdashline{1-1}
$\int_{\overline{\mathcal{M}}_{g,n}}\lambda_{g-1}\lambda_{g}\psi_{1}^{d_{1}}\cdots\psi_{n}^{d_{n}}$ \\
\end{tabular}} & $\cdots$ & {\renewcommand{\arraystretch}{3}
\begin{tabular}{C{4,3cm}}
? \\
\cdashline{1-1}
$\vdots$\\
\cdashline{1-1}
?\\
\cdashline{1-1}
$\int_{\overline{\mathcal{M}}_{g,n}}\lambda_{l}\lambda_{g}\psi_{1}^{d_{1}}\cdots\psi_{n}^{d_{n}}$ \\
\end{tabular}} & $\left.\rule{0cm}{2,5cm}\right\}\underset{=g}{\underbrace{k+l}}+1$  \qquad\qquad\tabularnewline
\cline{1-3} \cline{5-5} 
\end{tabular}
\par\end{center}

\subsubsection{String and dilaton for the quantum Witten-Kontsevich series\label{subsec : String-and-dilaton}}

In~\cite{goulden2005towards}, the authors prove the string and dilaton
equations for the Hurwitz correlators. Hence Theorem~\ref{thm: main thm}
implies that the $\epsilon=0$ restriction of the quantum Witten-Kontsevich
series satisfies the string and dilaton equations. These equations
are actually satisfied by the full quantum Witten-Kontsevich series.
\begin{thm}
\label{thm: string equation}The quantum Witten-Kontsevich series
satisfies the string equation
\[
\frac{\partial}{\partial t_{0}}F^{\left(q\right)}=\sum_{i\geq0}t_{i+1}\frac{\partial}{\partial t_{i}}F^{\left(q\right)}+\frac{t_{0}^{2}}{2}-\frac{i\hbar}{24}.
\]
\end{thm}
\begin{conjecture}
The quantum Witten-Kontsevich series satisfies the dilaton equation
\[
\frac{\partial}{\partial t_{1}}F^{\left(q\right)}=\sum_{i\geq0}t_{i}\frac{\partial}{\partial t_{i}}F^{\left(q\right)}+\epsilon\frac{\partial}{\partial\epsilon}F^{\left(q\right)}+2\hbar\frac{\partial}{\partial\hbar}F^{\left(q\right)}-2F^{\left(q\right)}+\frac{\epsilon^2}{24}.
\]
{} %
\end{conjecture}

\subsection{Plan of the paper}

In Section~\ref{sec: string equation}, we prove the string equation
that is Theorem~\ref{thm: string equation}. 

In Section~\ref{sec:Eulerian-numbers}, we prove some properties
of Eulerian numbers. These properties bring a key simplification in
the proof of Theorem~\ref{thm: main thm}. 

In Section~\ref{sec: Proof main thm}, we prove Theorem~\ref{thm: main thm}. 

in Section~\ref{sec:A-combinatorial-identity}, we prove a combinatorial
identity that is used in the last step of the proof Theorem~\ref{thm: main thm}.

In Section~\ref{sec: Level Structure}, we prove the vanishing properties
of the correlators stated in Proposition~\ref{prop: top vanishing correlators }.

\subsection{Acknowledgments}

I am very grateful to my PhD advisor Dimitri Zvonkine for his support
and expert guidance at every step of this project. I would also like
to thanks Paolo Rossi who introduced me to quantum tau functions and
who suggested me to study them. 

\section{The string equation\label{sec: string equation}}

The string equation for the quantum Witten-Kontsevich series does
not directly come from a comparison between the $\psi$-classes and
their pull-backs as it usually does. Indeed, the quantum Witten-Kontsevich
series is defined by various commutations of Hamiltonians, that are
themselves defined by integration of $\psi$-classes over the double
ramification cycle, and then the substitution $u_{i}=\delta_{i,1}$.
We have to follow this definition to prove the string equation and
any other property of the quantum Witten-Kontsevich series. 

\subsection{On the substitution $u_{i}=\delta_{i,1}$}

The definition of the quantum Witten-Kontsevich series uses differential
polynomials as elements of $\mathcal{A}$ and $\tilde{\mathcal{A}}$.
We need to explain how to substitute $u_{i}=\delta_{i,1}$ in an element
of $\tilde{\mathcal{A}}$. This is the purpose of the following lemma. 
\begin{lem}
\label{lem:evaluation }Let $\phi$ be a differential polynomial.
Write it as a formal Fourier series, that is 
\[
\phi\left(x\right)=\sum_{k=0}^{d}\sum_{a_{1},\dots,a_{k}\in\mathbb{Z}}\phi_{k}\left(a_{1},\dots,a_{k}\right)p_{a_{1}}\cdots p_{a_{k}}e^{ix(a_{1}+\cdots+a_{k})}
\]
where $\phi_{k}(a_{1},\dots,a_{k})\in\mathbb{C}\left[a_{1},\dots,a_{k}\right]\left[\left[\epsilon,\hbar\right]\right]$
is a symmetric polynomial in its $k$ indeterminates $a_{1},\dots,a_{k}$
for $0\leq k\leq d$. The substitution $u_{i}=\delta_{i,1}$ in $\phi$
is given by
\[
\phi\Big\vert_{u_{i}=\delta_{i,1}}=\sum_{k\geq0}\left(-i\right)^{k}\left[a_{1}\cdots a_{k}\right]\phi_{k}\left(a_{1},\dots,a_{k}\right),
\]
where $\left[a_{1}\cdots a_{k}\right]$ denotes the coefficient of
$a_{1}\cdots a_{k}$. 
\end{lem}
\begin{proof}
The Fourier series of $u$ and its $s$-th derivative is given by
$u_{s}\left(x\right)=\sum_{a\in\mathbb{Z}}\left(ia\right)^{s}p_{a}e^{iax}$.
Hence we get
\begin{align*}
\phi\left(x\right) & =\sum_{k\geq0}\sum_{a_{1},...,a_{k}\in\mathbb{Z}}\left(\sum_{s_{1,}\dots,s_{k}\geq0}a_{1}^{s_{1}}\cdots a_{k}^{s_{k}}\left[a_{1}^{s_{1}}\cdots a_{k}^{s_{k}}\right]\phi_{k}\left(a_{1},\dots,a_{k}\right)\right)p_{a_{1}}\cdots p_{a_{k}}e^{ix(a_{1}+\cdots+a_{k})}\\
 & =\sum_{k\geq0}\sum_{s_{1,}\dots,s_{k}\geq0}\left(-i\right)^{s_{1}+\ldots+s_{k}}u_{s_{1}}\cdots u_{s_{k}}\left[a_{1}^{s_{1}}\cdots a_{k}^{s_{k}}\right]\phi_{k}\left(a_{1},\dots,a_{k}\right).
\end{align*}
Hence $\phi\Bigr|_{u_{i}=\delta_{i,1}}=\sum_{k\geq0}\left(-i\right)^{k}\left[a_{1}\cdots a_{k}\right]\phi_{k}\left(a_{1},\dots,a_{k}\right)$.
\end{proof}

\subsection{A proof of the string equation}

We prove in this section  that the quantum Witten-Kontsevich series
satisfies the string equation
\[
\frac{\partial}{\partial t_{0}}F^{\left(q\right)}=\sum_{i\geq0}t_{i+1}\frac{\partial}{\partial t_{i}}F^{\left(q\right)}+\frac{t_{0}^{2}}{2}-\frac{i\hbar}{24}.
\]

\paragraph{Plan of the proof.}

The string equation is an equality of power series. In order to prove
this equation, we verify that the constant and linear terms of the
power series on the LHS and RHS correspond, then we show that the
second derivative of this equation is true. From the definition of
the quantum Witten-Kontsevich series, we find that the derivative
with respect to $t_{d_{1}}$ and $t_{d_{2}}$ of the string equation
yields
\begin{equation}
\frac{\partial}{\partial t_{0}}\Omega_{d_{1},d_{2}}^{\hbar,\bm{t}}\Big\vert_{u_{i}=\delta_{i,1}}=\sum_{k\geq0}t_{k+1}\frac{\partial}{\partial t_{k}}\Omega_{d_{1},d_{2}}^{\hbar,\bm{t}}\Big\vert_{u_{i}=\delta_{i,1}}+\Omega_{d_{1}-1,d_{2}}^{\hbar,\bm{t}}\Big\vert_{u_{i}=\delta_{i,1}}+\Omega_{d_{1},d_{2}-1}^{\hbar,\bm{t}}\Big\vert_{u_{i}=\delta_{i,1}}+\delta_{0,d_{1}}\delta_{0,d_{2}},\label{eq: string pour Omega}
\end{equation}
where $d_{1}$ and $d_{2}$ are two nonnegative integers. 

Before proving Eq.~(\ref{eq: string pour Omega}), let us focus on
the constant and linear terms of the string equation. To do so, we
need the following lemma to compute the two-point function with a
zero insertion. 
\begin{lem}
\label{lem: Omeg_0,d H_d-1}Fix a nonnegative integer $d$. We have
\[
\Omega_{0,d}^{\hbar}=H_{d-1}.
\]
\end{lem}
\begin{proof}
This is an equality between elements of $\tilde{A}$, that is between
two lists of polynomials. The equality between the first polynomial
of each list follows form $\Omega_{0,d}^{\hbar}\Big\vert_{p_{*}=0}=H_{d-1}\Big\vert_{p_{*}=0}$.
To obtain the equality of the rest of the polynomials, it is enough
to check that $\partial_{x}\Omega_{0,d}^{\hbar}=\partial_{x}H_{d-1}.$

Using the definition of the two-point function and the tau symmetry,
we find
\[
\partial_{x}\Omega_{0,d}^{\hbar}=\frac{1}{\hbar}\left[H_{-1},\overline{H}_{d}\right]=\frac{1}{\hbar}\left[H_{d-1},\overline{H_{0}}\right].
\]
We noticed in Remark~\ref{rem: d_t_0 =00003D d_x} that the commutator
of any differential polynomial with $\overline{H}_{0}$ corresponds
to the derivation with respect to $x$. Hence we get the equality
\[
\partial_{x}\Omega_{0,d}^{\hbar}=\partial_{x}H_{d-1}.
\]
\end{proof}

\paragraph{Constant term of the string equation.}

The constant term in the RHS of the string equation is given by $-\frac{i\hbar}{24}$.
We then have to show that $-\frac{i\hbar}{24}$ is also the coefficient
of $t_{0}$ in $F^{\left(q\right)}$. By construction, the coefficient
of $t_{0}$ is the coefficient of $t_{0}t_{1}$. The coefficient of
$t_{0}t_{1}$ in $F^{\left(q\right)}$ is given by $\Omega_{0,1}\Big\vert_{u_{i}=\delta_{i,1}}\underset{\text{Lemma }\ref{lem: Omeg_0,d H_d-1}}{=}H_{0}\Big\vert_{u_{i}=\delta_{i,1}}$.
We use the expression of $H_{0}$ given by the following lemma to
deduce that this coefficient is $-\frac{i\hbar}{24}$. 
\begin{lem}
\label{lem: calcul H0}We have
\[
H_{0}=\frac{u_{0}^{2}}{2}+\epsilon^{2}\frac{u_{2}}{12}-\frac{i\hbar}{24}.
\]
\end{lem}
\begin{proof}
We have to compute the polynomials
\[
P_{m,g,k}\left(a_{1},\dots,a_{m}\right)=\int_{{\rm DR}_{g}\left(0,a_{1},\dots,a_{m},-\sum a_{i}\right)}\psi_{1}\lambda_{k},
\]
for any $m,g,k\geq0$. The double ramification cycle ${\rm DR}_{g}\left(0,a_{1},\dots,a_{m},-\sum a_{i}\right)$
is a polynomial in the parts of the ramifications profiles with coefficients
in $H^{2g}\left(\overline{\mathcal{M}}_{g,n+2}\right)$, hence these
polynomials vanishes for dimensional reasons if $\left(m,g,k\right)\neq\left(2,0,0\right),\left(1,1,1\right),\left(0,1,0\right)\text{ or }\left(0,2,2\right).$
We have $P_{2,0,0}=\int_{\overline{\mathcal{M}}_{0,4}}\psi_1=1$ and
$P_{1,1,1}=\frac{a^{2}}{12}$ (see \cite[Section 4.3.1]{buryak2015double}).
If $\left(m,g,k\right)=\left(0,1,0\right)\text{ and }\left(0,2,2\right)$
we use that ${\rm DR}_{g}\left(0,0\right)=\left(-1\right)^{g}\lambda_{g}$
(see~\cite{janda2017double}) to get 
\[
P_{0,1,0}=\int_{{\rm DR}_{1}\left(0,0\right)}\psi_1=-\int_{\overline{\mathcal{M}}_{1,2}}\psi_1\lambda_{1}=-\frac{1}{24},\text{ and }P_{0,2,2}=\int_{{\rm DR}_{2}\left(0,0\right)}\psi_1\lambda_{2}=0.
\]
Using the expression of the Fourier series of $u$ and its derivative,
we obtain our expression of $H_{0}$.
\end{proof}

\paragraph{Linear terms of the string equation.}

We now focus on the linear terms. The coefficient of $t_{0}$ does
not appear on the RHS of the string equation. Let us show that it
vanishes in the LHS. The coefficient of $t_{0}$ in the LHS of string
equation is given by the coefficient of $t_{0}t_{0}$ in $F^{\left(q\right)}$.
This coefficient is 
\[
\Omega_{0,0}\Big\vert_{u_{i}=\delta_{i,1}}\underset{\text{Lemma }\ref{lem: Omeg_0,d H_d-1}}{=}H_{-1}\Big\vert_{u_{i}=\delta_{i,1}}=u_{0}\Big\vert_{u_{i}=\delta_{i,1}}=0.
\]
We used $H_{-1}=u_{0}$ as one can verify from the definition.

Fix an integer $d\geq1$. The coefficient of $t_{d}$ on the LHS of
the string equation is given by the coefficient of $t_{0}t_{d}$ in
$F^{\left(q\right)}$. The coefficient of $t_{d}$ on the RHS of the
string equation is given by the coefficient of $t_{d-1}$ in $F^{\left(q\right)}$,
which is equal to the coefficient of $t_{0}t_{d}$ in $F^{\left(q\right)}$
by definition of the linear terms of $F^{\left(q\right)}$. 

\paragraph{Some necessary lemmas.}

We prove three lemmas used for the proof of Equation~(\ref{eq: string pour Omega}).
The first one contains the geometric origin of the string equation.
\begin{lem}
\label{lem: string for hamiltonians}Fix a nonnegative integer $d$.
The Hamiltonian density $H_{d}$ satisfy the string equation
\[
\frac{\partial}{\partial p_{0}}H_{d}=H_{d-1}.
\]
\end{lem}
The proof is analogous to the one of Lemma~$2.7$ in~\cite{BuryakRossi2016}.
\begin{proof}
We have
\[
\frac{\partial}{\partial p_{0}}H_{d}=\sum_{\overset{g\geq0,m\geq0}{2g+m+1>0}}\frac{\left(i\hbar\right)^{g}}{m!}\sum_{a_{1},\dots,a_{m}\in\mathbb{Z}}\left(\int_{{\rm DR}_{g}\left(0,a_{1},\dots,a_{m},0,-\sum a_{i}\right)}\psi_{1}^{d+1}\Lambda\left(\frac{-\epsilon^{2}}{i\hbar}\right)\right)p_{a_{1}}\cdots p_{a_{m}}e^{ix\sum a_{i}}.
\]
Let $\pi:\overline{\mathcal{M}}_{g,m+3}\rightarrow\overline{\mathcal{M}}_{g,m+2}$
be the map defined when $\left(g,m\right)\neq\left(0,0\right)$ that
forgets the $\left(m+2\right)$-th marked point. We use that $\pi^{*}{\rm DR}_{g}\left(0,a_{1},\dots,a_{m},-\sum a_{i}\right)={\rm DR}_{g}\left(0,a_{1},\dots,a_{m},0,-\sum a_{i}\right)$,
$\pi^{*}\left(\Lambda\left(\frac{-\epsilon^{2}}{i\hbar}\right)\right)=\Lambda\left(\frac{-\epsilon^{2}}{i\hbar}\right)$
and $\psi_{1}^{d+1}=\pi^{*}\left(\psi_{1}^{d+1}\right)+D\pi^{*}\left(\psi_{1}^{d}\right)$,
where $D$ is the divisor with a bubble containing the marked points
$1$ and $\left(m+2\right)$, to obtain
\[
\int_{{\rm DR}_{g}\left(0,a_{1},\dots,a_{m},0,-\sum a_{i}\right)}\psi_{1}^{d+1}\Lambda\left(\frac{-\epsilon^{2}}{i\hbar}\right)=\begin{cases}
\int_{{\rm DR}_{g}\left(0,a_{1},\dots,a_{m},-\sum a_{i}\right)}\psi_{1}^{d}\Lambda\left(\frac{-\epsilon^{2}}{i\hbar}\right) & \text{if }2g+m>0\text{ and }d>0\\
0 & \text{if }2g+m>0\text{ and }d=0\\
\delta_{d,0} & \text{if }g=0,m=0\text{ and }d=0.
\end{cases}
\]
This proves the lemma. 
\end{proof}
\begin{lem}
\label{lem: derivee Omega_p,q}Fix $d_{1},d_{2}$ two positive integers,
we have
\[
\frac{\partial\Omega_{d_{1},d_{2}}^{\hbar}}{\partial p_{0}}=\Omega_{d_{1}-1,d_{2}}^{\hbar}+\Omega_{d_{1},d_{2}-1}^{\hbar}+\delta_{0,d_{1}}\delta_{0,d_{2}},
\]
where we use the convention that $\Omega^{\hbar}$ vanishes if at
least one of its indices is negative.
\end{lem}
\begin{proof}
If $\left(d_{1},d_{2}\right)=\left(0,0\right)$, we obtain $\Omega_{0,0}=H_{-1}=u_{0}$
from Lemma~\ref{lem: Omeg_0,d H_d-1}. Hence the equation is satisfied.

Otherwise, we want to prove an equality of elements of $\tilde{\mathcal{A}}$,
that is an equality of two lists of symmetric polynomials. The equality
of the first polynomial of each list follows from the choice of the
constant $\Omega_{d_{1},d_{2}}^{\hbar}\Big\vert_{p_{*}=0}$ in the
definition of $\Omega_{d_{1},d_{2}}^{\hbar}$. To obtain the equality
for the rest of the polynomials, it is enough to prove that the $x$-derivative
of the equation is verified. From the definition of $\Omega_{d_{1},d_{2}}^{\hbar}$,
we have
\begin{align*}
\partial_{x}\frac{\partial}{\partial p_{0}}\Omega_{d_{1},d_{2}}^{\hbar}=\frac{\partial}{\partial p_{0}}\partial_{x}\Omega_{d_{1},d_{2}}^{\hbar} & =\frac{\partial}{\partial p_{0}}\frac{1}{\hbar}\left[H_{d_{1}-1},\overline{H}_{d_{2}}\right]\\
 & =\frac{1}{\hbar}\left[\frac{\partial H_{d_{1}-1}}{\partial p_{0}},\overline{H}_{d_{2}}\right]+\frac{1}{\hbar}\left[H_{d_{1}-1},\frac{\partial\overline{H}_{d_{2}}}{\partial p_{0}}\right]\\
 & =\partial_{x}\Omega_{d_{1}-1,d_{2}}^{\hbar}+\partial_{x}\Omega_{d_{1},d_{2}-1}^{\hbar}.
\end{align*}
We used Lemma~\ref{lem: string for hamiltonians} to obtain the last
equality.
\end{proof}
\begin{lem}
Let $\phi$ be a differential polynomial, we have
\[
\partial_{x}\phi\Big\vert_{u_{i}=\delta_{i,1}}=\frac{\partial\phi}{\partial p_{0}}\Bigg\vert_{u_{i}=\delta_{i,1}}.
\]
\end{lem}
\begin{proof}
Write $\phi$ as an element of $\tilde{\mathcal{A}}$, that is
\[
\phi\left(x\right)=\sum_{k\geq0}^{d}\sum_{a_{1},...,a_{k}\in\mathbb{Z}}\phi_{k}(a_{1},\dots,a_{k})p_{a_{1}}\cdots p_{a_{k}}e^{ix(a_{1}+\cdots+a_{k})},
\]
where $\phi_{k}(a_{1},\dots,a_{k})\in\mathbb{C}\left[a_{1},\dots,a_{k}\right]\left[\left[\epsilon,\hbar\right]\right]$
is a symmetric polynomial in its $k$ indeterminates $a_{1},\dots,a_{k}$
for $0\leq k\leq d$. Then, thanks to Lemma~\ref{lem:evaluation },
we have
\begin{align*}
\partial_{x}\phi\Big\vert_{u_{i}=\delta_{i,1}} & =\sum_{k\geq0}\left(-1\right)^{k}i^{k+1}\left[a_{1}\cdots a_{k}\right]\phi_{k}(a_{1},\dots,a_{k})\left(a_{1}+\cdots+a_{k}\right)\\
 & =\sum_{k\geq0}\left(-1\right)^{k}i^{k+1}\sum_{j=1}^{k}\left[a_{1}\cdots\hat{a}_{j}\cdots a_{k}\right]\phi_{k}\left(a_{1},\dots,a_{j-1},0,a_{j+1},\dots,a_{k}\right)\\
 & =\frac{\partial\phi}{\partial p_{0}}\Bigg\vert_{u_{i}=\delta_{i,1}}.
\end{align*}
\end{proof}

\paragraph{Proof of Equation~(\ref{eq: string pour Omega}).}

We first recall this equation:
\[
\frac{\partial}{\partial t_{0}}\Omega_{d_{1},d_{2}}^{\hbar,\bm{t}}\Big\vert_{u_{i}=\delta_{i,1}}=\sum_{k\geq0}t_{k+1}\frac{\partial}{\partial t_{k}}\Omega_{d_{1},d_{2}}^{\hbar,\bm{t}}\Big\vert_{u_{i}=\delta_{i,1}}+\Omega_{d_{1}-1,d_{2}}^{\hbar,\bm{t}}\Big\vert_{u_{i}=\delta_{i,1}}+\Omega_{d_{1},d_{2}-1}^{\hbar,\bm{t}}\Big\vert_{u_{i}=\delta_{i,1}}+\delta_{0,d_{1}}\delta_{0,d_{2}}.
\]

\begin{proof}
We prove this equality at every degree in the indeterminates $\left(t_{0},t_{1},\dots\right)$.
Recall that $\Omega_{d_{1},d_{2}}^{\hbar,\bm{t}}=\exp\left(\sum_{k\geq0}\frac{t_{k}}{\hbar}\left[\cdot,\overline{H}_{k}\right]\right)\Omega_{d_{1},d_{2}}^{\hbar}$.
Let $n\geq0$ and $\left(d_{3},\dots,d_{n}\right)$ be a list of nonnegative
integers. Then the coefficient of $t_{d_{3}}\cdots t_{d_{n}}$ of the
LHS is given by
\begin{align*}
\frac{1}{\hbar^{n-1}}\left[\left[\cdots\left[\Omega_{d_{1},d_{2}}^{\hbar},\overline{H}_{d_{3}}\right]\cdots,\overline{H}_{d_{n}}\right],\overline{H}_{0}\right]\Bigg\vert_{u_{i}=\delta_{i,1}} & =\frac{1}{\hbar^{n-2}}\partial_{x}\left[\cdots\left[\Omega_{d_{1},d_{2}}^{\hbar},\overline{H}_{d_{3}}\right]\cdots,\overline{H}_{d_{n}}\right]\Bigg\vert_{u_{i}=\delta_{i,1}}\\
 & =\frac{1}{\hbar^{n-2}}\frac{\partial}{\partial p_{0}}\left[\cdots\left[\Omega_{d_{1},d_{2}}^{\hbar},\overline{H}_{d_{3}}\right]\cdots,\overline{H}_{d_{n}}\right]\Bigg\vert_{u_{i}=\delta_{i,1}}.
\end{align*}
We act with $\frac{\partial}{\partial p_{0}}$ on every elements of
the commutators. Then we use Lemmas~\ref{lem: string for hamiltonians}
and~\ref{lem: derivee Omega_p,q} to find
\begin{align*}
& \frac{1}{\hbar^{n-2}}\frac{\partial}{\partial p_{0}}\left[\cdots\left[\Omega_{d_{1},d_{2}}^{\hbar},\overline{H}_{d_{3}}\right]\cdots,\overline{H}_{d_{n}}\right]\Bigg\vert_{u_{i}=\delta_{i,1}} = \\
& \frac{1}{\hbar^{n-2}} \Bigg(\left[\cdots\left[\Omega_{d_{1}-1,d_{2}},\overline{H}_{d_{3}}\right]\cdots,\overline{H}_{d_{n}}\right]\Bigg\vert_{u_{i}=\delta_{i,1}}+\left[\cdots\left[\Omega_{d_{1},d_{2}-1},\overline{H}_{d_{3}}\right]\cdots,\overline{H}_{d_{n}}\right]\Bigg\vert_{u_{i}=\delta_{i,1}}\\
 & +\left[\cdots\left[\delta_{0,d_{1}}\delta_{0,d_{2}},\overline{H}_{d_{3}}\right]\cdots,\overline{H}_{d_{n}}\right]\Bigg\vert_{u_{i}=\delta_{i,1}}\\
 & +\left[\cdots\left[\Omega_{d_{1},d_{2}},\overline{H}_{d_{3}-1}\right]\cdots,\overline{H}_{d_{n}}\right]\Bigg\vert_{u_{i}=\delta_{i,1}}+\cdots+\left[\cdots\left[\Omega_{d_{1},d_{2}},\overline{H}_{d_{3}}\right]\cdots,\overline{H}_{d_{n}-1}\right]\Bigg\vert_{u_{i}=\delta_{i,1}}\Bigg).
\end{align*}
We recognize the coefficient of $t_{d_{3}}\dots t_{d_{n}}$ of the
RHS Equation~(\ref{eq: string pour Omega}).
\end{proof}

\section{\label{sec:Eulerian-numbers}Eulerian numbers}

Eulerian numbers naturally appear in the proof of the main theorem.
We will need some of their properties in this proof. 
\begin{defn}
Fix two nonnegative integers $k,n$ and a permutation $\sigma\in S_{n}$.
A descent of the permutation $\sigma$ is an integer $i\in\left\{ 1,\dots,n\right\} $
such that $\sigma\left(i\right)>\sigma\left(i+1\right)$. The \emph{Eulerian
number} $\left\langle \begin{array}{c}
n\\
k
\end{array}\right\rangle $ is the number of permutation of $S_{n}$ with $k$ descents.

The \emph{Eulerian polynomial} $E_{n}\left(t\right)$ is the generating
polynomial of Eulerian numbers
\[
E_{n}\left(t\right):=\sum_{k=0}^{n-1}\left\langle \begin{array}{c}
n\\
k
\end{array}\right\rangle t^{k}.
\]
The following two propositions are basic properties of Eulerian numbers.
They are proved in the first chapter of \cite{petersen2015eulerian}.
\end{defn}
\begin{prop}
[Carlitz identity]\label{prop: Carlitz}Let $d$ be a nonnegative
integer. Let $t$ be a formal variable. We have
\[
\sum_{k\geq1}k^{d}t^{k}=\frac{tE_{d}\left(t\right)}{\left(1-t\right)^{d+1}}.
\]
\end{prop}
\begin{rem}
\label{rem: Why Eulerian appear}In the proof of the main theorem,
we need to compute the numbers 
\[
\sum_{k_{1}+\cdots+k_{q}=N}k_{1}^{d_{1}}\cdots k_{q}^{d_{q}}=\left[t^{N}\right]\prod_{i=1}^{q}\left(\sum_{k\geq1}k^{d_{i}}t^{k}\right),
\]
where $d_{1},\dots,d_{q}$ and $N$ are nonnegative integers numbers.
According to Carliz identity, this can be done using Eulerian number.
We will use this fact to simplify our computations.
\end{rem}
\begin{prop}
\label{prop: gene series Eulerian numbers}Let $t$ and $z$ be two
formal variables. We have
\[
\sum_{n\geq0}\frac{E_{n}\left(t\right)}{n!}z^{n}=\frac{t-1}{t-e^{z\left(t-1\right)}}.
\]
\end{prop}
\begin{cor}
\label{cor: sum Eulerian gives ln}Let $t$ and $z$ be two formal
variables. We have
\[
\sum_{n\geq0}\frac{tE_{n}\left(t\right)z^{n+1}}{\left(n+1\right)!\left(1-t\right)^{n+1}}=-z-\ln\left(t-e^{-z}\right)+\ln\left(t-1\right).
\]
\end{cor}
\begin{proof}
Start from the formula of From Proposition~\ref{prop: gene series Eulerian numbers},
that is
\[
\sum_{n\geq0}\frac{E_{n}\left(t\right)}{n!}z^{n}=\frac{t-1}{t-e^{z\left(t-1\right)}}.
\]
First multiply both sides of this equality by $\frac{t}{1-t}$. Then
substitute in both sides $z:=\frac{z}{1-t}$. Finally choose in both
sides the primitive with respect to $z$ which vanishes when $z=0$.
On the LHS, we obtain $\sum_{n\geq0}\frac{tE_{n}\left(t\right)z^{n+1}}{\left(n+1\right)!\left(1-t\right)^{n+1}}$.
On the RHS we obtain $-z-\ln\left(t-e^{-z}\right)+\ln\left(t-1\right)$.
\end{proof}
\begin{notation}
\label{Notation S}In the following, we use the notation 
\[ S\left(z\right)=\frac{{\rm sh}\left(z/2\right)}{z/2}=\sum_{l\geq0}\frac{z^{2l}}{2^{2l}\left(2l+1\right)!}. \] 
\end{notation}
\begin{lem}
\label{lem: technique Eulerian 1}Let $A,B,t$ and $z$ be some formal
variables. We have
\[
\sum_{k>0}ABz^{2}kS\left(kAz\right)S\left(kBz\right)t^{k}=\ln\left(\frac{\left(1-te^{\frac{A-B}{2}z}\right)\left(1-te^{-\frac{A-B}{2}z}\right)}{\left(1-te^{\frac{A+B}{2}z}\right)\left(1-te^{-\frac{A+B}{2}z}\right)}\right).
\]
\end{lem}
\begin{proof}
We start from the LHS. Use the developed expression of $S$ (see Notation~\ref{Notation S})
to obtain
\begin{align*}
\sum_{k>0}ABz^{2}kS\left(kAz\right)S\left(kBz\right)t^{k} & =\sum_{k>0}\sum_{l_{1},l_{2}\geq0}\frac{A^{2l_{1}+1}B^{2l_{2}+1}z^{2\left(l_{1}+l_{2}\right)+2}}{2^{2\left(l_{1}+l_{2}\right)}\left(2l_{1}+1\right)!\left(2l_{2}+1\right)!}k^{2\left(l_{1}+l_{2}\right)+1}t^{k}.
\end{align*}
Then, we use the Carlitz identity (Proposition~\ref{prop: Carlitz})
to compute the sum running over $k$. We obtain 
\[
\sum_{l_{1},l_{2}\geq0}\frac{A^{2l_{1}+1}B^{2l_{2}+1}z^{2\left(l_{1}+l_{2}\right)+2}}{2^{2\left(l_{1}+l_{2}\right)}\left(2l_{1}+1\right)!\left(2l_{2}+1\right)!}\frac{tE_{2\left(l_{1}+l_{2}\right)+1}\left(t\right)}{\left(1-t\right)^{2\left(l_{1}+l_{2}\right)+2}}.
\]
Simplifying this expression using 
\[
\sum_{l_{1}+l_{2}=l}\frac{A^{2l_{1}+1}B^{2l_{2}+1}}{\left(2l_{1}+1\right)!\left(2l_{2}+1\right)!}=\frac{1}{2}\frac{\left(A+B\right)^{2l+2}-\left(A-B\right)^{2l+2}}{\left(2l+2\right)!},
\]
we obtain
\[
\sum_{l\geq0}\frac{\left(A+B\right)^{2l+2}-\left(A-B\right)^{2l+2}}{2^{2l+1}\left(2l+2\right)!}z^{2l+2}\frac{tE_{2l+1}\left(t\right)}{\left(1-t\right)^{2l+2}}.
\]
Denote by $F\left(z,t\right):=-z-\ln\left(t-e^{-z}\right)+\ln\left(t-1\right).$
According to Corollary~\ref{cor: sum Eulerian gives ln}, we have
\[
\sum_{l\geq0}\frac{\left(A+B\right)^{2l+2}}{2^{2l+1}\left(2l+2\right)!}z^{2l+2}\frac{tE_{2l+1}\left(t\right)}{\left(1-t\right)^{2l+2}}=F\left(\frac{A+B}{2}z,t\right)+\left(-\frac{A+B}{2}z,t\right)
\]
and
\[
\sum_{l\geq0}\frac{\left(A-B\right)^{2l+2}}{2^{2l+1}\left(2l+2\right)!}z^{2l+2}\frac{tE_{2l+1}\left(t\right)}{\left(1-t\right)^{2l+2}}=F\left(\frac{A-B}{2}z,t\right)+\left(-\frac{A-B}{2}z,t\right).
\]
Finally, we remark that
\begin{align*}
 & F\left(\frac{A+B}{2}z,t\right)+\left(-\frac{A+B}{2}z,t\right)-\left(F\left(\frac{A-B}{2}z,t\right)+\left(-\frac{A-B}{2}z,t\right)\right)\\
 & =\ln\left(\frac{\left(1-te^{\frac{A-B}{2}z}\right)\left(1-te^{-\frac{A-B}{2}z}\right)}{\left(1-te^{\frac{A+B}{2}z}\right)\left(1-te^{-\frac{A+B}{2}z}\right)}\right)
\end{align*}
to obtain the RHS of the Lemma.
\end{proof}
\begin{lem}
\label{lem: technique Eulerian 2}Let $A,B$ and $t$ be some formal
variables. We have
\[
\frac{\left(1-te^{A-B}\right)\left(1-te^{B-A}\right)}{\left(1-te^{A+B}\right)\left(1-te^{-A-B}\right)}=1+4\sum_{k>0}\frac{{\rm sh}\left(A\right){\rm sh}\left(B\right)}{{\rm sh}\left(A+B\right)}{\rm sh}\left(k\left(A+B\right)\right)t^{k}.
\]
\end{lem}
\begin{proof}
Start from the LHS of the equality. Develop the two geometric series
of the denominators, we obtain
\[
\frac{1}{\left(1-te^{A+B}\right)\left(1-te^{-A-B}\right)}=\sum_{n,m\geq0}e^{\left(n-m\right)\left(A+B\right)}t^{n+m}=\sum_{k\geq0}\frac{{\rm sh}\left(\left(k+1\right)\left(A+B\right)\right)}{{\rm sh}\left(A+B\right)}t^{k}.
\]
Express the numerator as $\left(1-te^{A-B}\right)\left(1-te^{B-A}\right)=t^{2}-2t{\rm ch}\left(A-B\right)+1$.
We then obtain the following expression for $\frac{\left(1-te^{A-B}\right)\left(1-te^{B-A}\right)}{\left(1-te^{A+B}\right)\left(1-te^{-A-B}\right)}$
:
\[
\frac{1}{{\rm sh}\left(A+B\right)}\sum_{k\geq0}\left({\rm sh}\left(\left(k-1\right)\left(A+B\right)\right)-2{\rm ch}\left(A-B\right){\rm sh}\left(k\left(A+B\right)\right)+{\rm sh}\left(\left(k+1\right)\left(A+B\right)\right)\right)t^{k}.
\]
We use first the hyperbolic identity ${\rm sh}\left(\left(k-1\right)\left(A+B\right)\right)+{\rm sh}\left(\left(k+1\right)\left(A+B\right)\right)=2{\rm ch}\left(A+B\right){\rm sh}\left(k\left(A+B\right)\right)$
and then ${\rm ch}\left(A+B\right){\rm ch}\left(A-B\right)=2{\rm sh}\left(A\right){\rm sh}\left(B\right)$
to obtain the result.
\end{proof}
The following property is a key ingredient of the proof of the main
theorem.
\begin{cor}
\label{prop:Eulerian fin}Let $A,B,t$ and $z$ be some formal variables.
We have
\begin{align*}
\exp\left(\sum_{k>0}ABz^{2}kS\left(kAz\right)S\left(kBz\right)t^{k}\right) & =\frac{\left(1-te^{\frac{A-B}{2}z}\right)\left(1-te^{-\frac{A-B}{2}z}\right)}{\left(1-te^{\frac{A+B}{2}z}\right)\left(1-te^{-\frac{A+B}{2}z}\right)}\\
 & =1+4\sum_{k>0}\frac{{\rm sh}\left(\frac{A}{2}z\right){\rm sh}\left(\frac{B}{2}z\right)}{{\rm sh}\left(\frac{A+B}{2}z\right)}{\rm sh}\left(k\frac{A+B}{2}z\right)t^{k}.
\end{align*}
\end{cor}
\begin{proof}
The first equality is obtained from Lemma~\ref{lem: technique Eulerian 1}.
The second is given by the formula of Lemma~\ref{lem: technique Eulerian 2}
with $A:=\frac{A}{2}z$ and $B:=\frac{B}{2}z$.
\end{proof}

\section{\label{sec: Proof main thm}Proof of the main theorem}

We give in this section the proof of Theorem~\ref{thm: main thm},
that is we prove the equality
\[
\langle\tau_{d_{1}}\dots\tau_{d_{n}}\rangle_{0,g}=\left\langle \left\langle \tau_{d_{1}}\dots\tau_{d_{n}}\right\rangle \right\rangle _{g},
\]
where $g,n,d_{1},\dots,d_{n}$ are some nonnegative integers. First,
we explain the strategy of the proof.

\subsection{Computing $\langle\tau_{d_{1}}\dots\tau_{d_{n}}\rangle_{0,g}$ and
$\left\langle \left\langle \tau_{d_{1}}\dots\tau_{d_{n}}\right\rangle \right\rangle _{g}$}

In the next section, we show that the string equation allows one to
express $\langle\tau_{d_{1}}\dots\tau_{d_{n}}\rangle_{0,g}$ and $\left\langle \left\langle \tau_{d_{1}}\dots\tau_{d_{n}}\right\rangle \right\rangle _{g}$
from the the quantum and Hurwitz correlators with a $\tau_{0}$ insertion.
We deduce that it is enough to prove the equality
\begin{equation}
\langle\tau_{0}\tau_{d_{1}}\dots\tau_{d_{n}}\rangle_{0,g}=\langle\langle\tau_{0}\tau_{d_{1}}\dots\tau_{d_{n}}\rangle\rangle_{g}\label{eq:thm}
\end{equation}
in order to prove Theorem~\ref{thm: main thm}. Then, we explain
how to obtain an explicit expression for the RHS, and how to compute
the LHS. The two expressions are completely different, but can be
used to prove the equality.

\subsubsection{String Equation\label{subsec: Inverse string equation}}

Fix a nonnegative integer $g$. The correlators of the quantum Witten-Kontsevich
and the Hurwitz correlators satisfy the string equation
\begin{align*}
\langle\tau_{0}\tau_{d_{1}}\dots\tau_{d_{n}}\rangle_{0,g} & =\sum_{i=1}^{n}\langle\tau_{d_{1}}\dots\tau_{d_{i}-1}\dots\tau_{d_{n}}\rangle_{0,g},\\
\left\langle \left\langle \tau_{0}\tau_{d_{1}}\dots\tau_{d_{n}}\right\rangle \right\rangle _{g} & =\sum_{i=1}^{n}\left\langle \left\langle \tau_{d_{1}}\dots\tau_{d_{i}-1}\dots\tau_{d_{n}}\right\rangle \right\rangle _{g}.
\end{align*}
The first equation is the statement of Theorem~\ref{thm: string equation}
proved in Section~\ref{sec: string equation}. The second equation
is Proposition~3.10 in \cite{goulden2005towards}. Let us define
the following generating series
\[
\mathring{G}_{g}^{\left(q\right)}\left(s_{1},\dots,s_{n}\right):=\sum_{d_{1},\dots,d_{n}\geq0}\langle\tau_{0}\tau_{d_{1}}\dots\tau_{d_{n}}\rangle_{0,g}s_{1}^{d_{1}}\cdots s_{n}^{d_{n}}
\]
and
\[
G_{g}^{\left(q\right)}\left(s_{1},\dots,s_{n}\right):=\sum_{d_{1},\dots,d_{n}\geq0}\langle\tau_{d_{1}}\dots\tau_{d_{n}}\rangle_{0,g}s_{1}^{d_{1}}\cdots s_{n}^{d_{n}}.
\]
We also define $\mathring{G}_{g}^{H}$ and $G_{g}^{H}$ by replacing
the quantum correlators by the Hurwitz correlators. According to the
string equation, we have 
\[
\mathring{G}_{g}^{\left(q\right)}=\left(s_{1}+\cdots+s_{n}\right)G_{g}^{\left(q\right)}\text{ and }\mathring{G}_{g}^{H}=\left(s_{1}+\cdots+s_{n}\right)G_{g}^{H}.
\]
We can inverse these two equations in the same way, we then obtain
$G_{g}^{\left(q\right)}$ in terms of $\mathring{G}_{g}^{\left(q\right)}$
and $G_{g}^{H}$ in terms of $\mathring{G}_{g}^{H}$. Hence, proving
$\mathring{G}_{g}^{\left(q\right)}=\mathring{G}_{g}^{H}$ is equivalent
to prove $G_{g}^{\left(q\right)}=G_{g}^{H}$.

\subsubsection{An explicit expression for $\langle\langle\tau_{0}\tau_{d_{1}}\dots\tau_{d_{n}}\rangle\rangle_{g}$}

In \cite{goulden2005towards}, Theorem~3.1 gives the following explicit
expression for the one-part double Hurwitz numbers
\[
H_{\left(d\right),\mu}^{g}=r!d^{r-1}\left[z^{2g}\right]\frac{\prod_{i=1}^{n}S\left(\mu_{i}z\right)}{S\left(z\right)},
\]
where $d=\mu_{1}+\dots+\mu_{n}$ is the degree of the ramified cover,
$r=2g-1+n$ is the number of simple ramifications and $\left[z^{2g}\right]$ denotes the coefficient of $z^{2g}$. We also used
Notation~\ref{Notation S}, that is $S\left(z\right)=\frac{{\rm sh}\left(z/2\right)}{z/2}$.
Note that the polynomiality of the one-part double Hurwitz numbers
$H_{\left(d\right),\mu}^{g}$ in their ramifications $\mu_{1},\dots,\mu_{n}$
is clear from this expression. From the definition of the Hurwitz
correlators given in Eq.~(\ref{eq: Hurwitz Correlators}), we find
\begin{equation}
\langle\langle\tau_{0}\tau_{d_{1}}\dots\tau_{d_{n}}\rangle\rangle_{g}=\left(-1\right)^{\frac{-2+n-\sum d_{i}}{2}}\left[\mu_{1}^{d_{1}}\cdots\mu_{n}^{d_{n}}\right]\left[z^{2g}\right]\left(\mu_{1}+\cdots+\mu_{n}\right)^{2g-2+n}\frac{S\left(\mu_{1}z\right)\cdots S\left(\mu_{n}z\right)}{S\left(z\right)},\label{eq:RHS thm}
\end{equation}
where we used $S\left(0\right)=1$.

\subsubsection{Computing $\langle\tau_{0}\tau_{d_{1}}\dots\tau_{d_{n}}\rangle_{0,g}$\label{subsec: Computing quantum commutator}}

From the construction of the quantum Witten-Kontsevich series (see
Section~\ref{subsec: Construction of the quantum Witten-Kontsevich series}),
we get the following expression of its correlators
\[
\langle\tau_{0}\tau_{d_{1}}\dots\tau_{d_{n}}\rangle_{0,g}=i^{g}\left[\epsilon^{0}\hbar^{g}\right]\left.\left(\frac{\partial^{n-1}\Omega_{0,d_{1}}^{t}}{\partial t_{d_{2}}\cdots\partial t_{d_{n}}}\right)\right|_{t_{*}=0,\,u_{i}=\delta_{1,i}}=i^{g}\left[\epsilon^{0}\hbar^{g}\right]\frac{1}{\hbar^{n-1}}\left[\cdots\left[\Omega_{0,d_{1}},\overline{H}_{d_{2}}\right],\dots,H_{d_{n}}\right]\Big\vert_{u_{i}=\delta_{1,i}}.
\]
Lemma~\ref{lem: Omeg_0,d H_d-1} gives $\Omega_{0,d}=H_{d-1}$. Thus
we have to study
\begin{equation}
\langle\tau_{0}\tau_{d_{1}}\dots\tau_{d_{n}}\rangle_{0,g}=i^{g}\left[\hbar^{g+n-1}\right]\left.\left[\cdots\left[H_{d_{1-1}},\overline{H}_{d_{2}}\right],\dots,\overline{H}_{d_{n}}\right]\right|_{u_{i}=\delta_{1,i},\epsilon=0}.\label{eq:LHS thm}
\end{equation}

We will compute this expression in the proof. To do so, we need a
computable expression of $H_{d}$ and a computable expression of the
commutator. We give these expressions in the next two paragraphs.
The computation will be carried out in $\tilde{\mathcal{A}}$. We
will then perform the substitution $u_{i}=\delta_{i,1}$ using Lemma~\ref{lem:evaluation }.

\paragraph{A computable expression of $H_{p}$.}

In \cite{buryak2015integrals}, Theorem $1$ gives an explicit expression
for the intersection number of a ${\rm DR}$-cycle with the maximal
power of a $\psi$-class. From this theorem we get
\[
\int_{DR_{g}\left(0,a_{1},\dots,a_{m},-\sum a_{i}\right)}\psi_{1}^{p+1}=\delta_{p+1,2g-1+m}\left[z^{2g}\right]\frac{S\left(a_{1}z\right)\cdots S\left(a_{m}z\right)S\left(\sum_{i=1}^{m}a_{i}z\right)}{S\left(z\right)}.
\]
We then obtain from the definition of $H_{p}$ given by Eq.~(\ref{eq: Def Hamiltoniens H_d}),
\begin{align}
H_{p}\left(x\right)\Big\rvert_{\epsilon=0} & =\sum_{\overset{g\geq0,m\geq0}{2g+m>0}}\frac{\left(i\hbar\right)^{g}}{m!}\sum_{a_{1},\dots,a_{m}\in\mathbb{Z}}\left(\int_{{\rm DR}_{g}\left(0,a_{1},\dots,a_{m},-\sum a_{i}\right)}\psi_{1}^{d+1}\Lambda\left(\frac{-\epsilon^{2}}{i\hbar}\right)\right)p_{a_{1}}\cdots p_{a_{m}}e^{ix\sum_{i=1}^{m}a_{i}}\nonumber \\
 & =\sum_{g\geq0}\frac{\left(i\hbar\right)^{g}}{m!}\sum_{a_{1},\dots,a_{m}\in\mathbb{Z}}\left(\left[z^{2g}\right]\frac{S\left(a_{1}z\right)\dots S\left(a_{m}z\right)S\left(\sum_{i=1}^{m}a_{i}z\right)}{S\left(z\right)}\right)p_{a_{1}}\cdots p_{a_{m}}e^{ix\sum_{i=1}^{m}a_{i}},\label{eq:Expression Hp}
\end{align}
with $m=p+2-2g$. 

\paragraph{Explicit expression of the star product.}

We will use the following expression of the star product. 
\begin{prop}
Let $f,g\in\mathcal{F}^{\hbar}\left(P\right)$, we have 
\begin{equation}
f\star g=f\exp\left(\sum_{k>0}i\hbar k\overleftarrow{\frac{\partial}{\partial p_{k}}}\overrightarrow{\frac{\partial}{\partial p_{-k}}}\right)g.\label{eq: explicit expr star pdt}
\end{equation}
The notations $\overleftarrow{\frac{\partial}{\partial p_{k}}}$ and
$\overrightarrow{\frac{\partial}{\partial p_{-k}}}$ mean that the
derivative acts on the left or on the right, that is $f\star g=fg+\sum_{k>0}i\hbar k\frac{\partial f}{\partial p_{k}}\frac{\partial g}{\partial p_{-k}}+\sum_{k_{1},k_{2}>0}\frac{\left(i\hbar\right)^{2}}{2}k_{1}k_{2}\frac{\partial^{2}f}{\partial p_{k_{1}}\partial p_{k_{2}}}\frac{\partial^{2}g}{\partial p_{-k_{1}}\partial p_{-k_{2}}}+\cdots$.
\end{prop}
This property is a consequence of Wick's theorem. To be exhaustive, we give a short proof. 
\begin{proof}
Recall from Section~\ref{subsec:Star-product-and} that the star product is defined by $:f\star g:=:f::g:$
in $W\left[\left[\epsilon,\hbar\right]\right]$. In order to obtain
the explicit expression of $f\star g$ in Eq.~(\ref{eq: explicit expr star pdt}), we must write $:f::g:$ as
a sum of normally ordered terms. To do so, we commute the $p_{<0}$'s
of $:g:$ to the left of $:f:$ using the commutation relation $\left[p_{a},p_{b}\right]=i\hbar a\delta_{a+b,0}$.
First, if all the $p_{<0}$'s of $:g:$ commute with the $p_{>0}$'s
of $:f:$, we obtain the term $:fg:$. Then, if one contraction between
a $p_{<0}$ of $:g:$ with a $p_{>0}$ of $:f:$ occurs, we obtain
the term $:\sum_{k>0}i\hbar k\frac{\partial f}{\partial p_{k}}\frac{\partial g}{\partial p_{-k}}:$.
Similarly, if $q\geq1$ contractions occur, we obtain the term $:\sum_{k_{1},\dots,k_{q}>0}\frac{\left(i\hbar\right)^{q}}{q!}k_{1}\cdots k_{q}\frac{\partial^{q}f}{\partial p_{k_{1}}\cdots\partial p_{k_{q}}}\frac{\partial^{q}g}{\partial p_{-k_{1}}\cdots\partial p_{-k_{q}}}:$.
\end{proof} 

\subsection{Proof of the equality $\langle\tau_{0}\tau_{d_{1}}\dots\tau_{d_{n}}\rangle_{0,g}=\langle\langle\tau_{0}\tau_{d_{1}}\dots\tau_{d_{n}}\rangle\rangle_{g}$}

In Section~\ref{subsec:Trivial-case}, we prove the equality of Eq.~(\ref{eq:thm})
for $n=1$. 

In Section~\ref{subsec:Less-trivial-case}, we prove the equality
of Eq.~(\ref{eq:thm}) for $n=2$. This particular case is included
as an example to illuminate the proof of the general case. 

In Section~\ref{subsec:General-proof}, we prove the equality of
Eq.~(\ref{eq:thm}) for $n\geq2$.

\paragraph{Convention. }

In the rest of the proof, we focus on the restriction $\epsilon=0$.
Hence, we forget the formal variable $\epsilon$ and always suppose
it to be zero. 

\subsubsection{Proof for $n=1$\label{subsec:Trivial-case}}

Fix two nonnegative integers $d$ and $g$. We prove in this section
that
\[
\langle\tau_{0}\tau_{d}\rangle_{0,g}=\langle\langle\tau_{0}\tau_{d}\rangle\rangle_{g}.
\]
We start from the LHS, $\langle\tau_{0}\tau_{d}\rangle_{0,g}$. Recall
that Eq.~(\ref{eq:LHS thm}) gives $\langle\tau_{0}\tau_{d}\rangle_{g}=i^{g}\left[\hbar^{g}\right]\,H_{d-1}\Bigr|_{u_{i}=\delta_{i,1}}$.
Then we use the expression of $H_{p}$ in Eq.~(\ref{eq:Expression Hp})
and perform the evaluation $u_{i}=\delta_{i,1}$ with Lemma~\ref{lem:evaluation }.
We find
\[
\langle\tau_{0}\tau_{d}\rangle_{0,g}=i^{g}\frac{i^{g}}{m!}\left(-i\right)^{m}\left[a_{1}\cdots a_{m}\right]\left[z^{2g}\right]\frac{S\left(a_{1}z\right)\cdots S\left(a_{m}z\right)S\left(\sum_{i=1}^{m}a_{i}z\right)}{S\left(z\right)}
\]
with $m=d+1-2g$. Note that there is no $a$-linear term in $S\left(az\right)$
because $S$ is an even function, hence the expression of $\langle\tau_{0}\tau_{d}\rangle_{0,g}$
simplifies to
\[
\langle\tau_{0}\tau_{d}\rangle_{0,g}=\frac{\left(-i\right)^{d+1}}{m!}\left[a_{1}\cdots a_{m}\right]\left[z^{2g}\right]\frac{S\left(\sum_{i=1}^{m}a_{i}z\right)}{S\left(z\right)}.
\]
Let $A=\sum_{i=1}^{m}a_{i}$. It is easy to check $\left[a_{1}\cdots a_{m}\right]S\left(\sum_{i=1}^{m}a_{i}z\right)=m!\left[A^{m}\right]S\left(Az\right)$.
Hence we get
\[
\langle\tau_{0}\tau_{d}\rangle_{0,g}=\left(-i\right)^{d+1}\left[z^{2g}A^{m}\right]\frac{S\left(Az\right)}{S\left(z\right)}=\left(-i\right)^{d+1}\left[z^{2g}A^{d+1-2g}\right]\frac{S\left(Az\right)}{S\left(z\right)},
\]
and we rewrite it as
\[
\langle\tau_{0}\tau_{d}\rangle_{0,g}=i^{-1-d}\left[A^{d}\right]\left[z^{2g}\right]A^{2g-1}\frac{S\left(Az\right)}{S\left(z\right)}.
\]
We recognize the expression of $\langle\langle\tau_{0}\tau_{d}\rangle\rangle_{g}$
given by Eq.~(\ref{eq:RHS thm}).

\subsubsection{Proof for $n=2$\label{subsec:Less-trivial-case}}

Fix three nonnegative integers $d_{1},d_{2}$ and $g$. We prove in
this section that
\[
\langle\tau_{0}\tau_{d_{1}}\tau_{d_{2}}\rangle_{0,g}=\langle\langle\tau_{0}\tau_{d_{1}}\tau_{d_{2}}\rangle\rangle_{g}.
\]
We start from the LHS. Recall that Eq.~(\ref{eq:LHS thm}) gives
$\langle\tau_{0}\tau_{d_{1}}\tau_{d_{2}}\rangle_{0,g}=\left[\hbar^{g}\right]\frac{i^{g}}{\hbar}\left.\left[H_{d_{1}},\overline{H}_{d_{2}}\right]\right|_{u_{i}=\delta_{1,i}}$.

In Step~$1$, we obtain an expression of $\frac{i^{g}}{\hbar}\left[H_{d_{1-1}},\overline{H}_{d_{2}}\right]$
using the formulas of $H_{d_{1}-1}$, $\overline{H}_{d_{2}}$ and
of the star product given in Section~\ref{subsec: Computing quantum commutator}. 

In Step~$2$, we first extract the coefficient of  $\hbar^{g}$ in
$\frac{i^{g}}{\hbar}\left[H_{d_{1-1}},\overline{H}_{d_{2}}\right]$.
Then we perform the substitution $u_{i}=\delta_{i,1}$ in $\left[\hbar^{g}\right]\frac{i^{g}}{\hbar}\left[H_{d_{1-1}},\overline{H}_{d_{2}}\right]$
and get a first expression of $\langle\tau_{0}\tau_{d_{1}}\tau_{d_{2}}\rangle_{0,g}$.
This expression will be totally different from the expression of $\langle\langle\tau_{0}\tau_{d_{1}}\tau_{d_{2}}\rangle\rangle_{g}$
given by Eq.~(\ref{eq:RHS thm}).

In Steps $3$,$4$ and $5$, we will transform this last expression
of $\langle\tau_{0}\tau_{d_{1}}\tau_{d_{2}}\rangle_{0,g}$ into the
expression of $\langle\langle\tau_{0}\tau_{d_{1}}\tau_{d_{2}}\rangle\rangle_{g}$
given by Eq.~(\ref{eq:RHS thm}). 

\paragraph{Step 1.}

We compute $\frac{i^{g}}{\hbar}\left[H_{d_{1-1}},\overline{H}_{d_{2}}\right]=\frac{i^{g}}{\hbar}\left(H_{d_{1-1}}\star\overline{H}_{d_{2}}-\overline{H}_{d_{2}}\star H_{d_{1-1}}\right)$
. Recall that
\[
H_{d_{1}-1}\left(x\right)=\sum_{\underset{\text{with}\,m_{1}=d_{1}+1-2g_{1}}{g_{1}\geq0}}\frac{\left(i\hbar\right)^{g_{1}}}{m_{1}!}\sum_{a_{1},\dots,a_{m_{1}}\in\mathbb{Z}}\left(\left[z^{2g_{1}}\right]\frac{S\left(a_{1}z\right)\cdots S\left(a_{m_{1}}z\right)S\left(\sum_{i=1}^{m_{1}}a_{i}z\right)}{S\left(z\right)}\right)p_{a_{1}}\cdots p_{a_{m_{1}}}e^{ix\sum_{i=1}^{m_{1}}a_{i}}
\]
and
\[
\overline{H}_{d_{2}}=\sum_{\underset{\text{with}\,m_{2}=d_{2}+2-2g_{2}}{g_{2}\geq0}}\frac{\left(i\hbar\right)^{g_{2}}}{m_{2}!}\sum_{\underset{\sum b_{i}=0}{b_{1},\dots,b_{m_{2}}\in\mathbb{Z}}}\left(\left[w^{2g_{2}}\right]\frac{S\left(b_{1}w\right)\cdots S\left(b_{m_{2}}w\right)}{S\left(w\right)}\right)p_{b_{1}}\cdots p_{b_{m_{2}}}.
\]
The expression of $\overline{H}_{d_{2}}$ is obtained by a formal
$x$-integration along $S^{1}$ of $H_{d_{2}}$. Hence this imposes
the condition $\sum_{i=1}^{m_{2}}b_{i}=0$ and then $S\left(\sum_{i=1}^{m_{2}}b_{i}w\right)=1$. 

From the expression of the star product~(\ref{eq: explicit expr star pdt}),
we get
\begin{align*}
H_{d_{1}-1}\star\overline{H}_{d_{2}} & =H_{d_{1}-1}\exp\left(\sum_{k>0}i\hbar k\overleftarrow{\frac{\partial}{\partial p_{k}}}\overrightarrow{\frac{\partial}{\partial p_{-k}}}\right)\overline{H}_{d_{2}}\\
 & =H_{d_{1}-1}\sum_{q\geq0}\frac{\left(i\hbar\right)^{q}}{q!}\left(\sum_{k_{1},\dots,k_{q}>0}k_{1}\cdots k_{q}\overleftarrow{\frac{\partial}{\partial p_{k_{1}}}}\cdots\overleftarrow{\frac{\partial}{\partial p_{k_{q}}}}\overrightarrow{\frac{\partial}{\partial p_{-k_{1}}}}\cdots\overrightarrow{\frac{\partial}{\partial p_{-k_{q}}}}\right)\overline{H}_{d_{2}}.
\end{align*}
We first describe the action of the left-derivatives of the star product
on $H_{d_{1}-1}$. Fix $g_{1}$ in $H_{d_{1}-1}$. The product of
$q$ left-derivatives $\overleftarrow{\frac{\partial}{\partial p_{k_{1}}}}\cdots\overleftarrow{\frac{\partial}{\partial p_{k_{q}}}}$
acts on the formal power series \\
$\sum_{a_{1},\dots,a_{m_{1}}\in\mathbb{Z}}S\left(a_{1}z\right)\cdots S\left(a_{m_{1}}z\right)S\left(\sum_{i=1}^{m_{1}}a_{i}z\right)p_{a_{1}}\cdots p_{a_{m_{1}}}$
yielding 
\[
m_{1}\cdots\left(\tilde{m}_{1}+1\right)\sum_{a_{1},\dots,a_{\tilde{m}_{1}}\in\mathbb{Z}}S\left(a_{1}z\right)\cdots S\left(a_{\tilde{m}_{1}}z\right)S\left(k_{1}z\right)\cdots S\left(k_{q}z\right)S\left(\left(\tilde{A}+K\right)z\right)p_{a_{1}}\cdots p_{a_{\tilde{m}_{1}}}
\]
where $\tilde{m}_{1}=m_{1}-q$, $\tilde{A}=\sum_{i=1}^{\tilde{m}_{1}}a_{i}$
and $K=\sum_{i=1}^{q}k_{i}$. Indeed, each derivative $\frac{\partial}{\partial p_{k_{j}}}$
may act on each factor $p_{a_{i}}$. Without loss of generality we
can assume that $i=\tilde{m}_{1}+j$, multiplying the result by $m_{1}\cdots\left(\tilde{m}_{1}+1\right)$
to account for the number of equivalent choices. The derivative yields
a nonvanishing result if and only if $a_{i}=k_{j}.$

Similarly we describe the action of the right-derivatives of the star
product on $H_{d_{2}}$. Fix $g_{2}$ in $H_{d_{2}}$. The product
of $q$ right-derivatives $\overrightarrow{\frac{\partial}{\partial p_{-k_{1}}}}\cdots\overrightarrow{\frac{\partial}{\partial p_{-k_{q}}}}$
acts on the formal series $\sum_{\underset{\sum b_{i}=0}{b_{1},\dots,b_{m_{2}}\in\mathbb{Z}}}S\left(b_{1}w\right)\cdots S\left(b_{m_{2}}w\right)p_{b_{1}}\cdots p_{b_{m_{2}}}$
yielding
\[
m_{2}\cdots\left(\tilde{m}_{2}+1\right)\sum_{\underset{\tilde{B}=K}{b_{1},\dots,b_{\tilde{m}_{2}}\in\mathbb{Z}}}S\left(b_{1}w\right)\cdots S\left(b_{\tilde{m}_{2}}w\right)S\left(-k_{1}w\right)\cdots S\left(-k_{q}w\right)p_{b_{1}}\cdots p_{b_{\tilde{m}_{2}}}
\]
where $\tilde{m}_{2}=m_{2}-q$ and $\tilde{B}=\sum_{i=1}^{\tilde{m}_{2}}b_{i}$. 

Recall that $S$ is an even function, hence $S\left(-k_{i}w\right)=S\left(k_{i}w\right)$.
Note that the condition $\sum_{i=1}^{m}b_{i}=0$ becomes $K=\tilde{B}$.

Finally, the expression of $H_{d_{1}-1}\star\overline{H}_{d_{2}}$
becomes
\begin{align}
 & \sum_{g\geq0}\sum_{g_{1}+g_{2}+q=g}\frac{\left(i\hbar\right)^{g}}{\tilde{m}_{1}!\tilde{m}_{2}!q!}\left[z^{2g_{1}}w^{2g_{2}}\right]\label{eq: H_d_1 star H_d_2}\\
 & \times\sum_{a_{1},\dots,a_{\tilde{m}_{1}}\in\mathbb{Z}}\sum_{b_{1},\dots,b_{\tilde{m}_{2}}\in\mathbb{Z}}\sum_{\underset{K=\tilde{B}}{k_{1},\dots,k_{q}>0}}k_{1}\cdots k_{q}\nonumber \\
 & \times\frac{S\left(a_{1}z\right)\cdots S\left(a_{\tilde{m}_{1}}z\right)S\left(k_{1}z\right)\cdots S\left(k_{q}z\right)S\left(\left(\tilde{A}+\tilde{B}\right)z\right)}{S\left(z\right)}\frac{S\left(b_{1}w\right)\cdots S\left(b_{\tilde{m}_{2}}w\right)S\left(-k_{1}w\right)\cdots S\left(-k_{q}w\right)}{S\left(w\right)}\nonumber \\
 & \times p_{a_{1}}\cdots p_{a_{\tilde{m}_{1}}}p_{b_{1}}\cdots p_{b_{\tilde{m}_{2}}}e^{ix\left(\tilde{A}+\tilde{B}\right)},\nonumber 
\end{align}
where $\tilde{m_{1}}=d_{1}+1-2g_{1}-q$ and $\tilde{m}_{2}=d_{2}+2-2g_{2}-q$.

We can re-do this exercise to compute $\overline{H}_{d_{2}}\star H_{d_{1}-1}$.
The main difference is the condition $K=\tilde{B}$ which becomes
$K=-\tilde{B}$. Thus, the expression of $\frac{i^{g}}{\hbar}\left[H_{d_{1-1}},\overline{H}_{d_{2}}\right]$
is 
\begin{align}
 & \sum_{g\geq0}\sum_{\underset{g_{1},g_{2}\geq0,\,q\geq1}{g_{1}+g_{2}+q-1=g}}\frac{i^{2g+1}\hbar^{g}}{\tilde{m}_{1}!\tilde{m}_{2}!q!}\left[z^{2g_{1}}w^{2g_{2}}\right]\label{eq: n=00003D2 Step 1}\\
 & \times\sum_{a_{1},\dots,a_{\tilde{m}_{1}}\in\mathbb{Z}}\sum_{b_{1},\dots,b_{\tilde{m}_{2}}\in\mathbb{Z}}\frac{S\left(a_{1}z\right)\cdots S\left(a_{\tilde{m}_{1}}z\right)S\left(\left(\tilde{A}+\tilde{B}\right)z\right)S\left(b_{1}w\right)\cdots S\left(b_{\tilde{m}_{2}}w\right)}{S\left(z\right)S\left(w\right)}\nonumber \\
 & \times\left(\sum_{k_{1}+\cdots+k_{q}=\tilde{B}}k_{1}\cdots k_{q}S\left(k_{1}z\right)\cdots S\left(k_{q}z\right)S\left(k_{1}w\right)\cdots S\left(k_{q}w\right)\right.\nonumber \\
 & \,\,\,\,\,-\left.\sum_{k_{1}+\cdots+k_{q}=-\tilde{B}}k_{1}\cdots k_{q}S\left(k_{1}z\right)\cdots S\left(k_{q}z\right)S\left(k_{1}w\right)\cdots S\left(k_{q}w\right)\right),\nonumber \\
 & \times p_{a_{1}}\cdots p_{a_{\tilde{m}_{1}}}p_{b_{1}}\cdots p_{b_{\tilde{m}_{2}}}e^{ix\left(\tilde{A}+\tilde{B}\right)}.\nonumber 
\end{align}
where $\tilde{m_{1}}=d_{1}+1-2g_{1}-q$ and $\tilde{m}_{2}=d_{2}+2-2g_{2}-q$. 
\begin{rem}
The term $q=0$ of the expression of $H_{d_{1}-1}\star\overline{H}_{d_{2}}$
in Eq.~(\ref{eq: H_d_1 star H_d_2}) corresponds to the commutative
part of the star product. This term disappears in the commutator $\frac{1}{\hbar}\left[H_{d_{1}-1},\overline{H}_{d_{2}}\right]$.
We can then suppose that $q\geq1$ in Eq.~(\ref{eq: n=00003D2 Step 1}). 
\end{rem}

\paragraph{Change of notation.}

For convenience, we change the notation by removing the tildes, i.e.
we set $m_{1}:=\tilde{m}_{1}$, $m_{2}:=\tilde{m}_{2}$, $A:=\tilde{A}$
and $B:=\tilde{B}$.

\paragraph{Step 2.}

We first extract the coefficient of $\hbar^{g}$ from $\frac{i^{g}}{\hbar}\left[H_{d_{1-1}},\overline{H}_{d_{2}}\right]$
. Then we evaluate this coefficient, which is a differential polynomial,
at $u_{i}=\delta_{i,1}$. We will then get an expression for $\langle\tau_{0}\tau_{d_{1}}\tau_{d_{2}}\rangle_{0,g}=\left[\hbar^{g}\right]\frac{i^{g}}{\hbar}\left.\left[H_{d_{1}},\overline{H}_{d_{2}}\right]\right|_{u_{i}=\delta_{1,i}}$.

We extract the coefficient of $\hbar^{g}$ in $\frac{i^{g}}{\hbar}\left[H_{d_{1-1}},\overline{H}_{d_{2}}\right]$
from its expression obtained in Eq.~(\ref{eq: n=00003D2 Step 1}).
This only removes the summation over $g$. With our new notations,
this coefficient is
\begin{align*}
\left[\hbar^{g}\right]\frac{i^{g}}{\hbar}\left[H_{d_{1-1}},\overline{H}_{d_{2}}\right]= & \sum_{\underset{g_{1},g_{2}\geq0,\,q\geq1}{g_{1}+g_{2}+q-1=g}}\frac{i^{2g+1}}{m_{1}!m_{2}!q!}\left[z^{2g_{1}}w^{2g_{2}}\right]\\
 & \times\sum_{a_{1},\dots,a_{m_{1}}\in\mathbb{Z}}\sum_{b_{1},\dots,b_{m_{2}}\in\mathbb{Z}}\frac{S\left(a_{1}z\right)\cdots S\left(a_{m_{1}}z\right)S\left(\left(A+B\right)z\right)S\left(b_{1}w\right)\cdots S\left(b_{m_{2}}w\right)}{S\left(z\right)S\left(w\right)}\\
 & \times\left(\sum_{k_{1}+\cdots+k_{q}=B}k_{1}\cdots k_{q}S\left(k_{1}z\right)\cdots S\left(k_{q}z\right)S\left(k_{1}w\right)\cdots S\left(k_{q}w\right)\right.\\
 & \,\,\,\,\,-\left.\sum_{k_{1}+\cdots+k_{q}=-B}k_{1}\cdots k_{q}S\left(k_{1}z\right)\cdots S\left(k_{q}z\right)S\left(k_{1}w\right)\cdots S\left(k_{q}w\right)\right)\\
 & \times p_{a_{1}}\cdots p_{a_{m_{1}}}p_{b_{1}}\cdots p_{b_{m_{2}}}e^{ix\left(A+B\right)},
\end{align*}
where $m_{1}=d_{1}+1-2g_{1}-q$ and $m_{2}=d_{2}+2-2g_{2}-q$. 

This last expression is a differential polynomial thanks to Proposition~\ref{prop: Stability commutator}.
In order to substitute $u_{i}=\delta_{i,1}$, we use Lemma~\ref{lem:evaluation }.
We get
\begin{align}
\langle\tau_{0}\tau_{d_{1}}\tau_{d_{2}}\rangle_{0,g}= & \sum_{\underset{g_{1},g_{2}\geq0,\,q\geq1}{g_{1}+g_{2}+q-1=g}}\left[z^{2g_{1}}w^{2g_{2}}\right]\left[a_{1}\ldots a_{m_{1}}b_{1}\ldots b_{m_{2}}\right]\nonumber \\
 & \times\frac{i^{-d_{1}-d_{2}}}{m_{1}!m_{2}!q!}\frac{S\left(a_{1}z\right)\cdots S\left(a_{m_{1}}z\right)S\left(\left(A+B\right)z\right)S\left(b_{1}w\right)\cdots S\left(b_{m_{2}}w\right)}{S\left(z\right)S\left(w\right)}\label{eq:step 2}\\
 & \times\left(\sum_{k_{1}+\cdots+k_{q}=B}k_{1}\cdots k_{q}S\left(k_{1}z\right)\cdots S\left(k_{q}z\right)S\left(k_{1}w\right)\cdots S\left(k_{q}w\right)\right.\nonumber \\
 & \,\,\,\,\,-\left.\sum_{k_{1}+\cdots+k_{q}=-B}k_{1}\cdots k_{q}S\left(k_{1}z\right)\cdots S\left(k_{q}z\right)S\left(k_{1}w\right)\cdots S\left(k_{q}w\right)\right),\nonumber 
\end{align}
where $m_{1}=d_{1}+1-2g_{1}-q$ and $m_{2}=d_{2}+2-2g_{2}-q$.
\begin{rem}
\label{rem: Ehrhart}It can look confusing that in this expression,
$b_{i}$ stands for a formal variable and an integer when we write
$k_{1}+\cdots+k_{q}=B=\sum_{i=1}^{m_{2}}b_{i}$. This is due to the
the presence of Ehrhart polynomials. Indeed, the coefficient of any
power of $z$ and $w$ in 
\[
\sum_{k_{1}+\cdots+k_{q}=B}k_{1}\cdots k_{q}S\left(k_{1}z\right)\cdots S\left(k_{q}z\right)S\left(k_{1}w\right)\cdots S\left(k_{q}w\right)
\]
is an Ehrhart polynomial in the indeterminate $B=\sum_{i=1}^{m_{2}}b_{i}$,
see \cite[Lemma A.1]{BuryakRossi2016} for a proof. Hence, when we
write $B$ as an integer and use this lemma to justify that $B$ can
also be used as a formal variable. The same phenomenon applies to
$\sum_{k_{1}+\cdots+k_{q}=-B}k_{1}\cdots k_{q}S\left(k_{1}z\right)\cdots S\left(k_{q}z\right)S\left(k_{1}w\right)\cdots S\left(k_{q}w\right)$
when $B<0$.
\end{rem}

\paragraph{Plan of Steps $3,4$ and $5$.}

This expression of $\langle\tau_{0}\tau_{d_{1}}\tau_{d_{2}}\rangle_{0,g}$
is completely different from the one of $\langle\langle\tau_{0}\tau_{d_{1}}\tau_{d_{2}}\rangle\rangle_{g}$
given by Eq.~(\ref{eq:RHS thm}). Moreover, it is difficult to compute
the number $\langle\tau_{0}\tau_{d_{1}}\tau_{d_{2}}\rangle_{0,g}$
from this expression. Let us point out the difficulties. The three
last lines of Eq.~(\ref{eq:step 2}) form a series depending on the
parameters $d_{1},d_{2},g_{1},g_{2},q$ in the indeterminates $a_{1},\dots,a_{m_{1}},b_{1},\dots,b_{m_{2}},z,w$.
We denote this series by $F_{d_{1},d_{2},g_{1},g_{2},q}\left(a_{1},\dots,a_{m_{1}},b_{1},\dots,b_{m_{2}},z,w\right)$
so that Eq.~(\ref{eq:step 2}) becomes
\[
\langle\tau_{0}\tau_{d_{1}}\tau_{d_{2}}\rangle_{0,g}=\sum_{\underset{g_{1},g_{2}\geq0,\,q\geq1}{g_{1}+g_{2}+q-1=g}}\left[z^{2g_{1}}w^{2g_{2}}\right]\left[a_{1}\ldots a_{m_{1}}b_{1}\ldots b_{m_{2}}\right]F_{d_{1},d_{2},g_{1},g_{2},q}\left(a_{1},\dots,a_{m_{1}},b_{1},\dots,b_{m_{2}},z,w\right).
\]
Hence, for each choice of the parameters $d_{1},d_{2},g_{1},g_{2},q$,
we have to extract the coefficient of \\
$z^{2g_{1}}w^{2g_{2}}a_{1}\ldots a_{m_{1}}b_{1}\ldots b_{m_{2}}$
in $F_{d_{1},d_{2},g_{1},g_{2},q}$. Then we sum these coefficients
over the parameters $g_{1},g_{2},q$. The main difficulty is to extract
the coefficient of $b_{i}$ from the expression appearing in parenthesis
in $F_{d_{1},d_{2},g_{1},g_{2},q}$, that is from
\begin{align*}
 & \sum_{k_{1}+\cdots+k_{q}=B}k_{1}\cdots k_{q}S\left(k_{1}z\right)\cdots S\left(k_{q}z\right)S\left(k_{1}w\right)\cdots S\left(k_{q}w\right)\\
 & \,\,\,\,\,-\sum_{k_{1}+\cdots+k_{q}=-B}k_{1}\cdots k_{q}S\left(k_{1}z\right)\cdots S\left(k_{q}z\right)S\left(k_{1}w\right)\cdots S\left(k_{q}w\right).
\end{align*}
As we explained in Remark~\ref{rem: Ehrhart}, the coefficient of
any power of $z$ and $w$ in each of the two sums is an Ehrhart polynomial
in the indeterminate $B=\sum_{i=1}^{m_{2}}b_{i}$. However we do not
have an explicit expression for the coefficients of these Ehrhart
polynomials. Luckily, Eulerian numbers appear in the computation of
these coefficients, see Remark~\ref{rem: Why Eulerian appear}. The
plan is to modify our expression of $\langle\tau_{0}\tau_{d_{1}}\tau_{d_{2}}\rangle_{0,g}$
in order to use known properties of Eulerian numbers. 

In Step~$3$, we will modify our expression of $\langle\tau_{0}\tau_{d_{1}}\tau_{d_{2}}\rangle_{0,g}$
using simplifications arising from extracting the coefficient of $a_{1}\ldots a_{m_{1}}b_{1}\ldots b_{m_{2}}$
in $F_{d_{1},d_{2},g_{1},g_{2},q}\left(a_{1},\dots,a_{m_{1}},b_{1},\dots,b_{m_{2}},z,w\right)$.
These simplifications mainly come from the fact that $z\rightarrow S\left(z\right)$
is even.

In Step~$4$, we use some changes of variables in order to get $\langle\tau_{0}\tau_{d_{1}}\tau_{d_{2}}\rangle_{0,g}$
as the coefficient of an exponential generating series. This is the
exponential of an expression that can be computed using Eulerian numbers. 

In Step~$5$, we finally we use a property of Eulerian number in
order to get a simplified expression of $\langle\tau_{0}\tau_{d_{1}}\tau_{d_{2}}\rangle_{0,g}$.
We then recover from it the expression of $\langle\langle\tau_{0}\tau_{d_{1}}\tau_{d_{2}}\rangle\rangle_{g}$
given by Eq.~(\ref{eq:RHS thm}).

\paragraph{Step 3.}

The evaluation $u_{i}=\delta_{i,1}$ brings many simplifications that
we explain now.
\begin{itemize}
\item First recall that $S$ is an even power series so that the coefficient
of $\alpha$ in $S\left(\alpha z\right)\times F\left(\alpha\right)$
where $F$ is a formal power series in $\alpha$ is the coefficient
of $\alpha$ in $F\left(\alpha\right)$. Hence, Expression~(\ref{eq:step 2})
simplifies as
\begin{align}
\langle\tau_{0}\tau_{d_{1}}\tau_{d_{2}}\rangle_{0,g}= & \sum_{\underset{g_{1},g_{2}\geq0,\,q\geq1}{g_{1}+g_{2}+q-1=g}}\frac{i^{-d_{1}-d_{2}}}{m_{1}!m_{2}!q!}\left[z^{2g_{1}}w^{2g_{2}}\right]\nonumber \\
 & \times\left[a_{1}\ldots a_{m_{1}}b_{1}\ldots b_{m_{2}}\right]\frac{S\left(\left(A+B\right)z\right)}{S\left(z\right)S\left(w\right)}\label{eq:step 2-1}\\
 & \times\left(\sum_{k_{1}+\cdots+k_{q}=B}k_{1}\cdots k_{q}S\left(k_{1}z\right)\cdots S\left(k_{q}z\right)S\left(k_{1}w\right)\cdots S\left(k_{q}w\right)\right.\nonumber \\
 & \,\,\,\,\,-\left.\sum_{k_{1}+\cdots+k_{q}=-B}k_{1}\cdots k_{q}S\left(k_{1}z\right)\cdots S\left(k_{q}z\right)S\left(k_{1}w\right)\cdots S\left(k_{q}w\right)\right).\nonumber 
\end{align}
\item Fix $g_{1},g_{2}$ and $q$ in Expression~(\ref{eq:step 2-1}) so
that $m_{1}$ and $m_{2}$ are fixed. We extract the coefficient of
$a_{1}\ldots a_{m_{1}}b_{1}\ldots b_{m_{2}}$ from a power series
that only depends of the sums $A=a_{1}+\ldots+a_{m_{1}}$ and $B=b_{1}+\ldots+b_{m_{2}}$.
This is equivalent to extracting the coefficient of $\frac{A^{m_{1}}B^{m_{2}}}{m_{1}!m_{2}!}$
from the same power series.
\item Consider the expression in parenthesis in Eq.~(\ref{eq:step 2-1}).
This is a power series in the indeterminate $B$. Moreover, when $B\geq0$
this power series becomes
\[
\sum_{k_{1}+\cdots+k_{q}=B}k_{1}\cdots k_{q}S\left(k_{1}z\right)\cdots S\left(k_{q}z\right)S\left(k_{1}w\right)\cdots S\left(k_{q}w\right)
\]
and when $B<0$ it becomes
\[
-\sum_{k_{1}+\cdots+k_{q}=-B}k_{1}\cdots k_{q}S\left(k_{1}z\right)\cdots S\left(k_{q}z\right)S\left(k_{1}w\right)\cdots S\left(k_{q}w\right).
\]
We are interested in the coefficients of this power series, so for
simplicity we can suppose $B>0$ . Expression~(\ref{eq:step 2-1})
becomes
\begin{align}
\langle\tau_{0}\tau_{d_{1}}\tau_{d_{2}}\rangle_{0,g}= & \sum_{\underset{g_{1},g_{2}\geq0,\,q\geq1}{g_{1}+g_{2}+q-1=g}}\frac{i^{-d_{1}-d_{2}}}{q!}\left[z^{2g_{1}}w^{2g_{2}}\right]\label{eq:step 2-2}\\
 & \times\left[A^{m_{1}}B^{m_{2}}\right]\frac{S\left(\left(A+B\right)z\right)}{S\left(z\right)S\left(w\right)}\nonumber \\
 & \times\sum_{k_{1}+\cdots+k_{q}=B}k_{1}\cdots k_{q}S\left(k_{1}z\right)\cdots S\left(k_{q}z\right)S\left(k_{1}w\right)\cdots S\left(k_{q}w\right).\nonumber 
\end{align}
\end{itemize}

\paragraph{Step 4.}

We perform the changes of variables $z:=Az$ and $w:=Bw$ in Expression~(\ref{eq:step 2-2}).
Recall that $m_{1}=d_{1}+1-2g_{1}-q$ and $m_{2}=d_{2}+2-2g_{2}-q$,
these changes of variables yield
\begin{align*}
\langle\tau_{0}\tau_{d_{1}}\tau_{d_{2}}\rangle_{0,g}= & \sum_{\underset{g_{1},g_{2}\geq0,\,q\geq1}{g_{1}+g_{2}+q-1=g}}\frac{i^{-d_{1}-d_{2}}}{q!}\left[z^{2g_{1}}w^{2g_{2}}A^{d_{1}+1-q}B^{d_{2}+2-q}\right]\\
 & \times\frac{S\left(\left(A+B\right)Az\right)}{S\left(Az\right)S\left(Bw\right)}\\
 & \times\sum_{k_{1}+\cdots+k_{q}=B}k_{1}\cdots k_{q}S\left(k_{1}Az\right)\cdots S\left(k_{q}Az\right)S\left(k_{1}Bw\right)\cdots S\left(k_{q}Bw\right).
\end{align*}
We re-write this as
\begin{align}
 & \sum_{q=1}^{g+1}\frac{i^{-d_{1}-d_{2}}}{q!}\left[A^{d_{1}+1-q}B^{d_{2}+2-q}\right]\nonumber \\
 & \times\sum_{g_{1}+g_{2}=g-q+1}\left[z^{2g_{1}}w^{2g_{2}}\right]\underset{G\left(z,w\right)}{\underbrace{\frac{S\left(\left(A+B\right)Az\right)}{S\left(Az\right)S\left(Bw\right)}\sum_{k_{1}+\cdots+k_{q}=B}k_{1}\cdots k_{q}S\left(k_{1}Az\right)\cdots S\left(k_{q}Az\right)S\left(k_{1}Bw\right)\cdots S\left(k_{q}Bw\right)}}\label{eq:step2-3}
\end{align}
Note that for any formal power series $G\left(z,w\right)=\sum_{i,j\geq0}G_{i,j}z^{i}w^{j}$,
we have $\sum_{g_{1}+g_{2}=h}\left[z^{g_{1}}w^{g_{2}}\right]G\left(z,w\right)=\sum_{g_{1}+g_{2}=h}G_{g_{1},g_{2}}=\left[z^{h}\right]G\left(z,z\right)$.
Using this remark in Expression~(\ref{eq:step2-3}) with $h=g+1-q$
and using that $G\left(z,w\right)$ is even in $z$ and $w$ , we
get the following expression for $\langle\tau_{0}\tau_{d_{1}}\tau_{d_{2}}\rangle_{0,g},$
\begin{align*}
 & \sum_{q=1}^{g+1}\frac{i^{-d_{1}-d_{2}}}{q!}\left[A^{d_{1}+1-q}B^{d_{2}+2-q}z^{2g-2q+2}\right]\\
 & \times\frac{S\left(\left(A+B\right)Az\right)}{S\left(Az\right)S\left(Bz\right)}\sum_{k_{1}+\cdots+k_{q}=B}k_{1}\cdots k_{q}S\left(k_{1}Az\right)\cdots S\left(k_{q}Az\right)S\left(k_{1}Bz\right)\cdots S\left(k_{q}Bz\right).
\end{align*}
Re-write this expression as
\begin{align}
\langle\tau_{0}\tau_{d_{1}}\tau_{d_{2}}\rangle_{0,g}= & i^{-d_{1}-d_{2}}\left[A^{d_{1}}B^{d_{2}}z^{2g}\right]\frac{S\left(\left(A+B\right)Az\right)}{S\left(Az\right)S\left(Bz\right)}\nonumber \\
 & \times\sum_{q=1}^{g+1}\frac{1}{q!}A^{q-1}B^{q-2}z^{2q-2}\sum_{k_{1}+\cdots+k_{q}=B}\prod_{i=1}^{q}k_{i}S\left(k_{i}Az\right)S\left(k_{i}Bz\right).\label{eq:step2-4}
\end{align}
We can extend the range of summation to $q$ running from $1$ to
$\infty$. Indeed, it is clear from the expression that the terms
with $q>g+1$ vanishes, since we extract the coefficient of $z^{2g}$
from a power series with a factor $z^{2q-2}$. Hence, the second line
of Expression~(\ref{eq:step2-4}) can be re-written as the coefficient
of an exponential series. Expression~(\ref{eq:step2-4}) becomes
\begin{align}
\langle\tau_{0}\tau_{d_{1}}\tau_{d_{2}}\rangle_{0,g} & =i^{-d_{1}-d_{2}}\left[A^{d_{1}}B^{d_{2}}z^{2g}\right]\frac{S\left(\left(A+B\right)Az\right)}{S\left(Az\right)S\left(Bz\right)}\nonumber \\
 & \,\,\,\,\,\,\times\frac{1}{AB^{2}z^{2}}\left[t^{B}\right]\left(\exp\left(\sum_{k>0}ABz^{2}kS\left(kAz\right)S\left(kBz\right)t^{k}\right)-1\right)\label{eq:step2-5}
\end{align}

\paragraph{Step 5.}

We use properties of Eulerian numbers in order to simplify the expression
of the exponential power series in Expression~(\ref{eq:step2-5}).
Using first Proposition~\ref{prop:Eulerian fin} and then the definition
$S\left(z\right)=\frac{{\rm sh}\left(z/2\right)}{z/2}$, we get
\begin{align*}
\exp\left(\sum_{k>0}ABz^{2}kS\left(kAz\right)S\left(kBz\right)t^{k}\right) & =1+4\sum_{k>0}\frac{{\rm sh}\left(\frac{A}{2}z\right){\rm sh}\left(\frac{B}{2}z\right)}{{\rm sh}\left(\frac{A+B}{2}z\right)}{\rm sh}\left(k\frac{A+B}{2}z\right)t^{k}\\
 & =1+\sum_{k>0}ABkz^{2}\frac{S\left(Az\right)S\left(Bz\right)}{S\left(\left(A+B\right)z\right)}S\left(k\left(A+B\right)z\right)t^{k}.
\end{align*}
Thus, Expression~(\ref{eq:step2-5}) becomes after extracting the
coefficient of $t^{B}$
\begin{align*}
\langle\tau_{0}\tau_{d_{1}}\tau_{d_{2}}\rangle_{0,g} & =i^{-d_{1}-d_{2}}\left[A^{d_{1}}B^{d_{2}}z^{2g}\right]\frac{S\left(\left(A+B\right)Az\right)}{S\left(Az\right)S\left(Bz\right)}\\
 & \,\,\,\,\,\,\times\frac{4}{AB^{2}z^{2}}\frac{{\rm sh}\left(\frac{A}{2}z\right){\rm sh}\left(\frac{B}{2}z\right)}{{\rm sh}\left(\frac{A+B}{2}z\right)}{\rm sh}\left(B\frac{A+B}{2}z\right).
\end{align*}
We re-write the second line as $\frac{S\left(Az\right)S\left(Bz\right)}{S\left(\left(A+B\right)z\right)}S\left(B\left(A+B\right)z\right)$
so that
\begin{align*}
\langle\tau_{0}\tau_{d_{1}}\tau_{d_{2}}\rangle_{0,g} & =i^{-d_{1}-d_{2}}\left[A^{d_{1}}B^{d_{2}}z^{2g}\right]\frac{S\left(\left(A+B\right)Az\right)S\left(B\left(A+B\right)z\right)}{S\left(\left(A+B\right)z\right)}.
\end{align*}
Finally, the change of variable $z:=\frac{z}{A+B}$ in this last expression
gives $\langle\langle\tau_{0}\tau_{d_{1}}\tau_{d_{2}}\rangle\rangle_{g}$
as expressed in Eq.~(\ref{eq:RHS thm}), that is 
\[
\langle\tau_{0}\tau_{d_{1}}\tau_{d_{2}}\rangle_{0,g}=\left(-1\right)^{\frac{-d_{1}-d_{2}}{2}}\left[A^{d_{1}}B^{d_{2}}z^{2g}\right]\left(A+B\right)^{2g}\frac{S\left(Az\right)S\left(Bz\right)}{S\left(z\right)}.
\]

\subsubsection{Proof for $n\geq2$\label{subsec:General-proof}}

\paragraph{Convention.}

Because of the multiple use of the index $i$, we choose to denote
the imaginary unit $i$ as $\sqrt{-1}$ in this section. 

Fix $n+1$ nonnegative integers $d_{1},\dots,d_{n}$ and $g$, we
prove in this section that
\[
\langle\tau_{0}\tau_{d_{1}}\dots\tau_{d_{n}}\rangle_{0,g}=\langle\langle\tau_{0}\tau_{d_{1}}\dots\tau_{d_{n}}\rangle\rangle_{g}
\]
for $n\geq2$. We start from the LHS. As explained in the strategy
of the proof, we have 
\[
\langle\tau_{0}\tau_{d_{1}}\dots\tau_{d_{n}}\rangle_{0,g}=\left[\hbar^{g}\right]\frac{\sqrt{-1}^{g}}{\hbar^{n-1}}\left.\left[\ldots\left[H_{d_{1-1}},\overline{H}_{d_{2}}\right],\dots,\overline{H}_{d_{n}}\right]\right|_{u_{i}=\delta_{1,i}}.
\]
We follow the steps of the previous section in this more general setting.
The main difference occurs in Step~$5$, where we need a combinatorial
lemma which was obvious when $n=2$. Let us recall these steps. 

In Step~$1$, we obtain an expression of $\left[\cdots\left[H_{d_{1-1}},\overline{H}_{d_{2}}\right],\dots,\overline{H}_{d_{n}}\right]$
using the formulas of $H_{d_{1}-1}$, $\overline{H}_{d_{2}},\dots,\overline{H}_{d_{n}}$
and of the developed expression of the star product. 

In Step~$2$, we first extract the coefficient of  $\hbar^{g}$ in
$\frac{\sqrt{-1}^{g}}{\hbar^{n-1}}\left[\cdots\left[H_{d_{1-1}},\overline{H}_{d_{2}}\right],\dots,\overline{H}_{d_{n}}\right]$.
Then we perform the substitution $u_{i}=\delta_{i,1}$ in it and get
a first expression of $\langle\tau_{0}\tau_{d_{1}}\dots\tau_{d_{n}}\rangle_{0,g}$.
However this expression will be totally different from the one of
$\langle\langle\tau_{0}\tau_{d_{1}}\dots\tau_{d_{n}}\rangle\rangle_{g}$
given by Eq.~(\ref{eq:RHS thm}).

In Steps $3$,$4$ and $5$, we will transform this last expression
of $\langle\tau_{0}\tau_{d_{1}}\dots\tau_{d_{n}}\rangle_{0,g}$ into
the expression of $\langle\langle\tau_{0}\tau_{d_{1}}\dots\tau_{d_{n}}\rangle\rangle_{g}$
given by Eq.~(\ref{eq:RHS thm}). 

\paragraph{Step $1$.}

We compute in the first step $\left[\ldots\left[H_{d_{1-1}},\overline{H}_{d_{2}}\right],\dots,\overline{H}_{d_{n}}\right]$.
From the expression of the Hamiltonian density Eq.~(\ref{eq:Expression Hp}),
we set 
\[
H_{d_{1}-1}\left(x\right)=\sum_{\underset{\text{s.t.}\,2g_{1}+m_{1}=d_{1}+1}{g_{1},m_{1}\geq0}}\frac{\left(\sqrt{-1}\hbar\right)^{g_{1}}}{m_{1}!}\sum_{a_{1}^{1},\dots,a_{m_{1}}^{1}\in\mathbb{Z}}\left(\left[z_{1}^{2g_{1}}\right]\frac{S\left(a_{1}^{1}z\right)\cdots S\left(a_{m_{1}}^{1}z\right)S\left(A_{1}z\right)}{S\left(z\right)}\right)p_{a_{1}^{1}}\cdots p_{a_{m_{1}}^{1}}e^{\sqrt{-1}xA_{1}},
\]
and
\[
\overline{H}_{d_{i}}=\sum_{\underset{\text{s.t.}\,2g_{i}+m_{i}=d_{i}+2}{g_{i},m_{i}\geq0}}\frac{\left(\sqrt{-1}\hbar\right)^{g_{i}}}{m_{i}!}\sum_{\underset{A_{i}=0}{a_{1}^{i},\dots,a_{m_{i}}^{i}\in\mathbb{Z}}}\left(\left[z_{i}^{2g_{i}}\right]\frac{S\left(a_{1}^{i}z_{i}\right)\cdots S\left(a_{m_{i}}^{i}z_{j}\right)}{S\left(z_{i}\right)}\right)p_{a_{1}^{i}}\cdots p_{a_{m_{i}}^{i}},\text{ with }2\leq i\leq n,
\]
where $A_{i}:=\sum_{j=1}^{m_{i}}a_{j}^{i}$ and $1\leq i\leq n$.
In these notations, we use the variables of summations $a_{j}^{1}$
in $H_{d_{1}-1}$ and the variables of summations $a_{j}^{i}$ in
$\overline{H}_{d_{i}}$, with $2\leq i\leq n$. Note that in the notation
$a_{j}^{i}$, $i$ is just an upper index. The expression of $\overline{H}_{d_{i}}$
is obtained by a formal $x$-integration along $S^{1}$ of $H_{d_{i}}$.
Hence, the sum over $a_{1}^{i},\dots,a_{m_{i}}^{i}$ has the constraint
$A_{i}=0$.

In Step $1.1$, we give an expression for $H_{d_{1}-1}\star\overline{H}_{d_{2}}\star\cdots\star\overline{H}_{d_{n}}$.
In Step $1.2$ we explain why this is the only term of the commutator
$\left[\ldots\left[H_{d_{1-1}},\overline{H}_{d_{2}}\right],\dots,\overline{H}_{d_{n}}\right]$
needed to compute $\langle\tau_{0}\tau_{d_{1}}\dots\tau_{d_{n}}\rangle_{0,g}$.

\subparagraph{Step $1.1$.}
\begin{prop}
\label{prop: multi star product}We have
\begin{align}
H_{d_{1}-1}\star\overline{H}_{d_{2}}\star\cdots\star\overline{H}_{d_{n}}= & \sum_{q_{I}\geq0,\,I\in\mathcal{C}}\sum_{g_{1},\dots,g_{n}\geq0}\sum_{\underset{{\rm with\,conditions\,\alpha}}{\tilde{m}_{1},\dots,\tilde{m}_{n}\geq0}}\label{eq: pdt star n hamiltoniens}\\
 & \prod_{i=1}^{n}\left(\frac{\left(\sqrt{-1}\hbar\right)^{g_{i}}}{\tilde{m}_{i}!}\left[z_{i}^{2g_{i}}\right]\sum_{a_{1}^{i},\dots,a_{\tilde{m}_{i}}^{i}\in\mathbb{Z}}W_{i}\left(a_{1}^{i},\dots,a_{\tilde{m}_{1}}^{i},z_{i}\right)\,p_{a_{1}}\cdots p_{a_{\tilde{m}_{i}}}e^{\sqrt{-1}x\tilde{A}_{i}}\right)\nonumber \\
 & \times\prod_{I\in\mathcal{C}}\left(\frac{\left(\sqrt{-1}\hbar\right)^{q_{I}}}{q_{I}!}\sum_{\underset{{\rm with\,conditions}\,\beta}{k_{1}^{I},\dots,k_{q_{I}}^{I}>0}}k_{1}^{I}\cdots k_{q_{I}}^{I}W^{I}\left(k_{1}^{I},\dots,k_{q_{I}}^{I},z_{I}\right)\right),\nonumber 
\end{align}
where
\begin{itemize}
\item $\mathcal{C}$ is the set of pairs (2-element subsets) of $\left\{ 1,\dots,n\right\} $;
we also denote by $\mathcal{C}_{i}\subset\mathcal{C}$ the subset
of pairs that contain $i$,
\item the conditions $\alpha$ on the summations running over $g_{1},\dots,g_{n},\tilde{m}_{1},\dots,\tilde{m}_{n}$
and $q_{I},\,I\in\mathcal{C}$ are
\[
2g_{1}+\tilde{m}_{1}+\sum_{I\in\mathcal{C}_{1}}q_{I}=d_{1}+1
\]
and
\[
2g_{j}+\tilde{m}_{j}+\sum_{I\in\mathcal{C}_{j}}q_{I}=d_{j}+2,\,\,{\rm with\,}2\leq j\leq n,
\]
\item the weight $W_{i}$ with $1\leq i\leq n$ is defined by
\[
W_{1}\left(a_{1}^{1},\dots,a_{\tilde{m}_{1}}^{1},z_{1}\right):=\frac{\prod_{j=1}^{\tilde{m}_{1}}S\left(a_{j}^{1}z_{1}\right)}{S\left(z_{1}\right)}S\left(\tilde{A}_{1}z_{1}+\cdots\tilde{A}_{n}z_{1}\right)
\]
and
\[
W_{i}\left(a_{1}^{i},\dots,a_{\tilde{m}_{1}}^{i},z_{i}\right):=\frac{\prod_{j=1}^{\tilde{m}_{i}}S\left(a_{j}^{i}z_{i}\right)}{S\left(z_{i}\right)},\,\text{for }2\leq i\leq n,
\]
\item $\tilde{A}_{j}=\sum_{i=1}^{\tilde{m}_{j}}a_{i}$,
\item the conditions $\beta$ are the following $\left(n-1\right)$ constraints
over the summations:
\begin{equation}
\tilde{A}_{i}-\sum_{j=1}^{i-1}K^{\left\{ j,i\right\} }+\sum_{j=i+1}^{n}K^{\left\{ i,j\right\} }=0,\,\,\,\,\,2\leq i\leq n\label{eq: conditions beta}
\end{equation}
where $K^{I}=k_{1}^{I}+\cdots+k_{q_{I}}^{I}$, for $I\in\mathcal{C}$,
\item the weight $W^{I}$ for $I=\left\{ i,j\right\} \in\mathcal{C}$ is
defined by 
\begin{align*}
W^{I}\left(k_{1}^{I},\dots,k_{q_{I}}^{I},z_{I}\right): & =S\left(k_{1}^{I}z_{i}\right)\cdots S\left(k_{q_{I}}^{I}z_{i}\right)\times S\left(k_{1}^{I}z_{j}\right)\cdots S\left(k_{q_{I}}^{I}z_{j}\right),
\end{align*}
note that the notation $z_{I}$ means $z_{i},z_{j}$. 
\end{itemize}
\end{prop}
\begin{proof}
We use the expression of the star product given by
\begin{equation}
f\star g=f\sum_{q\geq0}\frac{\left(\sqrt{-1}\hbar\right)^{q}}{q!}\left(\sum_{k_{1},\dots,k_{q}>0}k_{1}\cdots k_{q}\overleftarrow{\frac{\partial}{\partial p_{k_{1}}}}\cdots\overleftarrow{\frac{\partial}{\partial p_{k_{q}}}}\overrightarrow{\frac{\partial}{\partial p_{-k_{1}}}}\cdots\overrightarrow{\frac{\partial}{\partial p_{-k_{q}}}}\right)g.\label{eq: star pdt}
\end{equation}
The star product is associative as one can check from Eq.~(\ref{eq: star pdt}).
We use this associativity in the following way
\[
H_{d_{1}-1}\star\overline{H}_{d_{2}}\star\cdots\star\overline{H}_{d_{n}}=\left(\cdots\left(H_{d_{1}-1}\star\overline{H}_{d_{2}}\right)\star\cdots\star\overline{H}_{d_{n}}\right).
\]
Each of the $n-1$ star products has couples of derivatives acting
on the left and on the right with opposite indices. Let $2\leq i\leq n$.
The $\left(i-1\right)$th star product acts on the left on $H_{d_{1}-1},\,\overline{H}_{d_{2}},\dots,\overline{H}_{d_{i-1}}$
and on the right only on $\overline{H}_{d_{i}}$. Fix a nonnegative
integer $q$ and $q$ positive integers $k_{1},\dots,k_{q}$. Consider
the term 
\[
\frac{\left(\sqrt{-1}\hbar\right)^{q}}{q!}k_{1}\cdots k_{q}\overleftarrow{\frac{\partial}{\partial p_{k_{1}}}}\cdots\overleftarrow{\frac{\partial}{\partial p_{k_{q}}}}\overrightarrow{\frac{\partial}{\partial p_{-k_{1}}}}\cdots\overrightarrow{\frac{\partial}{\partial p_{-k_{q}}}}
\]
in the development in $\hbar$ of the $\left(i-1\right)$th star product.
Among these $q$ left derivatives, we denote by $q_{\left\{ i,j\right\} }$,
with $j<i$, the number of derivatives acting on $\overline{H}_{d_{j}}$
(or $H_{d_{1}-1}$ if $j=1$). We furthermore add an upper index $\left\{ i,j\right\} $
on the corresponding $k$ variables so that the $q_{\left\{ i,j\right\} }$
left derivatives coming from the $\left(i-1\right)$th star product
and acting on the $j$th Hamiltonian are denoted by 
\[
\overleftarrow{\frac{\partial}{\partial p_{k_{1}^{\left\{ i,j\right\} }}}}\cdots\overleftarrow{\frac{\partial}{\partial p_{k_{q_{\left\{ i,j\right\} }}^{\left\{ i,j\right\} }}}}.
\]
The associate right derivatives $\overrightarrow{\frac{\partial}{\partial p_{-k_{1}^{\left\{ i,j\right\} }}}}\cdots\overrightarrow{\frac{\partial}{\partial p_{-k_{q_{\left\{ i,j\right\} }}^{\left\{ i,j\right\} }}}}$
act on $\overline{H}_{d_{i}}$ with opposite indices. With this notation,
we obtain
\begin{align}
 & H_{d_{1}-1}\star\overline{H}_{d_{2}}\star\cdots\star\overline{H}_{d_{n}}\nonumber \\
 & =\sum_{q_{I}\geq0,\,I\in\mathcal{C}}\,\,\,\sum_{k_{1}^{I},\dots,k_{q_{I}}^{I}\geq0,\,I\in\mathcal{C}}\,\,\,\prod_{I\in\mathcal{C}}\frac{\left(\sqrt{-1}\hbar\right)^{q_{I}}}{q_{I}!}k_{1}^{I}\cdots k_{q_{I}}^{I}\nonumber \\
 & \times\prod_{J\in\mathcal{C}_{1}}\frac{\partial^{q_{J}}}{\partial p_{k_{1}^{J}}\cdots\partial p_{k_{q_{J}}^{J}}}H_{d_{1}-1}\nonumber \\
 & \times\prod_{i=2}^{n}\prod_{j=1}^{i-1}\left(\frac{\partial^{q_{\left\{ i,j\right\} }}}{\partial p_{-k_{1}^{\left\{ i,j\right\} }}\cdots\partial p_{-k_{q_{\left\{ i,j\right\} }}^{\left\{ i,j\right\} }}}\right)\prod_{l=i+1}^{n}\left(\frac{\partial^{q_{\left\{ i,l\right\} }}}{\partial p_{k_{1}^{\left\{ i,l\right\} }}\cdots\partial p_{k_{q_{\left\{ i,l\right\} }}^{\left\{ i,l\right\} }}}\right)\overline{H}_{d_{i}}.\label{eq: 1st formula multiple star products}
\end{align}
Let us explain this formula.\textbf{ }The derivatives acting on $\overline{H}_{d_{i}}$
have two different origins; the derivatives coming from the $\left(i-1\right)$th
star product, these are the derivatives with negative indices in the
product running over the variable $j$, and the derivatives coming
from the $i$th to the $\left(n-1\right)$th star product, these are
the derivatives with positive indices in the product running over
the variable $l$. Similarly, the derivatives acting on $H_{d_{1}-1}$
come from all the star products and have positive indices. Moreover,
when we develop the star products, we have to choose which derivative
acts on which Hamiltonian so that multinomial coefficients appear
and simplify the factorials. 

We now describe the action of the $\sum_{j=1,j\neq i}^{n}q_{\left\{ i,j\right\} }$
derivatives of the last line of Eq.~(\ref{eq: 1st formula multiple star products})
on $\overline{H}_{d_{i}}$. We find
\begin{align}
 & \prod_{j=1}^{i-1}\left(\frac{\partial^{q_{\left\{ i,j\right\} }}}{\partial p_{-k_{1}^{\left\{ i,j\right\} }}\cdots\partial p_{-k_{q_{\left\{ i,j\right\} }}^{\left\{ i,j\right\} }}}\right)\prod_{l=i+1}^{n}\left(\frac{\partial^{q_{\left\{ i,l\right\} }}}{\partial p_{k_{1}^{\left\{ i,l\right\} }}\cdots\partial p_{k_{q_{\left\{ i,l\right\} }}^{\left\{ i,l\right\} }}}\right)\overline{H}_{d_{i}}\nonumber \\
 & =\sum_{\underset{2g_{i}+\tilde{m}_{i}+\sum_{I\in\mathcal{C}_{i}}q_{I}=d_{i}+2}{g_{i},\tilde{m_{i}}\geq0}}\frac{\left(\sqrt{-1}\hbar\right)^{g_{i}}}{\tilde{m}_{i}!}\left[z^{2g_{i}}\right]\nonumber \\
 & \times\sum_{\underset{{\rm with\,condition\,\beta}}{a_{1}^{i},\dots,a_{\tilde{m}_{i}}^{i}\in\mathbb{Z}}}W_{i}\left(a_{1}^{i},\dots,a_{\tilde{m}_{1}}^{i},z_{i}\right)\,p_{a_{1}}\cdots p_{a_{\tilde{m}_{i}}}\nonumber \\
 & \times\prod_{j=1}^{i-1}S\left(k_{1}^{\left\{ i,j\right\} }z_{i}\right)\cdots S\left(k_{q_{\left\{ i,j\right\} }}^{\left\{ i,j\right\} }z_{i}\right)\prod_{l=i+1}^{n}S\left(k_{1}^{\left\{ i,l\right\} }z_{i}\right)\cdots S\left(k_{q_{\left\{ i,l\right\} }}^{\left\{ i,l\right\} }z_{i}\right).\label{eq: derivatives on H_d_i}
\end{align}
Indeed, when the $\sum_{j=1,j\neq i}^{n}q_{\left\{ i,j\right\} }$
derivatives act on $\frac{S\left(a_{1}^{i}z_{i}\right)\cdots S\left(a_{m_{i}}^{i}z_{j}\right)}{S\left(z_{i}\right)}p_{a_{1}^{i}}\cdots p_{a_{m_{i}}^{i}}$
in $\overline{H}_{d_{i}}$, it remains $\tilde{m}_{i}=m_{i}-\sum_{j=1,j\neq i}^{n}q_{\left\{ i,j\right\} }$
variables $p$. The part of this expression which is not reached by
the derivatives is contained in the third line while the part reached
by the derivatives with negative and positive indices is the content
of the last line. Finally, the condition $A_{i}=0$ in $\overline{H}_{d_{i}}$
becomes $\tilde{A}_{i}+\sum_{j=i+1}^{n}K^{\left\{ i,j\right\} }-\sum_{j=1}^{n-1}K^{\left\{ j,i\right\} }=0$,
this is the $i$th equation of conditions $\beta$.

Similarly, there are only derivatives with positive indices acting
on $H_{d_{1}-1}$. We find 
\begin{align}
\prod_{J\in\mathcal{C}_{1}}\frac{\partial^{q_{J}}}{\partial p_{k_{1}^{J}}\cdots\partial p_{k_{q_{J}}^{J}}}H_{d_{1}-1} & =\sum_{\underset{\text{s.t.}\,2g_{1}+\tilde{m_{1}}+\sum_{I\in\mathcal{C}_{1}}q_{I}=d_{1}+1}{g_{1},m_{1}\geq0}}\frac{\left(\sqrt{-1}\hbar\right)^{g_{1}}}{\tilde{m}_{1}!}\left[z^{2g_{1}}\right]\nonumber \\
 & \times\sum_{a_{1}^{1},\dots,a_{\tilde{m}_{1}}^{1}\in\mathbb{Z}}W_{1}\left(a_{1}^{1},\dots,a_{\tilde{m}_{1}}^{1},z_{1}\right)\,p_{a_{1}}\cdots p_{a_{\tilde{m}_{1}}}e^{\sqrt{-1}x\sum_{i=1}^{n}\tilde{A}_{i}}\label{eq: derivatives on H_d_1}\\
 & \times\prod_{J\in\mathcal{C}_{1}}S\left(k_{1}^{J}z_{1}\right)\cdots S\left(k_{q_{J}}^{J}z_{1}\right).\nonumber 
\end{align}
Note that $A_{1}$ becomes after the action of the derivatives 
\[
\tilde{A}_{1}-\sum_{j=2}^{n}K^{\left(1,j\right)}=\sum_{i=1}^{n}\tilde{A}_{i}.
\]
We obtained this equality by summing the $\left(n-1\right)$ equations
of conditions $\beta$ (see Eq.~(\ref{eq: conditions beta})). 

By combining Eq.~(\ref{eq: 1st formula multiple star products}),
(\ref{eq: derivatives on H_d_i}) and (\ref{eq: derivatives on H_d_1})
we obtain Eq.~(\ref{eq: pdt star n hamiltoniens}). This proves the
proposition. 
\end{proof}

\subparagraph{Step $1.2$.}

Similarly to the expression of $H_{d_{1}-1}\star\overline{H}_{d_{2}}\star\cdots\star\overline{H}_{d_{n}}$
obtained in Step~$1.1$ , we can get the expressions of the $2^{n-1}-1$
others terms appearing in $\left[\ldots\left[H_{d_{1}-1},\overline{H}_{d_{2}}\right],\dots,\overline{H}_{d_{n}}\right]$.
To each term we associate a permutation $\sigma\in\mathcal{S}_{n}$
such that $\sigma\left(i\right)$ is the index of the Hamiltonian
appearing at the $i$th position (the index of $H_{d_{1}-1}$ is $1$
and the index of $\overline{H}_{d_{i}}$ is $i$). Then, the expression
of the term in $\left[\ldots\left[H_{d_{1}-1},\overline{H}_{d_{2}}\right],\dots,\overline{H}_{d_{n}}\right]$
corresponding to the permutation $\sigma$ is given by Expression~(\ref{eq: pdt star n hamiltoniens})
with one modifications: conditions $\beta$ become 
\[
\tilde{A}_{i}-\sum_{j=1}^{\sigma^{-1}\left(i\right)-1}K^{\left\{ i,\sigma\left(j\right)\right\} }+\sum_{j=\sigma^{-1}\left(i\right)+1}^{n}K^{\left\{ i,\sigma\left(j\right)\right\} },\,\,\,2\leq i\leq n.
\]

However, it is not necessary to compute these $2^{n-1}-1$ other terms
in order to obtain $\langle\tau_{0}\tau_{d_{1}}\dots\tau_{d_{n}}\rangle_{0,g}$.
Indeed, $\left[\ldots\left[H_{d_{1}-1},\overline{H}_{d_{2}}\right],\dots,\overline{H}_{d_{n}}\right]$
is a power series in the indeterminates $a_{1}^{1},\dots,a_{\tilde{m}_{1}}^{1},\dots,a_{1}^{n}\dots,a_{\tilde{m}_{n}}^{n}$.
In Step~$2$ we extract one coefficient of this power series. Then,
we can restrict to compute the terms in $\left[\ldots\left[H_{d_{1}-1},\overline{H}_{d_{2}}\right],\dots,\overline{H}_{d_{n}}\right]$
such that
\[
\tilde{A}_{i}>0,\,\text{with }i\geq2.
\]
The only term in $\left[\ldots\left[H_{d_{1}-1},\overline{H}_{d_{2}}\right],\dots,\overline{H}_{d_{n}}\right]$
satisfying these inequalities is $H_{d_{1}-1}\star\overline{H}_{d_{2}}\star\cdots\star\overline{H}_{d_{n}}$.
Indeed, in the other terms, the condition coming from the Hamiltonian
on the leftmost in the star product, say $\overline{H}_{d_{i}}$,
is 
\[
\tilde{A}_{i}+\sum_{1<j}K^{\left\{ 1,j\right\} }=0,
\]
that is $\tilde{A}_{j}<0$.

Although, we will use one simplification coming from the commutators.
A commutator simplifies the constant term in $\hbar$ in the star
product, that is the term coming from the commutative product. These
terms correspond in the expression of $H_{d_{1}-1}\star\overline{H}_{d_{2}}\star\cdots\star\overline{H}_{d_{n}}$
given by Eq.~(\ref{eq: pdt star n hamiltoniens}) to the terms satisfying
\[
\sum_{j=1}^{i-1}q_{\left\{ i,j\right\} }=0,\,{\rm for}\,2\leq i\leq n.
\]
Indeed, $\sum_{j=1}^{i-1}q_{\left\{ i,j\right\} }$ counts the number
of left (or right) derivatives coming from the $\left(i-1\right)$th
star product and the commutative term in the star product is the one
without derivatives. We call conditions $\gamma$ the inequalities
\[
\sum_{j=1}^{i-1}q_{\left\{ i,j\right\} }\geq1,\,{\rm for}\,2\leq i\leq n.
\]

\paragraph{Change of notation.}

We remove the tildes in our notations, i.e. we set $m_{i}:=\tilde{m}_{i}$
and $A_{i}:=\tilde{A}_{i}$ for any $1\leq i\leq n$.

\paragraph{Step 2. }

We first extract the coefficient of $\hbar^{g}$ from $\frac{\sqrt{-1}^{g}}{\hbar^{n-1}}\left[\ldots\left[H_{d_{1-1}},\overline{H}_{d_{2}}\right],\dots,\overline{H}_{d_{n}}\right]$.
Then we evaluate this coefficient, which is a differential polynomial,
at $u_{i}=\delta_{i,1}$. We will then get an expression for $\langle\tau_{0}\tau_{d_{1}}\dots\tau_{d_{n}}\rangle_{0,g}=\left[\hbar^{g}\right]\frac{\sqrt{-1}^{g}}{\hbar^{n-1}}\left.\left[\ldots\left[H_{d_{1-1}},\overline{H}_{d_{2}}\right],\dots,\overline{H}_{d_{n}}\right]\right|_{u_{i}=\delta_{1,i}}$
.

We want to extract the coefficient of $\hbar^{g}$ in $\frac{\sqrt{-1}^{g}}{\hbar^{n-1}}\left[\ldots\left[H_{d_{1-1}},\overline{H}_{d_{2}}\right],\dots,\overline{H}_{d_{n}}\right]$.
As explained in Step $1.2$, we only need to study the term $H_{d_{1}-1}\star\overline{H}_{d_{2}}\star\cdots\star\overline{H}_{d_{n}}$
in $\left[\ldots\left[H_{d_{1-1}},\overline{H}_{d_{2}}\right],\dots,\overline{H}_{d_{n}}\right]$.
The coefficient of $\hbar^{g}$ in $\frac{\sqrt{-1}^{g}}{\hbar^{n-1}}H_{d_{1}-1}\star\overline{H}_{d_{2}}\star\cdots\star\overline{H}_{d_{n}}$
is easily obtained from Expression~(\ref{eq: pdt star n hamiltoniens}).
We get with our new notations

\begin{align*}
\left[\hbar^{g}\right]\frac{\sqrt{-1}^{g}}{\hbar^{n-1}}\left[\ldots\left[H_{d_{1-1}},\overline{H}_{d_{2}}\right],\dots,\overline{H}_{d_{n}}\right] & =\sum_{\underset{{\rm with\,conditions\,\gamma}}{g_{1}+\cdots+g_{n}+\sum_{I\in\mathcal{C}}q_{I}=g+n-1}}\sqrt{-1}^{2g+\left(n-1\right)}\\
 & \times\prod_{i=1}^{n}\left(\frac{1}{m_{i}!}\left[z_{i}^{2g_{i}}\right]\sum_{a_{1}^{i},\dots,a_{m_{1}}^{i}\in\mathbb{Z}}W_{i}\left(a_{1}^{i},\dots,a_{m_{1}}^{i},z_{i}\right)\,p_{a_{1}}\cdots p_{a_{m_{i}}}e^{\sqrt{-1}xA_{i}}\right)\\
 & \times\prod_{I\in\mathcal{C}}\left(\frac{1}{q_{I}!}\sum_{\underset{\text{with conditions }\beta}{k_{1}^{I},\dots,k_{q_{I}}^{I}>0}}k_{1}^{I}\cdots k_{q_{I}}^{I}W^{I}\left(k_{1}^{I},\dots,k_{q_{I}}^{I},z_{I}\right)\right)\\
 & +2^{n-1}-1\text{ other terms},
\end{align*}
where we used conditions $\alpha$ to fix
\[
m_{1}=d_{1}+1-2g_{1}-\sum_{I\in\mathcal{C}_{1}}q_{I},\text{ and }m_{i}=d_{i}+2-2g_{i}-\sum_{I\in\mathcal{C}_{i}}q_{I},\text{ when }2\leq i\leq n.
\]

This last expression is a differential polynomial thanks to Proposition~\ref{prop: Stability commutator}.
In order to substitute $u_{i}=\delta_{i,1}$, we use Lemma~\ref{lem:evaluation }.
We get 

\begin{align*}
\langle\tau_{0}\tau_{d_{1}}\dots\tau_{d_{n}}\rangle_{0,g}= & \sum_{\underset{{\rm with\,conditions\,\gamma}}{g_{1}+\cdots+g_{n}+\sum_{I\in\mathcal{C}}q_{I}=g+n-1}}\sqrt{-1}^{n-2-\left|\bm{d}\right|}\\
 & \times\prod_{i=1}^{n}\left(\frac{1}{m_{i}!}\left[a_{1}^{i}\cdots a_{m_{1}}^{i}\right]\left[z_{i}^{2g_{i}}\right]W_{i}\left(a_{1}^{i},\dots,a_{m_{1}}^{i},z_{i}\right)\right)\\
 & \times\prod_{I\in\mathcal{C}}\left(\frac{1}{q_{I}!}\sum_{\underset{{\rm with\,conditions\,\beta}}{k_{1}^{I},\dots,k_{q_{I}}^{I}>0}}k_{1}^{I}\cdots k_{q_{I}}^{I}W^{I}\left(k_{1}^{I},\dots,k_{q_{I}}^{I},z_{I}\right)\right)\\
 & +2^{n-1}-1\text{ other terms},
\end{align*}
where we simplified the power of $\sqrt{-1}$ using $m_{1}:=d_{1}+1-2g_{1}-\sum_{I\in\mathcal{C}_{1}}q_{I}$
and $m_{i}:=d_{i}+2-2g_{i}-\sum_{I\in\mathcal{C}_{i}}q_{I}$, when
$2\leq i\leq n$. We used the notation $\left|\bm{d}\right|=\sum_{i=1}^{n}d_{i}$.
\begin{rem}
It can look confusing that in this expression, $a_{j}^{i}$ for $2\leq i\leq n$
and $1\leq j\leq m_{i}$ stands for a formal variable and an integer
when we write the $i$th constraint
\[
A_{i}-\sum_{l=1}^{i-1}K^{\left\{ l,i\right\} }+\sum_{l=i+1}^{n}K^{\left\{ i,l\right\} }=0
\]
of conditions $\beta$. This is due to the presence of Ehrhart polynomials.
Indeed, according to \cite[Lemma A.1]{BuryakRossi2016}, for any list
of integers $A_{2},\dots,A_{n}$, the coefficient of any power in
$z_{1},\dots,z_{n}$ of 
\[
\prod_{I\in\mathcal{C}}\left(\frac{1}{q_{I}!}\sum_{\underset{{\rm with\,conditions\,\beta}}{k_{1}^{I},\dots,k_{q_{I}}^{I}>0}}k_{1}^{I}\cdots k_{q_{I}}^{I}W^{I}\left(k_{1}^{I},\dots,k_{q_{I}}^{I},z_{I}\right)\right)
\]
is a polynomial in the variables $A_{2},\dots,A_{n}$. We will then
use the $A_{i}$'s and the $a_{j}^{i}$'s as integers and formal variables.
\end{rem}

\paragraph{Step 3.}

We use the same simplification than Step~$3$ in Section~\ref{subsec:Less-trivial-case},
that is we consider the simplifications coming from extracting the
coefficient of $a_{1}^{i}\cdots a_{m_{1}}^{i}$ in each factor of the
product over $i$.
\begin{itemize}
\item Recall that $S$ an even power series so that the coefficient of $\alpha$
in S$\left(\alpha z\right)\times F\left(\alpha\right)$ is the coefficient
of $\alpha$ in $F$. Hence, we can replace, in our expression of
$\langle\tau_{0}\tau_{d_{1}}\dots\tau_{d_{n}}\rangle_{0,g}$, $W_{1}$
by $\frac{S\left(\sum_{i=1}^{n}A_{i}z_{1}\right)}{S\left(z_{1}\right)}$
and $W_{i}$ by $\frac{1}{S\left(z_{i}\right)}$, when $2\leq i\leq n$.
\item Thanks to these simplifications, we see that we extract the coefficient
of $\prod_{i=1}^{n}a_{1}^{i}\cdots a_{m_{1}}^{i}$ from a power series
which only depends in the $a_{i}^{j}$'s through their sums $A_{i}=\sum_{j=1}^{m_{i}}a_{j}^{i}$,
for $1\leq i\leq n$. This is equivalent to extracting the coefficient
of $\prod_{i=1}^{n}\frac{A_{i}^{m_{i}}}{m_{i}!}$ from the same power
series.
\item For simplicity (see Step $1.2$), we can suppose that $A_{i}>0,\,\text{with }i\geq2.$
\end{itemize}
Thanks to these three points, we get
\begin{align}
\langle\tau_{0}\tau_{d_{1}}\dots\tau_{d_{n}}\rangle_{0,g}= & \sum_{\underset{{\rm with\,conditions\,\gamma}}{g_{1}+\cdots+g_{n}+\sum_{I\in\mathcal{C}}q_{I}=g+n-1}}\sqrt{-1}^{n-2-\left|\bm{d}\right|}\nonumber \\
 & \times\prod_{i=1}^{n}\left(\left[A_{i}^{m_{i}}\right]\left[z_{i}^{2g_{i}}\right]\frac{1}{S\left(z_{i}\right)}\right)S\left(\left|\bm{A}\right|z_{1}\right)\nonumber \\
 & \times\prod_{I\in\mathcal{C}}\left(\frac{1}{q_{I}!}\sum_{\underset{{\rm with\,conditions\,\beta}}{k_{1}^{I},\dots,k_{q_{I}}^{I}>0}}k_{1}^{I}\cdots k_{q_{I}}^{I}W^{I}\left(k_{1}^{I},\dots,k_{q_{I}}^{I},z_{I}\right)\right),\label{eq: step 3 n hamiltonians}
\end{align}
where we used the notations $\left|\bm{A}\right|=\sum_{i=1}^{n}A_{i}$.

\paragraph{Step 4.}

We perform some change of variable in our expression of $\langle\tau_{0}\tau_{d_{1}}\dots\tau_{d_{n}}\rangle_{0,g}$.
We will then organize this expression to see $\langle\tau_{0}\tau_{d_{1}}\dots\tau_{d_{n}}\rangle_{0,g}$
as the coefficient of a product of exponential power series. This
will allow us to use known properties of Eulerian numbers.

We perform the change of variables $z_{i}:=A_{i}z_{i}$, with $1\leq i\leq n$
in Expression~(\ref{eq: step 3 n hamiltonians}). We get 
\begin{align*}
\langle\tau_{0}\tau_{d_{1}}\dots\tau_{d_{n}}\rangle_{0,g}= & \sum_{\underset{{\rm with\,conditions\,\gamma}}{g_{1}+\cdots+g_{n}+\sum_{I\in\mathcal{C}}q_{I}=g+n-1}}\sqrt{-1}^{n-2-\left|\bm{d}\right|}\\
 & \times\prod_{i=1}^{n}\left(\left[A_{i}^{m_{i}+2g_{i}}\right]\left[z_{i}^{2g_{i}}\right]\frac{1}{S\left(A_{i}z_{i}\right)}\right)S\left(\left|\bm{A}\right|A_{1}z_{1}\right)\\
 & \times\prod_{I\in\mathcal{C}}\left(\frac{1}{q_{I}!}\sum_{\underset{{\rm with\,conditions\,\beta}}{k_{1}^{I},\dots,k_{q_{I}}^{I}>0}}k_{1}^{I}\cdots k_{q_{I}}^{I}\tilde{W}^{I}\left(k_{1}^{I},\dots,k_{q_{I}}^{I},z_{I}\right)\right),
\end{align*}
where we used the notation
\begin{align*}
\tilde{W}^{I}\left(k_{1}^{I},\dots,k_{q_{I}}^{I},z_{I}\right): & =W^{I}\left(k_{1}^{I},\dots,k_{q_{I}}^{I},A_{i}z_{i},A_{j}z_{j}\right)\\
 & =S\left(k_{1}^{I}A_{i}z_{i}\right)\cdots S\left(k_{q_{I}}^{I}A_{i}z_{i}\right)S\left(k_{1}^{I}A_{j}z_{j}\right)\cdots S\left(k_{q_{I}}^{I}A_{j}z_{j}\right)
\end{align*}
for any pair $I=\left\{ i,j\right\} $ of $\mathcal{C}$.

Using that $m_{1}=d_{1}+1-2g_{1}-\sum_{I\in\mathcal{C}_{1}}q_{I}$
and $m_{i}=d_{i}+2-2g_{i}-\sum_{I\in\mathcal{C}_{i}}q_{I}$, when
$2\leq i\leq n$, we re-write this expression as
\begin{align*}
\langle\tau_{0}\tau_{d_{1}}\dots\tau_{d_{n}}\rangle_{0,g}= & \sum_{\underset{{\rm with\,conditions\,\gamma}}{\sum_{I\in\mathcal{C}}q_{I}\leq g+n-1}}\sqrt{-1}^{n-2-\left|\bm{d}\right|}\left[A_{1}^{d_{1}+1-\sum_{I\in\mathcal{C}_{1}}q_{I}}A_{2}^{d_{2}+2-\sum_{I\in\mathcal{C}_{2}}q_{I}}\cdots A_{n}^{d_{n}+2-\sum_{I\in\mathcal{C}_{n}}q_{I}}\right]\\
 & \times\sum_{g_{1}+\cdots+g_{n}=g+n-1-\sum_{I\in\mathcal{C}}q_{I}}\left[z_{1}^{2g_{1}}\cdots z_{n}^{2g_{n}}\right]\\
 & \times\underset{G\left(z_{1},\dots,z_{n}\right)}{\underbrace{S\left(\left|\bm{A}\right|A_{1}z_{1}\right)\prod_{i=1}^{n}\left(\frac{1}{S\left(A_{i}z_{i}\right)}\right)\prod_{I\in\mathcal{C}}\left(\frac{1}{q_{I}!}\sum_{\underset{{\rm with\,conditions\,\beta}}{k_{1}^{I},\dots,k_{q_{I}}^{I}>0}}k_{1}^{I}\cdots k_{q_{I}}^{I}\tilde{W}^{I}\left(k_{1}^{I},\dots,k_{q_{I}}^{I},z_{I}\right)\right)}}.
\end{align*}
We use in this expression that
\[
\sum_{g_{1}+\cdots+g_{n}=g+n-1-\sum_{I\in\mathcal{C}}q_{I}}\left[z_{1}^{2g_{1}}\cdots z_{n}^{2g_{n}}\right]G\left(z_{1},\dots,z_{n}\right)=\left[z^{2g+2n-2-\sum_{I\in\mathcal{C}}2q_{I}}\right]G\left(z,\dots,z\right)
\]
to obtain
\begin{align*}
\langle\tau_{0}\tau_{d_{1}}\dots\tau_{d_{n}}\rangle_{0,g}= & \sum_{\underset{{\rm with\,conditions\,\gamma}}{\sum_{I\in\mathcal{C}}q_{I}\leq g+n-1}}\sqrt{-1}^{n-2-\left|\bm{d}\right|}\left[A_{1}^{d_{1}+1-\sum_{I\in\mathcal{C}_{1}}q_{I}}A_{2}^{d_{2}+2-\sum_{I\in\mathcal{C}_{2}}q_{I}}\cdots A_{n}^{d_{n}+2-\sum_{I\in\mathcal{C}_{n}}q_{I}}\right]\\
 & \times\left[z^{2g+2n-2-\sum_{I\in\mathcal{C}}2q_{I}}\right]\\
 & \times S\left(\left|\bm{A}\right|A_{1}z\right)\prod_{i=1}^{n}\left(\frac{1}{S\left(A_{i}z_{i}\right)}\right)\prod_{I\in\mathcal{C}}\left(\frac{1}{q_{I}!}\sum_{\underset{{\rm with\,conditions\,\beta}}{k_{1}^{I},\dots,k_{q_{I}}^{I}>0}}k_{1}^{I}\cdots k_{q_{I}}^{I}\tilde{W}^{I}\left(k_{1}^{I},\dots,k_{q_{I}}^{I},z_{I}\right)\right).
\end{align*}
Then, we rewrite this expression as
\begin{align*}
\langle\tau_{0}\tau_{d_{1}}\dots\tau_{d_{n}}\rangle_{0,g}= & \sqrt{-1}^{n-2-\left|\bm{d}\right|}\left[A_{1}^{d_{1}}A_{2}^{d_{2}}\cdots A_{n}^{d_{n}}z^{2g}\right]S\left(\left|\bm{A}\right|A_{1}z\right)\prod_{i=1}^{n}\left(\frac{1}{S\left(A_{i}z_{i}\right)}\right)\frac{1}{A_{1}A_{2}^{2}\cdots A_{n}^{2}z^{2n-2}}\\
 & \times\sum_{\underset{{\rm with\,conditions\,\gamma}}{\sum_{I\in\mathcal{C}}q_{I}\leq g+n-1}}\prod_{I=\left\{ i,j\right\} \in\mathcal{C}}\left(\frac{A_{i}^{q_{I}}A_{j}^{q_{I}}z^{2q_{I}}}{q_{I}!}\sum_{\underset{{\rm with\,conditions\,\beta}}{k_{1}^{I},\dots,k_{q_{I}}^{I}>0}}k_{1}^{I}\cdots k_{q_{I}}^{I}\tilde{W}^{I}\left(k_{1}^{I},\dots,k_{q_{I}}^{I},z_{I}\right)\right).
\end{align*}
We can extend the range of summation to $\sum_{I\in\mathcal{C}}q_{I}$
running from $0$ to $\infty$. Indeed, it is clear from this expression
that the terms with $\sum_{I\in\mathcal{C}}q_{I}>g+n-1$ vanishes,
since we extract the coefficient of $z^{2g}$ from a power series
with a factor $\frac{z^{2\sum_{I\in\mathcal{C}}q_{I}}}{z^{2n-2}}$.
Hence, we rewrite the second line of the expression in the following
way
\begin{align}
 & \langle\tau_{0}\tau_{d_{1}}\dots\tau_{d_{n}}\rangle_{0,g}=\sqrt{-1}^{n-2-\left|\bm{d}\right|}\left[A_{1}^{d_{1}}A_{2}^{d_{2}}\cdots A_{n}^{d_{n}}z^{2g}\right]S\left(\left|\bm{A}\right|A_{1}z\right)\prod_{i=1}^{n}\left(\frac{1}{S\left(A_{i}z\right)}\right)\frac{1}{A_{1}A_{2}^{2}\cdots A_{n}^{2}z^{2n-2}}\label{eq: Before expo}\\
 & \,\,\,\,\,\,\,\,\,\times\prod_{i=2}^{n}\left(\prod_{j=1}^{i-1}\left(\sum_{q_{\left\{ i,j\right\} }\geq0}\frac{A_{i}^{q_{\left\{ i,j\right\} }}A_{j}^{q_{\left\{ i,j\right\} }}z^{2q_{\left\{ i,j\right\} }}}{q_{\left\{ i,j\right\} }!}\sum_{\underset{\text{with\,conditions\,\ensuremath{\beta}}}{k_{1}^{I},\dots,k_{q_{\left\{ i,j\right\} }}^{I}>0}}k_{1}^{\left\{ i,j\right\} }\cdots k_{q_{\left\{ i,j\right\} }}^{\left\{ i,j\right\} }\tilde{W}^{\left\{ i,j\right\} }\left(k_{1}^{\left\{ i,j\right\} },\dots,k_{q_{\left\{ i,j\right\} }}^{\left\{ i,j\right\} },z_{i},z_{j}\right)\right)-1\right).\nonumber 
\end{align}
Note that when $\sum_{j=1}^{i-1}q_{\left\{ i,j\right\} }=0,\,{\rm for\,any}\,2\leq i\leq n$,
the product running over $j$ equals $1$ so that conditions~$\gamma$
are satisfied. Then we rewrite the last line of this expression as
the coefficient of an exponential series. Using the expression of
$\tilde{W}^{\left\{ i,j\right\} }\left(k_{1}^{\left\{ i,j\right\} },\dots,k_{q_{\left\{ i,j\right\} }}^{\left\{ i,j\right\} },z_{i},z_{j}\right)$
and conditions $\beta$, we get
\begin{align*}
 & \prod_{j=2}^{n}\left(\prod_{i=1}^{j-1}\left(\sum_{q_{\left\{ i,j\right\} }\geq0}\frac{A_{i}^{q_{\left\{ i,j\right\} }}A_{j}^{q_{\left\{ i,j\right\} }}z^{2q_{\left\{ i,j\right\} }}}{q_{\left\{ i,j\right\} }!}\sum_{\underset{\text{with\,conditions\,\ensuremath{\beta}}}{k_{1}^{I},\dots,k_{q_{\left\{ i,j\right\} }}^{I}>0}}k_{1}^{\left\{ i,j\right\} }\cdots k_{q_{\left\{ i,j\right\} }}^{\left\{ i,j\right\} }\tilde{W}^{\left\{ i,j\right\} }\left(k_{1}^{\left\{ i,j\right\} },\dots,k_{q_{\left\{ i,j\right\} }}^{\left\{ i,j\right\} },z_{i},z_{j}\right)\right)-1\right)\\
= & \left[t_{2}^{A_{2}}\cdots t_{n}^{A_{n}}\right]\left.\prod_{i=2}^{n}\left(\prod_{j=1}^{i-1}\exp\left(A_{i}A_{j}z^{2}\sum_{k>0}kS\left(kA_{i}z_{i}\right)S\left(kA_{j}z_{j}\right)\left(\frac{t_{i}}{t_{j}}\right)^{k}\right)-1\right)\right|_{t_{1}=1}.
\end{align*}
Substituting this expression in Eq.~(\ref{eq: Before expo}), we
get
\begin{align}
\langle\tau_{0}\tau_{d_{1}}\dots\tau_{d_{n}}\rangle_{0,g}= & \sqrt{-1}^{n-2-\left|\bm{d}\right|}\left[A_{1}^{d_{1}}A_{2}^{d_{2}}\cdots A_{n}^{d_{n}}z^{2g}\right]S\left(\left|\bm{A}\right|A_{1}z\right)\prod_{i=1}^{n}\left(\frac{1}{S\left(A_{i}z\right)}\right)\frac{1}{A_{1}A_{2}^{2}\cdots A_{n}^{2}z^{2n-2}}\nonumber \\
 & \times\left[t_{2}^{A_{2}}\cdots t_{n}^{A_{n}}\right]\left.\prod_{i=2}^{n}\left(\prod_{j=1}^{i-1}\exp\left(A_{i}A_{j}z^{2}\sum_{k>0}kS\left(kA_{i}z_{i}\right)S\left(kA_{j}z_{j}\right)\left(\frac{t_{i}}{t_{j}}\right)^{k}\right)-1\right)\right|_{t_{1}=1}.\label{eq: end step 4}
\end{align}

\paragraph{Step 5.}

The second line of Eq.~(\ref{eq: end step 4}) is simplified using
the following property. 
\begin{prop}
[The products of exponentials formula]\label{prop: The combinatorial prop of the main proof}Fix
$n$ positive integers $A_{1},...,A_{n}$. Fix $n$ formal variables
$t_{1},\dots,t_{n}$ ; by convention, let $t_{1}=1$. We have
\begin{align*}
 & \left[t_{2}^{A_{2}}\cdots t_{n}^{A_{n}}\right]\prod_{i=2}^{n}\left(\prod_{j=1}^{i-1}\exp\left(A_{i}A_{j}z^{2}\sum_{k>0}kS\left(kA_{i}z\right)S\left(kA_{j}z\right)\left(\frac{t_{i}}{t_{j}}\right)^{k}\right)-1\right)\\
= & A_{1}A_{2}^{2}\cdots A_{n}^{2}z^{2n-2}\left(\sum_{i=1}^{n}A_{i}\right)^{n-2}\frac{\prod_{i=1}^{n}S\left(A_{i}z\right)}{S\left(A_{1}z+\cdots+A_{n}z\right)}\prod_{r=2}^{n}S\left(A_{r}\left(A_{1}z+\cdots+A_{n}z\right)\right).
\end{align*}
\end{prop}
This proposition is proved in Section~\ref{sec:A-combinatorial-identity}
using properties of Eulerian numbers. 

According to this proposition, we get

\begin{align*}
\langle\tau_{0}\tau_{d_{1}}\dots\tau_{d_{n}}\rangle_{0,g}= & \sqrt{-1}^{n-2-\left|\bm{d}\right|}\left[A_{1}^{d_{1}}A_{2}^{d_{2}}\cdots A_{n}^{d_{n}}z^{2g}\right]S\left(\left|\bm{A}\right|A_{1}z\right)\prod_{i=1}^{n}\left(\frac{1}{S\left(A_{i}z\right)}\right)\frac{1}{A_{1}A_{2}^{2}\cdots A_{n}^{2}z^{2n-2}}\\
 & \times A_{1}A_{2}^{2}\cdots A_{n}^{2}z^{2n-2}\left|\bm{A}\right|^{n-2}\frac{\prod_{i=1}^{n}S\left(A_{i}z\right)}{S\left(\left|\bm{A}\right|z\right)}\prod_{r=2}^{n}S\left(A_{r}\left|\bm{A}\right|z\right)
\end{align*}
that we simplify as
\begin{align*}
\langle\tau_{0}\tau_{d_{1}}\dots\tau_{d_{n}}\rangle_{0,g}= & \sqrt{-1}^{n-2-\left|\bm{d}\right|}\left[A_{1}^{d_{1}}A_{2}^{d_{2}}\cdots A_{n}^{d_{n}}z^{2g}\right]\frac{\left|\bm{A}\right|^{n-2}}{S\left(\left|\bm{A}\right|z\right)}\prod_{r=1}^{n}S\left(A_{r}\left|\bm{A}\right|z\right).
\end{align*}
Finally, with the change of variable $z:=\frac{z}{\left|\bm{A}\right|}$,
we get 
\begin{align*}
\langle\tau_{0}\tau_{d_{1}}\dots\tau_{d_{n}}\rangle_{0,g}= & \sqrt{-1}^{n-2-\left|\bm{d}\right|}\left[A_{1}^{d_{1}}A_{2}^{d_{2}}\cdots A_{n}^{d_{n}}z^{2g}\right]\left|\bm{A}\right|^{2g+n-2}\frac{1}{S\left(z\right)}\prod_{r=1}^{n}S\left(A_{r}z\right)
\end{align*}
and we recognize the expression of $\left\langle \left\langle \tau_{0}\tau_{d_{1}}\dots\tau_{d_{n}}\right\rangle \right\rangle _{g}$
given by Expression~(\ref{eq:RHS thm}).

\section{\label{sec:A-combinatorial-identity}Proof of the products of the
exponentials formula}

The purpose of this section is to prove Proposition~\ref{prop: The combinatorial prop of the main proof}
which ends the proof of the main theorem. To do so, we first use Corollary~\ref{prop:Eulerian fin}
which follows from Eulerian numbers properties. We get
\begin{align*}
 & \left[t_{2}^{A_{2}}\cdots t_{n}^{A_{n}}\right]\prod_{i=2}^{n}\left(\prod_{j=1}^{i-1}\exp\left(A_{i}A_{j}z^{2}\sum_{k>0}kS\left(kA_{i}z\right)S\left(kA_{j}z\right)\left(\frac{t_{i}}{t_{j}}\right)^{k}\right)-1\right)\\
 & =\left[t_{2}^{A_{2}}\cdots t_{n}^{A_{n}}\right]\prod_{i=2}^{n}\left(\prod_{j=1}^{i-1}\left(1+4\sum_{k>0}\frac{{\rm sh}\left(\frac{A_{i}}{2}z\right){\rm sh}\left(\frac{A_{j}}{2}z\right)}{{\rm sh}\left(\frac{A_{i}+A_{j}}{2}z\right)}{\rm sh}\left(k\frac{A_{i}+A_{j}}{2}z\right)\left(\frac{t_{i}}{t_{j}}\right)^{k}\right)-1\right),
\end{align*}
where we used the convention $t_{1}=1$. Recall that $A_{1},\dots,A_{n}$
are positive integers and $t_{2},\dots,t_{n},z$ are formal variables.
\begin{prop}
\label{prop: combinatorial identity 1}Fix $\left(n-1\right)$ positive
integers $a_{2},...,a_{n}$ and $n$ formal variables $\tilde{A}_{1},...,\tilde{A}_{n}$.
Fix $\left(n-1\right)$ more formal variables $t_{2},\dots,t_{n}$
; by convention, let $t_{1}=1$. We have
\begin{align*}
 & \left[t_{2}^{a_{2}}\cdots t_{n}^{a_{n}}\right]\prod_{i=2}^{n}\left(\prod_{j=1}^{i-1}\left(1+4\sum_{k>0}\frac{{\rm sh}\left(\tilde{A}_{i}\right){\rm sh}\left(\tilde{A}_{j}\right)}{{\rm sh}\left(\tilde{A}_{i}+\tilde{A}_{j}\right)}{\rm sh}\left(k\left(\tilde{A}_{i}+\tilde{A}_{j}\right)\right)\left(\frac{t_{i}}{t_{j}}\right)^{k}\right)-1\right)\\
 & =2^{2\left(n-1\right)}\frac{\prod_{i=1}^{n}{\rm sh}\left(\tilde{A}_{i}\right)}{{\rm sh}\left(\tilde{A}_{1}+\cdots+\tilde{A}_{n}\right)}\left(\prod_{r=2}^{n}\left({\rm sh}\left(a_{r}\left(\tilde{A}_{1}+...+\tilde{A}_{r}\right)+\tilde{A}_{r}\left(a_{r+1}+...+a_{n}\right)\right)\right)\right).
\end{align*}
\end{prop}
Before proving this proposition, let us end the proof of the products
of exponentials formula. According to this proposition by substituting
$\tilde{A}_{i}:=\frac{A_{i}}{2}z$ and $a_{i}:=A_{i}$, we find
\begin{align*}
 & \left[t_{2}^{A_{2}}\cdots t_{n}^{A_{n}}\right]\prod_{i=2}^{n}\left(\prod_{j=1}^{i-1}\left(1+4\sum_{k>0}\frac{{\rm sh}\left(\frac{A_{i}}{2}z\right){\rm sh}\left(\frac{A_{j}}{2}z\right)}{{\rm sh}\left(\frac{A_{i}+A_{j}}{2}z\right)}{\rm sh}\left(k\frac{A_{i}+A_{j}}{2}z\right)\left(\frac{t_{i}}{t_{j}}\right)^{k}\right)-1\right)\\
 & =2^{2\left(n-1\right)}\frac{\prod_{i=1}^{n}{\rm sh}\left(\frac{A_{i}}{2}z\right)}{{\rm sh}\left(\frac{A_{1}+\cdots+A_{n}}{2}z\right)}\prod_{r=2}^{n}{\rm sh}\left(\frac{A_{r}\left(A_{1}+...+A_{n}\right)}{2}z\right)\\
 & =A_{1}A_{2}^{2}\cdots A_{n}^{2}z^{2n-2}\left(\sum_{i=1}^{n}A_{i}\right)^{n-2}\frac{\prod_{i=1}^{n}S\left(A_{i}z\right)}{S\left(A_{1}z+\cdots+A_{n}z\right)}\prod_{r=2}^{n}S\left(A_{r}\left(A_{1}z+\cdots+A_{n}z\right)\right),
\end{align*}
where we used the notation $S\left(z\right)=\frac{{\rm sh}\left(z/2\right)}{z/2}$
to obtain this equality. This proves the products of exponentials
formula. 

\paragraph{Convention.}

In the rest of this section, we make no use of the positive integers
$A_{1},\dots,A_{n}$. However we will intensively use the formal variables
$\tilde{A}_{1},\dots,\tilde{A}_{n}$. For convenience, we change the
notation by removing the tildes on these formal variables. 
\begin{proof}
[Proof of Proposistion \ref{prop: combinatorial identity 1} ] We
prove this formula by induction over $n$. The first step $n=2$ is
obvious. Suppose this induction is proved until step $n$. We prove the $\left(n+1\right)$-th
step. Start from the LHS 
\[
\left[t_{2}^{a_{2}}\cdots t_{n+1}^{a_{n+1}}\right]\prod_{i=2}^{n+1}\left(\prod_{j=1}^{i-1}\left(1+4\sum_{k>0}\frac{{\rm sh}\left(A_{i}\right){\rm sh}\left(A_{j}\right)}{{\rm sh}\left(A_{i}+A_{j}\right)}{\rm sh}\left(k\left(A_{i}+A_{j}\right)\right)\left(\frac{t_{i}}{t_{j}}\right)^{k}\right)-1\right).
\]
We decompose the product in order to use the induction hypothesis
; we move the terms corresponding to $i=n+1$ on a second line and
get
\begin{align*}
 & \left[t_{2}^{a_{2}}\cdots t_{n}^{a_{n}}t_{n+1}^{a_{n+1}}\right]\\
 & \times\prod_{i=2}^{n}\left(\prod_{j=1}^{i-1}\left(1+4\sum_{k>0}\frac{{\rm sh}\left(A_{i}\right){\rm sh}\left(A_{j}\right)}{{\rm sh}\left(A_{i}+A_{j}\right)}{\rm sh}\left(k\left(A_{i}+A_{j}\right)\right)\left(\frac{t_{i}}{t_{j}}\right)^{k}\right)-1\right)\\
 & \times\left(\prod_{s=1}^{n}\left(1+4\sum_{l>0}\frac{{\rm sh}\left(A_{s}\right){\rm sh}\left(A_{n+1}\right)}{{\rm sh}\left(A_{s}+A_{n+1}\right)}{\rm sh}\left(k\left(A_{s}+A_{n+1}\right)\right)\left(\frac{t_{n+1}}{t_{s}}\right)^{l}\right)-1\right).
\end{align*}
We simplify the series on the second line of the expression using
the induction hypothesis. We get
\begin{align*}
 & \left[t_{2}^{a_{2}}\cdots t_{n}^{a_{n}}t_{n+1}^{a_{n+1}}\right]\\
 & \times2^{2\left(n-1\right)}\frac{\prod_{i=1}^{n}{\rm sh}\left(A_{i}\right)}{{\rm sh}\left(A_{1}+\cdots+A_{n}\right)}\sum_{i_{2},\dots,i_{n}>0}\prod_{r=2}^{n}\text{sh}\left(i_{r}\left(A_{1}+...+A_{r}\right)+A_{r}\left(i_{r+1}+...+k_{n}\right)\right)t_{r}^{i_{r}}\\
 & \times\left(\prod_{s=1}^{n}\left(1+4\sum_{k>0}\frac{{\rm sh}\left(A_{s}\right){\rm sh}\left(A_{n+1}\right)}{{\rm sh}\left(A_{s}+A_{n+1}\right)}{\rm sh}\left(k\left(A_{s}+A_{n+1}\right)\right)\left(\frac{t_{n+1}}{t_{s}}\right)^{k}\right)-1\right).
\end{align*}
We obtain the result using the following proposition with 
\[
u:=t_{n+1},\,\,b:=a_{n+1},\,\,B:=A_{n+1},\,\,{\rm and}\,\,X_{r}=0\,\,{\rm for}\,\,2\leq r\leq n.
\]
\begin{prop}
[The sinh formula]\label{prop: Main combinatorial proposition}Fix
$n$ positive integers $a_{2},\dots,a_{n},b$ and $2n$ formal variables
$A_{1},\dots,A_{n},B,X_{2},\dots,X_{n}$. Fix $n$ more formal variables
$t_{2},\dots,t_{n},u$ ; by convention, let $t_{1}=1$. Then the coefficient
of $t_{2}^{a_{2}}\cdots t_{n}^{a_{n}}u^{b}$ in the formal power series
\begin{align*}
\sum_{i_{2},\dots,i_{n}>0}\prod_{r=2}^{n}{\rm sh}\left(i_{r}\left(A_{1}+\cdots+A_{r}\right)+A_{r}\left(i_{r+1}+\cdots+i_{n}\right)+X_{r}\right)t_{r}^{i_{r}}\\
\times\left\{ \prod_{s=1}^{n}\left(1+4\sum_{j_{s}>0}\frac{{\rm sh}\left(A_{s}\right){\rm sh}\left(B\right)}{{\rm sh}\left(A_{s}+B\right)}{\rm sh}\left(j_{s}\left(A_{s}+B\right)\right)\left(\frac{u}{t_{s}}\right)^{j_{s}}\right)-1\right\} 
\end{align*}
is\fontsize{10}{12}
\[
4\frac{{\rm sh}\left(A_{1}+\cdots+A_{n}\right){\rm sh}\left(B\right)}{{\rm sh}\left(A_{1}+\cdots+A_{n}+B\right)}\prod_{r=2}^{n}\Big({\rm sh}\left(a_{r}\left(A_{1}+\cdots+A_{r}\right)+A_{r}\left(a_{r+1}+\cdots+a_{n}+b\right)+X_{r}\right)\Big){\rm sh}\left(b\left(A_{1}+\cdots+A_{n}+B\right)\right).
\]
\end{prop}
\end{proof}
The purpose of the rest of this section is to prove the sinh formula.
The proof goes by induction. We prove the case $n=2$ in Section~\ref{subsec: Initialisation combinatorial identity}.
The heredity is proved in Section~\ref{subsec: Heredity combinatorial proposition}. 

\subsection{Proof by induction of the sinh formula: initialization\label{subsec: Initialisation combinatorial identity}}

We prove in this section, the case $n=2$ of the sinh formula. We
begin by a series of lemmas.
\begin{lem}
\label{lem:sinhcosh} Let $\alpha,\beta$ and $\gamma$ be some formal
variables. We have
\[
{\rm ch}\left(\alpha\right){\rm sh}\left(\beta+\gamma\right)-{\rm ch}\left(\beta\right){\rm sh}\left(\alpha+\gamma\right)={\rm sh}\left(\alpha-\beta\right){\rm sh}\left(\gamma\right).
\]
\end{lem}
\begin{proof}
One can check this formula using the basic hyperbolic identities.
\end{proof}
\begin{lem}
\label{lem:sum} Let $\mu$ and $\nu$ be some formal variables. We
have
\begin{align*}
\sum_{j=0}^{b}{\rm sh}\left(\mu j+\nu\right) & =\frac{{\rm sh}\left(\mu\left(b+1\right)/2\right){\rm sh}\left(\mu b/2+\nu\right)}{{\rm sh}\left(\mu/2\right)}\\
 & =\frac{{\rm ch}\left(\mu/2\right)}{{\rm sh}\left(\mu/2\right)}{\rm sh}\left(b\mu/2\right){\rm sh}\left(b\mu/2+\nu\right)+{\rm ch}\left(b\mu/2\right){\rm sh}\left(b\mu/2+\nu\right).
\end{align*}
\end{lem}
\begin{proof}
Using the exponential form of the hyperbolic sine and geometric sums,
we obtain the first equality. The second equality is obtained by using
${\rm sh}\left(\frac{\mu b}{2}+\frac{\mu}{2}\right)={\rm ch}\left(\frac{\mu}{2}\right){\rm sh}\left(\frac{b\mu}{2}\right)+{\rm ch}\left(\frac{b\mu}{2}\right){\rm sh}\left(\frac{\mu}{2}\right)$
and simplifying.
\end{proof}
\begin{lem}
\label{lem:main lemma} Fix two integers $a,b$ and four formal variables
$A_{1},A_{2},B,X$. We have
\begin{align*}
 & \sum_{j=0}^{b}{\rm sh}\Big(\left(a+j\right)\left(A_{1}+A_{2}\right)+X\Bigr){\rm sh}\Big(\left(b-j\right)\left(A_{1}+B\right)\Bigr){\rm sh}\Big(j\left(A_{2}+B\right)\Bigr)\\
 & =\frac{{\rm ch}\left(A_{1}\right)}{{\rm sh}\left(A_{1}\right)}{\rm sh}\left(bA_{1}\right){\rm sh}\left(a\left(A_{1}+A_{2}\right)-bB+X\right)\\
 & +\frac{{\rm ch}\left(A_{2}\right)}{{\rm sh}\left(A_{2}\right)}{\rm sh}\left(bA_{2}\right){\rm sh}\left(a\left(A_{1}+A_{2}\right)+b\left(A_{1}+A_{2}+B\right)+X\right)\\
 & +\frac{{\rm ch}\left(B\right)}{{\rm sh}\left(B\right)}{\rm sh}\left(bB\right){\rm sh}\left(-a\left(A_{1}+A_{2}\right)-bA_{1}-X\right)\\
 & +\frac{{\rm ch}\left(A_{1}+A_{2}+B\right)}{{\rm sh}\left(A_{1}+A_{2}+B\right)}{\rm sh}\left(b\left(A_{1}+A_{2}+B\right)\right){\rm sh}\left(-a\left(A_{1}+A_{2}\right)-bA_{2}-X\right).
\end{align*}
\end{lem}
\begin{proof}
We start from the LHS of this formula. We first linearize the product
of the three hyperbolic sine using
\begin{equation}
{\rm sh}\left(u\right){\rm sh}\left(v\right){\rm sh}\left(w\right)={\rm sh}\left(u+v+w\right)+{\rm sh}\left(u-v-w\right)+{\rm sh}\left(-u+v-w\right)+{\rm sh}\left(-u-v+w\right)\label{eq:1}
\end{equation}
with
\begin{align*}
u & =\left(a+j\right)\left(A_{1}+A_{2}\right)+X\\
v & =\left(b-j\right)\left(A_{1}+B\right)\\
w & =j\left(A_{2}+B\right).
\end{align*}
Then, we use the formula of Lemma~\ref{lem:sum}
\[
\sum_{j=0}^{b}{\rm sh}\left(\mu j+\nu\right)=\frac{{\rm ch}\left(\mu/2\right)}{{\rm sh}\left(\mu/2\right)}{\rm sh}\left(\mu b/2\right){\rm sh}\left(\mu b/2+\nu\right)+{\rm ch}\left(\mu b/2\right){\rm sh}\left(\mu b/2+\nu\right)
\]
to compute the finite sum of each of the four terms in the RHS of
Eq.~(\ref{eq:1}). We find
\noindent \begin{flushleft}
\begin{align}
\sum_{j=0}^{b}{\rm sh}\left(u-v-w\right)= & \frac{{\rm ch}\left(A_{1}\right)}{{\rm sh}\left(A_{1}\right)}{\rm sh}\left(bA_{1}\right){\rm sh}\left(a\left(A_{1}+A_{2}\right)-bB+X\right)\label{ini1}\\
 & +{\rm ch}\left(bA_{1}\right){\rm sh}\left(a\left(A_{1}+A_{2}\right)-bB+X\right)\nonumber \\
\sum_{j=0}^{b}{\rm sh}\left(u+v+w\right)= & \frac{{\rm ch}\left(A_{2}\right)}{{\rm sh}\left(A_{2}\right)}{\rm sh}\left(bA_{2}\right){\rm sh}\left(a\left(A_{1}+A_{2}\right)+b\left(A_{1}+A_{2}+B\right)+X\right)\label{ini2}\\
 & +{\rm ch}\left(bA_{2}\right){\rm sh}\left(a\left(A_{1}+A_{2}\right)+b\left(A_{1}+A_{2}+B\right)+X\right)\nonumber \\
\sum_{j=0}^{b}{\rm sh}\left(-u-v+w\right)= & \frac{{\rm ch}\left(B\right)}{{\rm sh}\left(B\right)}{\rm sh}\left(bB\right){\rm sh}\left(-a\left(A_{1}+A_{2}\right)-bA_{1}-X\right)\label{ini3}\\
 & +{\rm ch}\left(bB\right){\rm sh}\left(-a\left(A_{1}+A_{2}\right)-bA_{1}-X\right)\nonumber \\
\sum_{j=0}^{b}{\rm sh}\left(-u+v-w\right)= & \frac{{\rm ch}\left(A_{1}+A_{2}+B\right)}{{\rm sh}\left(A_{1}+A_{2}+B\right)}{\rm sh}\left(b\left(A_{1}+A_{2}+B\right)\right){\rm sh}\left(-a\left(A_{1}+A_{2}\right)-bA_{2}-X\right)\label{ini4}\\
 & +{\rm ch}\left(b\left(A_{1}+A_{2}+B\right)\right){\rm sh}\left(-a\left(A_{1}+A_{2}\right)-bA_{2}-X\right).\nonumber 
\end{align}
\par\end{flushleft}
Finally, we prove that the sum of the seconds terms in the RHS of
Equations (\ref{ini1}), (\ref{ini2}), (\ref{ini3}), (\ref{ini4})
vanishes. This will end the proof. 

We sum the second term of the RHS of Eq.~(\ref{ini1}) and the second
term of the RHS of Eq.~(\ref{ini3}). Using Lemma~\ref{lem:sinhcosh},
that is 
\[
{\rm ch}\left(\alpha\right){\rm sh}\left(\beta+\gamma\right)-{\rm ch}\left(\beta\right){\rm sh}\left(\alpha+\gamma\right)={\rm sh}\left(\alpha-\beta\right){\rm sh}\left(\gamma\right)
\]
with $\alpha=bA_{1}$, $\beta=-bB$ and $\gamma=a\left(A_{1}+A_{2}\right)+X$,
we find 
\begin{align*}
 & {\rm ch}\left(bA_{1}\right){\rm sh}\left(a\left(A_{1}+A_{2}\right)-bB+X\right)+{\rm ch}\left(bB\right){\rm sh}\left(-a\left(A_{1}+A_{2}\right)-bA_{1}-X\right)\\
 & ={\rm sh}\left(b\left(A_{1}+B\right)\right){\rm sh}\left(a\left(A_{1}+A_{2}\right)+X\right).
\end{align*}
We sum the second term of the RHS of Eq.~(\ref{ini2}) and the second
term of the RHS of Eq.~(\ref{ini4}). Using Lemma~\ref{lem:sinhcosh}
with $\alpha=bA_{2}$, $\beta=b\left(A_{1}+A_{2}+B\right)$ and $\gamma=a\left(A_{1}+A_{2}\right)+X$,
we find
\begin{align*}
 & {\rm ch}\left(bA_{2}\right){\rm sh}\left(a\left(A_{1}+A_{2}\right)+b\left(A_{1}+A_{2}+B\right)+X\right)+{\rm ch}\left(b\left(A_{1}+A_{2}+B\right)\right){\rm sh}\left(-a\left(A_{1}+A_{2}\right)-bA_{2}-X\right)\\
 & =-{\rm sh}\left(b\left(A_{1}+B\right)\right){\rm sh}\left(a\left(A_{1}+A_{2}\right)+X\right).
\end{align*}
Hence, 
\begin{align*}
0= & {\rm ch}\left(bA_{1}\right){\rm sh}\left(a\left(A_{1}+A_{2}\right)-bB+X\right)+{\rm ch}\left(bB\right){\rm sh}\left(-a\left(A_{1}+A_{2}\right)-bA_{1}-X\right)\\
 & +{\rm ch}\left(bA_{2}\right){\rm sh}\left(a\left(A_{1}+A_{2}\right)+b\left(A_{1}+A_{2}+B\right)+X\right)+{\rm ch}\left(b\left(A_{1}+A_{2}+B\right)\right){\rm sh}\left(-a\left(A_{1}+A_{2}\right)-bA_{2}-X\right)
\end{align*}
\end{proof}
\begin{lem}
\label{lem:sinhsinh} Let $\alpha,\beta$ and $\gamma$ be some formal
variables. We have
\[
{\rm sh}\left(\alpha\right){\rm sh}\left(\beta\right)+{\rm sh}\left(\gamma\right){\rm sh}\left(\alpha+\beta+\gamma\right)={\rm sh}\left(\alpha+\gamma\right){\rm sh}\left(\beta+\gamma\right).
\]
\end{lem}
\begin{proof}
One can check this formula using the usual hyperbolic identities.
\end{proof}
\begin{lem}
\label{lem: 4 lignes de cosh et sinh} Let $A_{1},A_{2}$ and $B$
be some formal variables. We have 
\begin{align*}
 & {\rm ch}\left(A_{1}\right){\rm sh}\left(A_{2}\right){\rm sh}\left(B\right){\rm sh}\left(A_{1}+A_{2}+B\right)\\
+ & {\rm sh}\left(A_{1}\right){\rm ch}\left(A_{2}\right){\rm sh}\left(B\right){\rm sh}\left(A_{1}+A_{2}+B\right)\\
+ & {\rm sh}\left(A_{1}\right){\rm sh}\left(A_{2}\right){\rm ch}\left(B\right){\rm sh}\left(A_{1}+A_{2}+B\right)\\
- & {\rm sh}\left(A_{1}\right){\rm sh}\left(A_{2}\right){\rm sh}\left(B\right){\rm ch}\left(A_{1}+A_{2}+B\right)\\
= & {\rm sh}\left(A_{1}+B\right){\rm sh}\left(A_{2}+B\right){\rm sh}\left(A_{1}+A_{2}\right).
\end{align*}
\end{lem}
\begin{proof}
We start from the LHS. We use the hyperbolic identity ${\rm sh}\left(\alpha+\beta\right)={\rm sh}\left(\alpha\right){\rm ch}\left(\beta\right)+{\rm ch}\left(\alpha\right){\rm sh}\left(\beta\right)$.
We sum the two first lines using $\alpha=A_{1}$ and $\beta=A_{2}$,
we obtain 
\[
{\rm sh}\left(A_{1}+A_{2}\right){\rm sh}\left(B\right){\rm sh}\left(A_{1}+A_{2}+B\right).
\]
We sum the two last lines using $\alpha=-B$ and $\beta=A_{1}+A_{2}+B$,
we obtain 
\[
{\rm sh}\left(A_{1}\right){\rm sh}\left(A_{2}\right){\rm sh}\left(A_{1}+A_{2}\right).
\]
Finally we sum these two terms. We factor ${\rm sh}\left(A_{1}+A_{2}\right)$
and use Lemma~\ref{lem:sinhsinh} to obtain the expected result. 
\end{proof}
We now prove the case $n=2$ of the sinh formula.
\begin{prop}
\label{prop:initialisation} Fix two integers $a_{2},b$ and four
formal variables $A_{1},A_{2},B,X$. Fix two more formal variables
$t_{2},u$. Then the coefficient of $t_{2}^{a_{2}}u^{b}$ in the formal
power series
\begin{align*}
 & \sum_{i>0}{\rm sh}\left(i\left(A_{1}+A_{2}\right)+X\right)t_{2}\\
 & \left\{ \left(1+4\sum_{j_{1}>0}\frac{{\rm sh}\left(A_{1}\right){\rm sh}\left(B\right)}{{\rm sh}\left(A_{1}+B\right)}{\rm sh}\left(j_{0}\left(A_{1}+B\right)\right)u^{j_{1}}\right)\left(1+4\sum_{j_{2}>0}\frac{{\rm sh}\left(A_{2}\right){\rm sh}\left(B\right)}{{\rm sh}\left(A_{2}+B\right)}{\rm sh}\left(j_{1}\left(A_{2}+B\right)\right)\left(\frac{u}{t_{2}}\right)^{j_{2}}\right)\right.\\
 & \,\,\,\,\,\,\,-1\Biggl\}
\end{align*}
is
\[
4\frac{{\rm sh}\left(A_{1}+A_{2}\right){\rm sh}\left(B\right)}{{\rm sh}\left(A_{1}+A_{2}+B\right)}\left({\rm sh}\left(a_{2}\left(A_{1}+A_{2}\right)+A_{2}b+X\right)\right){\rm sh}\left(b\left(A_{1}+A_{2}+B\right)\right)
\]
\end{prop}
\begin{proof}
We re-write the power series of the LHS of the proposition as
\[
\Phi\left\{ \left(1+\Delta_{1}\right)\left(1+\Delta_{2}\right)-1\right\} 
\]
with
\begin{align*}
\Phi & =\sum_{i\geq0}{\rm sh}\left(i\left(A_{1}+A_{2}\right)+X\right)t_{2}\\
\Delta_{1} & =4\sum_{j_{1}\geq0}\frac{{\rm sh}\left(A_{1}\right){\rm sh}\left(B\right)}{{\rm sh}\left(A_{1}+B\right)}{\rm sh}\left(j_{0}\left(A_{1}+B\right)\right)u^{j_{1}}\\
\Delta_{2} & =4\sum_{j_{2}\geq0}\frac{{\rm sh}\left(A_{2}\right){\rm sh}\left(B\right)}{{\rm sh}\left(A_{2}+B\right)}{\rm sh}\left(j_{1}\left(A_{2}+B\right)\right)\left(\frac{u}{t_{2}}\right)^{j_{2}}.
\end{align*}
We begin by expanding the expression
\[
\Phi\left\{ \left(1+\Delta_{1}\right)\left(1+\Delta_{2}\right)-1\right\} =\underset{\mbox{term }\left(i\right)}{\underbrace{\Delta_{1}\Phi}}+\underset{\mbox{term }\left(ii\right)}{\underbrace{\Delta_{2}\Phi}}+\underset{\mbox{term }\left(iii\right)}{\underbrace{\Phi\Delta_{1}\Delta_{2}}}.
\]
Then we extract the coefficient of $t_{2}^{a_{2}}u^{b}$ :
\begin{align*}
\mbox{in term}\left(i\right) & =4\frac{{\rm sh}\left(A_{1}\right){\rm sh}\left(B\right)}{{\rm sh}\left(A_{1}+B\right)}{\rm sh}\left(b\left(A_{1}+B\right)\right){\rm sh}\left(a_{2}\left(A_{1}+A_{2}\right)+X\right),\\
\mbox{in term }\left(ii\right) & =4\frac{{\rm sh}\left(A_{2}\right){\rm sh}\left(B\right)}{{\rm sh}\left(A_{2}+B\right)}{\rm sh}\left(b\left(A_{2}+B\right)\right){\rm sh}\left(\left(a_{2}+b\right)\left(A_{1}+A_{2}\right)+X\right),\\
\mbox{in term }\left(iii\right) & =16\frac{{\rm sh}\left(A_{1}\right){\rm sh}\left(B\right)}{{\rm sh}\left(A_{1}+B\right)}\frac{{\rm sh}\left(A_{2}\right){\rm sh}\left(B\right)}{{\rm sh}\left(A_{2}+B\right)}\\
 & \times\sum_{j=0}^{b}{\rm sh}\left(\left(a_{2}+j\right)\left(A_{1}+A_{2}\right)+X\right){\rm sh}\left(\left(b-j\right)\left(A_{1}+B\right)\right){\rm sh}\left(j\left(A_{2}+B\right)\right).
\end{align*}
Observing that $4{\rm sh}\left(B\right)$ appears in terms $\left(i\right)$,
$\left(ii\right)$ and $\left(iii\right)$ but also in the result,
we factor it out. For reasons that will become clear, we also factor
out $\frac{1}{{\rm sh}\left(A_{1}+B\right){\rm sh}\left(A_{2}+B\right){\rm sh}\left(A_{1}+A_{2}+B\right)}$.
Hence, we re-define our three terms by 
\begin{align*}
\mbox{term }\left(i\right) & :={\rm sh}\left(A_{1}\right){\rm sh}\left(A_{2}+B\right){\rm sh}\left(A_{1}+A_{2}+B\right){\rm sh}\left(b\left(A_{1}+B\right){\rm sh}\left(a_{2}\left(A_{1}+A_{2}\right)+X\right)\right),\\
\mbox{term }\left(ii\right) & :={\rm sh}\left(A_{2}\right){\rm sh}\left(A_{1}+B\right){\rm sh}\left(A_{1}+A_{2}+B\right){\rm sh}\left(b\left(A_{2}+B\right){\rm sh}\left(\left(a_{2}+b\right)\left(A_{1}+A_{2}\right)+X\right)\right),\\
\mbox{term }\left(iii\right) & :=4{\rm sh}\left(A_{1}\right){\rm sh}\left(B\right){\rm sh}\left(A_{2}\right){\rm sh}\left(A_{1}+A_{2}+B\right)\\
 & \quad+\sum_{j=0}^{b}{\rm sh}\left(\left(a_{2}+j\right)\left(A_{1}+A_{2}\right)+X\right){\rm sh}\left(\left(b-j\right)\left(A_{1}+B\right)\right){\rm sh}\left(j\left(A_{2}+B\right)\right).
\end{align*}
In Step 1, we develop terms $\left(i\right)$ and $\left(ii\right)$
using a basic hyperbolic identity. We also compute the sum of term
$\left(iii\right)$ using Lemma~\ref{lem:main lemma}. We obtain
$8$ terms from terms $\left(i\right)$, $\left(ii\right)$ and $\left(iii\right)$.
Then in Step~$2$, we combine $7$ of these terms using Lemma~\ref{lem:sinhsinh}.
Finally, in Step~3, we combine all the terms and we use Lemma~\ref{lem: 4 lignes de cosh et sinh}
to obtain the result. 

\paragraph{Step 1.}

\textbf{Term $\left(i\right)$} : using the hyperbolic identity ${\rm sh}\left(A_{2}+B\right)={\rm sh}\left(A_{2}\right){\rm ch}\left(B\right)+{\rm sh}\left(B\right){\rm ch}\left(A_{2}\right)$,
we split term $\left(i\right)$ in a sum of two terms. We get 
\begin{align*}
 & {\rm sh}\left(A_{1}\right){\rm sh}\left(A_{2}\right){\rm ch}\left(B\right){\rm sh}\left(A_{1}+A_{2}+B\right)\times{\rm sh}\left(b\left(A_{1}+B\right)\right){\rm sh}\left(a_{2}\left(A_{1}+A_{2}\right)+X\right)\tag{\ensuremath{ia}}\\
 & +{\rm sh}\left(A_{1}\right){\rm ch}\left(A_{2}\right){\rm sh}\left(B\right){\rm sh}\left(A_{1}+A_{2}+B\right)\times{\rm sh}\left(b\left(A_{1}+B\right)\right){\rm sh}\left(a_{2}\left(A_{1}+A_{2}\right)+X\right)\tag{\ensuremath{ib}}.
\end{align*}

\textbf{Term $\left(ii\right)$} : using the same hyperbolic identity
for ${\rm sh}\left(A_{1}+B\right)$, we split term $\left(ii\right)$
in a sum of two terms. We get
\begin{align*}
 & {\rm sh}\left(A_{1}\right){\rm sh}\left(A_{2}\right){\rm ch}\left(B\right){\rm sh}\left(A_{1}+A_{2}+B\right)\times{\rm sh}\left(b\left(A_{2}+B\right)\right){\rm sh}\left(\left(a_{2}+b\right)\left(A_{1}+A_{2}\right)+X\right)\tag{\ensuremath{iia}}\\
 & +{\rm ch}\left(A_{1}\right){\rm sh}\left(A_{2}\right){\rm sh}\left(B\right){\rm sh}\left(A_{1}+A_{2}+B\right)\times{\rm sh}\left(b\left(A_{2}+B\right)\right){\rm sh}\left(\left(a_{2}+b\right)\left(A_{1}+A_{2}\right)+X\right)\tag{\ensuremath{iib}}.
\end{align*}

\textbf{Term $\left(iii\right)$} : we use Lemma~\ref{lem:main lemma}
to compute the sum. This gives four terms 
\begin{align*}
 & {\rm ch}\left(A_{1}\right){\rm sh}\left(A_{2}\right){\rm sh}\left(B\right){\rm sh}\left(A_{1}+A_{2}+B\right)\times{\rm sh}\left(bA_{1}\right){\rm sh}\left(a_{2}\left(A_{1}+A_{2}\right)-bB+X\right)\tag{\ensuremath{iiia}}\\
 & +{\rm sh}\left(A_{1}\right){\rm ch}\left(A_{2}\right){\rm sh}\left(B\right){\rm sh}\left(A_{1}+A_{2}+B\right)\times{\rm sh}\left(bA_{2}\right){\rm sh}\left(a_{2}\left(A_{1}+A_{2}\right)+b\left(A_{1}+A_{2}+B\right)+X\right)\tag{\ensuremath{iiib}}\\
 & +{\rm sh}\left(A_{1}\right){\rm sh}\left(A_{2}\right){\rm ch}\left(B\right){\rm sh}\left(A_{1}+A_{2}+B\right)\times{\rm sh}\left(bB\right){\rm sh}\left(-a_{2}\left(A_{1}+A_{2}\right)-bA_{1}-X\right)\tag{\ensuremath{iiic}}\\
 & +{\rm sh}\left(A_{1}\right){\rm sh}\left(A_{2}\right){\rm sh}\left(B\right){\rm ch}\left(A_{1}+A_{2}+B\right)\times{\rm sh}\left(b\left(A_{1}+A_{2}+B\right)\right){\rm sh}\left(-a_{2}\left(A_{1}+A_{2}\right)-bA_{2}-X\right)\tag{\ensuremath{iiid}}.
\end{align*}

\paragraph{Step 2.}

We now combine the terms $\left(ia\right)$ until $\left(iiic\right)$.
We will do so using the formula of Lemma~\ref{lem:sinhsinh}, that
is 
\[
{\rm sh}\left(\alpha\right){\rm sh}\left(\beta\right)+{\rm sh}\left(\gamma\right){\rm sh}\left(\alpha+\beta+\gamma\right)={\rm sh}\left(\alpha+\gamma\right){\rm sh}\left(\beta+\gamma\right).
\]
\begin{itemize}
\item We sum terms $\left(iiia\right)$ and $\left(iib\right)$. They have
the common factor ${\rm ch}\left(A_{1}\right){\rm sh}\left(A_{2}\right){\rm sh}\left(B\right){\rm sh}\left(A_{1}+A_{2}+B\right)$.
We get
\begin{align*}
 & {\rm ch}\left(A_{1}\right){\rm sh}\left(A_{2}\right){\rm sh}\left(B\right){\rm sh}\left(A_{1}+A_{2}+B\right)\\
 & \times\left[{\rm sh}\left(bA_{1}\right){\rm sh}\left(a_{2}\left(A_{1}+A_{2}\right)-bB+X\right)+{\rm sh}\left(b\left(A_{2}+B\right)\right){\rm sh}\left(\left(a_{2}+b\right)\left(A_{1}+A_{2}\right)+X\right)\right].
\end{align*}
Then simplify the expression appearing inside the brackets using Lemma~\ref{lem:sinhsinh}
with $\alpha=bA_{1}$, $\beta=a_{2}\left(A_{1}+A_{2}\right)-bB+X$
and $\gamma=b\left(A_{2}+B\right)$. We find
\begin{equation}
{\rm ch}\left(A_{1}\right){\rm sh}\left(A_{2}\right){\rm sh}\left(B\right){\rm sh}\left(A_{1}+A_{2}+B\right)\times{\rm sh}\left(b\left(A_{1}+A_{2}+B\right)\right){\rm sh}\left(a_{2}\left(A_{1}+A_{2}\right)+bA_{2}+X\right).\label{eq: Premier Term}
\end{equation}
\item We sum terms $\left(ib\right)$ and $\left(iiib\right)$ using the
same computation. They have the common factor \\
${\rm sh}\left(A_{1}\right){\rm ch}\left(A_{2}\right){\rm sh}\left(B\right){\rm sh}\left(A_{1}+A_{2}+B\right)$.
We get
\begin{align*}
 & {\rm sh}\left(A_{1}\right){\rm ch}\left(A_{2}\right){\rm sh}\left(B\right){\rm sh}\left(A_{1}+A_{2}+B\right)\\
 & \times\left[{\rm sh}\left(b\left(A_{1}+B\right)\right){\rm sh}\left(a_{2}\left(A_{1}+A_{2}\right)+X\right)+{\rm sh}\left(bA_{2}\right){\rm sh}\left(a_{2}\left(A_{1}+A_{2}\right)+b\left(A_{1}+A_{2}+B\right)+X\right)\right].
\end{align*}
Then simplify the expression appearing inside the brackets using Lemma~\ref{lem:sinhsinh}
with $\alpha=b\left(A_{1}+B\right)$, $\beta=a_{2}\left(A_{1}+A_{2}\right)+X$,
$\gamma=bA_{2}$. We find
\begin{equation}
{\rm sh}\left(A_{1}\right){\rm ch}\left(A_{2}\right){\rm sh}\left(B\right){\rm sh}\left(A_{1}+A_{2}+B\right)\times{\rm sh}\left(b\left(A_{1}+A_{2}+B\right)\right){\rm sh}\left(a_{2}\left(A_{1}+A_{2}\right)+bA_{2}+X\right).\label{eq: Deuxieme Term}
\end{equation}
\item We sum the terms $\left(ia\right)$ and $\left(iia\right)$ and $\left(iiic\right)$.
They have the common factor\\
${\rm sh}\left(A_{1}\right){\rm sh}\left(A_{2}\right){\rm ch}\left(B\right){\rm sh}\left(A_{1}+A_{2}+B\right)$.
First re-write $\left(ia\right)+\left(iiic\right)$ as
\begin{align*}
 & {\rm sh}\left(A_{1}\right){\rm sh}\left(A_{2}\right){\rm ch}\left(B\right){\rm sh}\left(A_{1}+A_{2}+B\right)\\
\times & \left[{\rm sh}\left(b\left(A_{1}+B\right)\right){\rm sh}\left(a_{2}\left(A_{1}+A_{2}\right)+X\right)+{\rm sh}\left(bB\right){\rm sh}\left(-a_{2}\left(A_{1}+A_{2}\right)-bA_{1}-X\right)\right].
\end{align*}
We then apply Lemma~\ref{lem:sinhsinh} with $\alpha=b\left(A_{1}+B\right)$,
$\beta=a_{2}\left(A_{1}+A_{2}\right)+X$ and $\gamma=-bB$. We get
\[
{\rm sh}\left(A_{1}\right){\rm sh}\left(A_{2}\right){\rm ch}\left(B\right){\rm sh}\left(A_{1}+A_{2}+B\right)\times{\rm sh}\left(bA_{1}\right){\rm sh}\left(a_{2}\left(A_{1}+A_{2}\right)-bB+X\right).
\]
Then, we add the term $\left(iia\right)$ to this expression, we get
\begin{align*}
 & {\rm sh}\left(A_{1}\right){\rm sh}\left(A_{2}\right){\rm ch}\left(B\right){\rm sh}\left(A_{1}+A_{2}+B\right)\\
 & \times\left[{\rm sh}\left(bA_{1}\right){\rm sh}\left(a_{2}\left(A_{1}+A_{2}\right)-bB+X\right)+{\rm sh}\left(b\left(A_{2}+B\right)\right){\rm sh}\left(\left(a_{2}+b\right)\left(A_{1}+A_{2}\right)+X\right)\right].
\end{align*}
Finally, we use Lemma~\ref{lem:sinhsinh} with $\alpha=bA_{1}$,
$\beta=a_{2}\left(A_{1}+A_{2}\right)-bB+X$ and $\gamma=b\left(A_{2}+B\right)$.
We find 
\begin{equation}
{\rm sh}\left(A_{1}\right){\rm sh}\left(A_{2}\right){\rm ch}\left(B\right){\rm sh}\left(A_{1}+A_{2}+B\right)\times{\rm sh}\left(b\left(A_{1}+A_{2}+B\right)\right){\rm sh}\left(a_{2}\left(A_{1}+A_{2}\right)+bA_{2}+X\right).\label{eq: Troisieme Term}
\end{equation}
\end{itemize}

\paragraph{Step 3.}

We sum the three terms (\ref{eq: Premier Term}), (\ref{eq: Deuxieme Term})
and (\ref{eq: Troisieme Term}) obtained in Step 2 with the remaining
term of Step~$1$, that is term $\left(iiid\right)$. These four
terms have the common factor ${\rm sh}\left(b\left(A_{1}+A_{2}+B\right)\right){\rm sh}\left(a_{2}\left(A_{1}+A_{2}\right)+bA_{2}+X\right)$.
We factor it. The sum of the four remaining terms is the sum of the
LHS of Lemma~\ref{lem: 4 lignes de cosh et sinh}. Using this Lemma,
we get
\[
{\rm sh}\left(A_{1}+B\right){\rm sh}\left(A_{2}+B\right){\rm sh}\left(A_{1}+A_{2}\right)\times{\rm sh}\left(b\left(A_{1}+A_{2}+B\right)\right){\rm sh}\left(a_{2}\left(A_{1}+A_{2}\right)+bA_{2}+X\right).
\]

Before re-defining terms $\left(i\right),\left(ii\right)$ and $\left(iii\right)$
we factored out $\frac{4{\rm sh}\left(B\right)}{{\rm sh}\left(A_{1}+B\right){\rm sh}\left(A_{2}+B\right){\rm sh}\left(A_{1}+A_{2}+B\right)}$.
Multiplying this factor with the expression we just obtained, we get
the result.
\end{proof}

\subsection{Proof by induction of the sinh formula: heredity\label{subsec: Heredity combinatorial proposition}}

Let us recall the sinh formula before proving it. 

Fix $n$ positive integers $a_{2},\dots,a_{n},b$ and $2n$ formal
variables $A_{1},\dots,A_{n},B,X_{2},\dots,X_{n}$. Fix $n$ more
formal variables $t_{2},\dots,t_{n},u$ ; by convention, let $t_{1}=1$.
The coefficient of $t_{2}^{a_{2}}\dots t_{n}^{a_{n}}u^{b}$ in the
formal power series
\begin{align*}
\sum_{i_{2},\dots,i_{n}>0}\prod_{r=2}^{n}{\rm sh}\left(i_{r}\left(A_{1}+\cdots+A_{r}\right)+A_{r}\left(i_{r+1}+\cdots+i_{n}\right)+X_{r}\right)t_{r}^{i_{r}}\\
\times\left\{ \prod_{s=1}^{n}\left(1+4\sum_{j_{s}>0}\frac{\left(A_{s}\right){\rm sh}\left(B\right)}{{\rm sh}\left(A_{s}+B\right)}{\rm sh}\left(j_{s}\left(A_{s}+B\right)\right)\left(\frac{u}{t_{s}}\right)^{j_{s}}\right)-1\right\} 
\end{align*}
is\fontsize{10}{12}
\[
4\frac{{\rm sh}\left(A_{1}+\cdots+A_{n}\right){\rm sh}\left(B\right)}{{\rm sh}\left(A_{1}+\cdots+A_{n}+B\right)}\prod_{r=2}^{n}\Big({\rm sh}\left(a_{r}\left(A_{1}+\cdots+A_{r}\right)+A_{r}\left(a_{r+1}+\cdots+a_{n}+b\right)+X_{r}\right)\Big){\rm sh}\left(b\left(A_{1}+\cdots+A_{n}+B\right)\right).
\]

\begin{proof}
We prove this formula by induction over $n$. The first step, for
$n=2$, is proved by Proposition~\ref{prop:initialisation} in the
preceding section.

We now prove the $n$th step by induction. We can schematically write
the formula of the LHS of the proposition as
\[
\Phi\left\{ \Psi\left(1+\Sigma\right)-1\right\} ,
\]
with 
\begin{align*}
\Phi & =\sum_{i_{2},...,i_{n}>0}\prod_{r=2}^{n}\text{sh}\left(i_{r}\left(A_{1}+\cdots+A_{r}\right)+A_{r}\left(i_{r+1}+\cdots+i_{n}\right)+X_{r}\right)t_{r}^{i_{r}}\\
\Psi & =\prod_{s=1}^{n-1}\left(1+4\sum_{j_{s}>0}\frac{\text{sh}\left(A_{s}\right)\text{sh}\left(B\right)}{\text{sh}\left(A_{s}+B\right)}\text{sh}\left(j_{s}\left(A_{s}+B\right)\right)\left(\frac{u}{t_{s}}\right)^{j_{s}}\right)\\
\Sigma & =4\sum_{j_{n}>0}\frac{\text{sh}\left(A_{n}\right)\text{sh}\left(B\right)}{\text{sh}\left(A_{n}+B\right)}\text{sh}\left(j_{n}\left(A_{n}+B\right)\right)\left(\frac{u}{t_{n}}\right)^{j_{n}}.
\end{align*}
We split the expression in three terms :
\[
\Psi\left(1+\Sigma\right)-1=\underset{term\,1}{\underbrace{\Sigma}}+\underset{term\,2}{\underbrace{\left(\Psi-1\right)}}+\underset{term\,3}{\underbrace{\Sigma\left(\Psi-1\right)}}.
\]

We now extract the coefficient $t_{2}^{a_{2}}\cdots t_{n}^{a_{n}}u^{b}$
in the three terms coming from this development, that is from $\Phi\Psi$,
$\Phi\left(\Psi-1\right)$ and $\Phi\Sigma\left(\Psi-1\right)$. In
the second and third term, we will need the $\left(n-1\right)$th
step of the induction to extract this coefficient. Then, we will sum
these three coefficients, see this sum as the coefficient of a series
and use the first step of the induction to conclude. 

\subparagraph{Term 1.}

We extract the coefficient of $t_{2}^{a_{2}}\cdots t_{n}^{a_{n}}u^{b}$
in $\Phi\Sigma$. To do this, we remove the summations and substitute
$i_{2}:=a_{2},\dots,i_{n-1}=a_{n-1},i_{n}=a_{n}+b,j_{n}=b$, we get
\begin{align*}
 & \prod_{r=2}^{n}{\rm sh}\left(a_{r}\left(A_{1}+\cdots+A_{r}\right)+A_{r}\left(a_{r+1}+\cdots+a_{n}+b\right)+X_{r}\right)\\
 & \times4\frac{{\rm sh}\left(A_{n}\right){\rm sh}\left(B\right)}{{\rm sh}\left(A_{n}+B\right)}{\rm sh}\left(b\left(A_{n}+B\right)\right).
\end{align*}
For reasons that will become clear later, we move the factor $r=n$
of the product to the second line :
\begin{align*}
 & \prod_{r=2}^{n-1}{\rm sh}\left(a_{r}\left(A_{1}+\cdots+A_{r}\right)+A_{r}\left(a_{r+1}+\cdots+a_{n}+b\right)+X_{r}\right)\\
 & \times{\rm sh}\left(\left(a_{n}+b\right)\left(A_{2}+\cdots+A_{n}\right)+X_{n}\right)\times4\frac{{\rm sh}\left(A_{n}\right){\rm sh}\left(B\right)}{{\rm sh}\left(A_{n}+B\right)}{\rm sh}\left(b\left(A_{n}+B\right)\right).
\end{align*}

\subparagraph{Term 2.}

We want to extract the coefficient of $t_{2}^{a_{2}}\cdots t_{n}^{a_{n}}u^{b}$
in $\Phi\left(\Psi-1\right)$. First we extract the coefficient of
$t_{n}^{a_{n}}$. We get
\begin{align*}
 & \sum_{i_{2},\dots,i_{n-1}>0}\prod_{r=2}^{n-1}{\rm sh}\left(i_{r}\left(A_{1}+\cdots+A_{r}\right)+A_{r}\left(i_{r+1}+\cdots+i_{n-1}\right)+A_{r}a_{n}+X_{r}\right)t_{r}^{i_{r}}\times{\rm sh}\left(a_{n}\left(A_{1}+\cdots+A_{n}\right)+X_{n}\right)\\
 & \times\left\{ \prod_{s=1}^{n-1}\left(1+4\sum_{j_{s}>0}\frac{{\rm sh}\left(A_{s}\right){\rm sh}\left(B\right)}{{\rm sh}\left(A_{s}+B\right)}{\rm sh}\left(k_{s}\left(A_{s}+B\right)\right)\left(\frac{u}{t_{s}}\right)^{k_{s}}\right)-1\right\} 
\end{align*}
and we re-arrange the product as 
\begin{align*}
 & {\rm sh}\left(a_{n}\left(A_{1}+\cdots+A_{n}\right)+X_{n}\right)\times\left[\sum_{i_{2},\dots,i_{n-1}>0}\prod_{r=2}^{n-1}{\rm sh}\left(i_{r}\left(A_{1}+\cdots+A_{r}\right)+A_{r}\left(i_{r+1}+\cdots+i_{n-1}\right)+A_{r}a_{n}+X_{r}\right)t_{r}^{i_{r}}\right.\\
 & \left.\times\left\{ \prod_{s=1}^{n-1}\left(1+4\sum_{j_{s}>0}\frac{{\rm sh}\left(A_{s}\right){\rm sh}\left(B\right)}{{\rm sh}\left(A_{s}+B\right)}{\rm sh}\left(k_{s}\left(A_{s}+B\right)\right)\left(\frac{u}{t_{s}}\right)^{k_{s}}\right)-1\right\} \right].
\end{align*}

We then have to extract the coefficient of $t_{2}^{a_{2}}...t_{n-1}^{a_{n-1}}u^{b}$
of the term in squared the bracket. Using the recursion hypothesis
on this term with
\[
A_{1}:=A_{1},\dots,\,A_{n-1}:=A_{n-1},\,B:=B,\,X_{r}:=X_{r}+A_{r}a_{n}
\]
we get
\begin{align*}
 & {\rm sh}\left(a_{n}\left(A_{1}+\cdots+A_{n}\right)+X_{n}\right)\times\left[4\frac{{\rm sh}\left(A_{1}+\cdots+A_{n-1}\right){\rm sh}\left(B\right)}{{\rm sh}\left(A_{1}+\cdots+A_{n-1}+B\right)}\right.\\
 & \left.\prod_{r=2}^{n-1}\Bigl({\rm sh}\left(a_{r}\left(A_{1}+\cdots+A_{r}\right)+A_{r}\left(a_{r+1}+\cdots+a_{n}+b\right)+X_{r}\right)\Bigr){\rm sh}\left(b\left(A_{1}+\cdots+A_{n-1}+B\right)\right)\right].
\end{align*}
Finally, we re-arrange the product as
\begin{align*}
 & \prod_{r=2}^{n-1}\text{sh}\left(a_{r}\left(A_{1}+\cdots+A_{r}\right)+A_{r}\left(a_{r+1}+\cdots+a_{n}+b\right)+X_{r}\right)\\
 & \times4\frac{\text{sh}\left(A_{1}+\cdots+A_{n-1}\right){\rm sh}\left(B\right)}{\text{sh}\left(A_{1}+\cdots+A_{n-1}+B\right)}\text{sh}\left(a_{n}\left(A_{1}+\cdots+A_{n}\right)+X_{n}\right)\text{sh}\left(b\left(A_{1}+\cdots+A_{n-1}+B\right)\right).
\end{align*}

\subparagraph{Term 3.}

We want to extract the coefficient of $t_{2}^{a_{2}}...t_{n}^{a_{n}}u^{b}$
in $\Phi\Sigma\left(\Psi-1\right)$. Here we start by re-arranging
the product as follows :
\begin{align*}
 & \left(4\sum_{j_{n}>0}\frac{{\rm sh}\left(A_{n}\right){\rm sh}\left(B\right)}{{\rm sh}\left(A_{n}+B\right)}{\rm sh}\left(j_{n}\left(A_{n}+B\right)\right)\left(\frac{u}{t_{n}}\right)^{j_{n}}\right)\left(\sum_{i_{n}>0}{\rm sh}\left(i_{n}\left(A_{1}+\cdots+A_{n}\right)+X_{n}\right)t_{n}^{i_{n}}\right)\\
 & \left[\left(\sum_{i_{2},\dots,i_{n-1}>0}\prod_{r=2}^{n-1}{\rm sh}\left(i_{r}\left(A_{1}+\dots+A_{r}\right)+A_{r}\left(i_{r+1}+\cdots+i_{n-1}\right)+A_{r}i_{n}+X_{r}\right)t_{r}^{i_{r}}\right)\right.\\
 & \left.\left\{ \prod_{s=1}^{n-1}\left(1+4\sum_{k_{s}>0}\frac{{\rm sh}\left(A_{s}\right){\rm sh}\left(B\right)}{{\rm sh}\left(A_{s}+B\right)}{\rm sh}\left(k_{s}\left(A_{s}+B\right)\right)\left(\frac{u}{t_{s}}\right)^{k_{s}}\right)-1\right\} \right].
\end{align*}
Now we extract the coefficient of $t_{n}^{a_{n}}u^{b}$. Note that
$t_{n}$ is only present in the first line of the previous expression.
In $\Sigma$, $1/t_{n}$ appears with the same exponent as $u$, that
is $j_{n}$. We then extract the coefficient of $t_{2}^{a_{2}}\cdots t_{n-1}^{a_{n-1}}u^{b-j_{n}}$
from the expression in the square brackets. This is done using the
recursion hypothesis with
\[
A_{1}:=A_{1},\dots,\,A_{n-1}:=A_{n-1},\,B:=B,\,X_{r}:=X_{r}+A_{r}\left(a_{n}-j_{n}\right),
\]
we get 
\begin{align*}
 & \sum_{j_{n}=0}^{b}\left(4\frac{{\rm sh}\left(A_{n}\right){\rm sh}\left(B\right)}{{\rm sh}\left(A_{n}+B\right)}{\rm sh}\left(j_{n}\left(A_{n}+B\right)\right)\right)\left(\sum_{i_{n}>0}{\rm sh}\left(\left(a_{n}+j_{n}\right)\left(A_{1}+\cdots+A_{n}\right)+X_{n}\right)\right)\\
 & \left[4\frac{{\rm sh}\left(A_{1}+\cdots+A_{n-1}\right){\rm sh}\left(B\right)}{{\rm sh}\left(A_{1}+\cdots+A_{n-1}+B\right)}\right.\\
 & \left.\prod_{r=2}^{n-1}\Bigg({\rm sh}\left(a_{r}\left(A_{1}+\cdots+A_{r}\right)+A_{r}\left(a_{r+1}+\cdots+a_{n}+b\right)+X_{r}\right)\Bigg){\rm sh}\left(\left(b-j_{n}\right)\left(A_{1}+\cdots+A_{n-1}+B\right)\right)\right]
\end{align*}
Again, we re-arrange the product
\begin{align*}
 & \prod_{r=2}^{n-1}{\rm sh}\left(a_{r}\left(A_{1}+\cdots+A_{r}\right)+A_{r}\left(a_{r+1}+\cdots+a_{n}+b\right)+X_{r}\right)\\
 & \Biggl(16\frac{{\rm sh}\left(A_{1}+\cdots+A_{n-1}\right){\rm sh}\left(B\right)}{{\rm sh}\left(A_{1}+\cdots+A_{n-1}+B\right)}\frac{{\rm sh}\left(A_{n}\right){\rm sh}\left(B\right)}{{\rm sh}\left(A_{n}+B\right)}\Biggr.\\
 & \Biggl.\sum_{j_{n}=0}^{b}{\rm sh}\left(j_{n}\left(A_{n}+B\right)\right){\rm sh}\left(\left(a_{n}+j_{n}\right)\left(A_{1}+\cdots+A_{n}\right)+X_{n}\right){\rm sh}\left(\left(b-j_{n}\right)\left(A_{1}+\cdots+A_{n-1}+B\right)\right)\Biggr).
\end{align*}
We now combine terms 1, 2 and 3. We factor out $\prod_{r=2}^{n-1}{\rm sh}\left(a_{r}\left(A_{1}+\cdots+A_{r}\right)+A_{r}\left(a_{r+1}+\cdots+a_{n}+b\right)+X_{r}\right)$
in these three terms. We present the rest as the coefficient of a
formal series in $x$ and $y$ as follows 
\begin{align*}
 & \Biggl({\rm sh}\left(\left(a_{n}+b\right)\left(A_{2}+\cdots+A_{n}\right)+X_{n}\right)\times4\frac{{\rm sh}\left(A_{n}\right){\rm sh}\left(B\right)}{{\rm sh}\left(A_{n}+B\right)}{\rm sh}\left(b\left(A_{n}+B\right)\right)\Biggr)\\
+ & \Biggl(4\frac{{\rm sh}\left(A_{1}+\cdots+A_{n-1}\right){\rm sh}\left(B\right)}{{\rm sh}\left(A_{1}+\cdots+A_{n-1}+B\right)}{\rm sh}\left(a_{n}\left(A_{1}+\cdots+A_{n}\right)+X_{n}\right){\rm sh}\left(b\left(A_{1}+\cdots+A_{n-1}+B\right)\right)\Biggr)\\
+ & \Biggl(16\frac{{\rm sh}\left(A_{1}+\cdots+A_{n-1}\right){\rm sh}\left(B\right)}{{\rm sh}\left(A_{1}+\cdots+A_{n-1}+B\right)}\frac{{\rm sh}\left(A_{n}\right){\rm sh}\left(B\right)}{{\rm sh}\left(A_{n}+B\right)}\Biggr.\\
 & \Biggl.\sum_{j_{n}=0}^{b}{\rm sh}\left(j_{n}\left(A_{n}+B\right)\right){\rm sh}\left(\left(a_{n}+j_{n}\right)\left(A_{1}+\cdots+A_{n}\right)+X_{n}\right){\rm sh}\left(\left(b-j_{n}\right)\left(A_{1}+\cdots+A_{n-1}+B\right)\right)\Biggr)\\
=\\
 & \left[x^{a_{n}}y^{b}\right]\sum_{i>0}{\rm sh}\left(i\left(A_{1}+\cdots+A_{n}\right)+X_{n}\right)x^{i}\\
 & \left\{ \left(1+4\sum_{j>0}\frac{{\rm sh}\left(A_{1}+\cdots+A_{n-1}\right){\rm sh}\left(B\right)}{{\rm sh}\left(A_{1}+\cdots+A_{n-1}+B\right)}{\rm sh}\left(j\left(A_{1}+\cdots+A_{n-1}+B\right)\right)y^{j}\right)\right.\\
 & \left.\left(1+4\sum_{k>0}\frac{{\rm sh}\left(A_{n}\right){\rm sh}\left(B\right)}{{\rm sh}\left(A_{n}+B\right)}{\rm sh}\left(k\left(A_{n}+B\right)\right)\left(\frac{y}{x}\right)^{k}\right)-1\right\} .
\end{align*}
We recognize the recursion hypothesis for $n=2$ with $t_{1}:=x,\,t_{2}:=y,\,A_{1}:=a_{n},\,A_{2}:=b$
and
\[
A_{1}:=A_{1}+\cdots+A_{n-1},\,A_{2}:=A_{n},\,B:=B,\,X_{2}:=X_{n}.
\]
Thus the expression above simplifies to
\[
4\frac{{\rm sh}\left(A_{1}+\cdots+A_{n}\right){\rm sh}\left(B\right)}{{\rm sh}\left(A_{1}+\cdots+A_{n}+B\right)}{\rm sh}\left(a_{n}\left(A_{1}+\cdots+A_{n}\right)+A_{n}b+X_{n}\right){\rm sh}\left(b\left(A_{1}+\cdots+A_{n}+B\right)\right).
\]
This ends the proof.
\end{proof}

\section{Proof of the level structure of the correlators\label{sec: Level Structure}}

In this section we prove Proposition~\ref{prop: top vanishing correlators }
that is the vanishing of the correlator $\langle\tau_{d_{1}}\dots\tau_{d_{n}}\rangle_{l,g-l}$
if $\sum d_{i}>4g-3+n-l$ of if $\sum d_{i}$ has the parity of $n-l$.
In Section~\ref{subsec: String equation level structure}, we explain
why it is sufficient to prove these vanishings when the correlator
has a $\tau_{0}$ insertion. The correlator $\langle\tau_{0}\tau_{d_{1}}\dots\tau_{d_{n}}\rangle_{l,g-l}$
is expressed in term of Ehrhart polynomials that we study in Section~\ref{subsec: Ehrart}.
We then deduce the proof of Proposition~\ref{prop: top vanishing correlators }
in Section~\ref{subsec: Proof levels}.

\subsection{String equation\label{subsec: String equation level structure}}

We can rewrite the string equation (Theorem~\ref{thm: string equation})
as the following infinite system of equations
\[
\langle\tau_{0}\tau_{d_{1}}\dots\tau_{d_{n}}\rangle_{l,g-l}=\sum_{i=1}^{n}\langle\tau_{d_{1}}\dots\tau_{d_{i}-1}\dots\tau_{d_{n}}\rangle_{l,g-l},
\]
where $g,l,d_{1},\dots,d_{n}\geq0$ and such that a correlator vanishes
if $\tau$ has a negative index. In these equations, the quantity
defined by the sum of the indices of $\tau$ minus the number of $\tau$
insertions does not depend on the correlator and is equal to $\sum_{i=1}^{n}d_{i}-n-1$.
Moreover, the correlators of the LHS and RHS depend on the same indices
$l$ and $g-l$. We use this system of equations to express any correlator
$\langle\tau_{d_{1}}\dots\tau_{d_{n}}\rangle_{l,g-l}$ as a sum of
correlators with a $\tau_{0}$ insertion. Our two remarks are still
valid: $\langle\tau_{d_{1}}\dots\tau_{d_{n}}\rangle_{l,g-l}$ is expressed
as a sum of correlators with a $\tau_{0}$ insertion such that each
correlator is indexed by $l$ and $g-l$ and the quantity defined
by the sum of the indices of $\tau$ minus the number of $\tau$ insertions
does not depend of the correlator and is equal to $\sum d_{i}-n$
(see Section~\ref{subsec: Inverse string equation} to solve explicitly
this system using of generating series). It is then sufficient to
prove that the correlators $\langle\tau_{0}\tau_{d_{1}}\dots\tau_{d_{n}}\rangle$
vanish if $\sum d_{i}>4g-2+n-l$ or if $\sum d_{i}$ has the parity
of $n-l+1$ in order to prove Proposition~\ref{prop: top vanishing correlators }.%

\subsection{Properties of the Ehrhart polynomials of Buryak and Rossi\label{subsec: Ehrart}}

The vanishing of the correlators come from the properties the following
Ehrhart polynomials. 
\begin{lem}
[\cite{BuryakRossi2016}] Fix a list of $q$ positive integers $\left(r_{1},\dots,r_{q}\right)$.
The function
\[
C^{r_{1},\dots,r_{q}}\left(N\right)=\sum_{k_{1}+\cdots+k_{q}=N}k_{1}^{r_{1}}\cdots k_{q}^{r_{q}}
\]
is a polynomial in $N$ of degree $q-1+\sum r_{i}$. Moreover, this
polynomial has the parity of $q-1+\sum r_{i}$. %
\end{lem}
We then deduce the following lemma.
\begin{lem}
Let $P\left(k_{1},\dots,k_{q}\right)\in\mathbb{C}\left[k_{1},\dots,k_{n}\right]$
be an even (resp. odd) polynomial. Then
\[
\sum_{k_{1}+\cdots+k_{q}=N}k_{1}\cdots k_{q}P\left(k_{1},\dots,k_{q}\right)
\]
is an odd (resp. even) polynomial in the indeterminate $N$ of degree
$2q-1+{\rm deg}P$.
\end{lem}
By induction, we obtain the following lemma.
\begin{lem}
\label{lem: multi Ehrhart polynomials}Fix an integer $n\geq2$ and
a list $A_{2},\dots,A_{n}$ of nonnegative integers. Let $\mathcal{C}$
be the set of pairs (2-element subsets) of $\left\{ 1,\dots,n\right\} $.
Fix another list of nonnegative integers $\left(q_{I},\,I\in\mathcal{C}\right)$.
Let $P\left(k_{i}^{I},I\in\mathcal{C},1\leq i\leq q_{I}\right)$ be
an even (resp. odd) polynomial in the inderterminates $k_{i}^{I}$,
where $I\in\mathcal{C}$ and $1\leq i\leq q_{I}$. Then
\begin{align*}
 & \sum_{\sum_{i=1}^{n-1}K^{\left\{ i,n\right\} }=A_{n}}\,\,\prod_{I\in\mathcal{C}_{n}}k_{1}^{I}\cdots k_{q_{I}}^{I}\\
 & \times\sum_{\sum_{i=1}^{n-2}K^{\left\{ i,n-1\right\} }=A_{n-1}+K^{\left\{ n-1,n\right\} }}\,\,\,\prod_{I\in\mathcal{C}_{n-1}\backslash\mathcal{C}_{n}}k_{1}^{I}\cdots k_{q_{I}}^{I}\\
 & \times\cdots\\
 & \times\sum_{K^{\left\{ 1,2\right\} }=A_{2}+\sum_{j=3}^{n}K^{\left\{ 2,j\right\} }}\,\,\,k_{1}^{\left\{ 1,2\right\} }\cdots k_{q_{\left\{ 1,2\right\} }}^{\left\{ 1,2\right\} }\\
 & \times P\left(k_{i}^{I},I\in\mathcal{C},1\leq i\leq q_{I}\right),
\end{align*}
is a polynomial in the indeterminates $A_{2},\dots,A_{n}$ with the
parity of ${\rm deg}P-\left(n-1\right)$ (resp. ${\rm deg}P-n$) of
degree $2\sum_{I\in\mathcal{C}}q_{I}-\left(n-1\right)+{\rm deg}P$.
We used the notation $K^{I}=\sum_{i=1}^{q_{I}}k_{i}^{I}$. 
\end{lem}

\subsection{Proof of the level structure\label{subsec: Proof levels}}

We now prove Proposition~\ref{prop: top vanishing correlators }.
We first obtain an expression of the correlators $\langle\tau_{0}\tau_{d_{1}}\dots\tau_{d_{n}}\rangle_{l,g-l}$
in term of the polynomials of Lemma~\ref{lem: multi Ehrhart polynomials}.
The level structure follows from the vanishing properties of these
polynomials.
\begin{prop}
Fix three nonnegative integers $n,g,l$ and a list $\left(d_{1},\dots,d_{n}\right)$
of nonnegative integers. We have
\begin{align}
 & \langle\tau_{0}\tau_{d_{1}}\dots\tau_{d_{n}}\rangle_{l,g-l}\nonumber \\
 & =\sqrt{-1}^{n-2-3l-\sum d_{i}}\sum_{\underset{{\rm with\,conditions\,\alpha\,and\,\gamma}}{g_{1}+\cdots+g_{n}+\sum_{I\in\mathcal{C}}q_{I}=g+n-1}}\,\,\,\sum_{\underset{l_{i}\leq g_{i}}{l_{1}+\cdots+l_{n}=l}}\,\,\,\left[a_{1}^{1}\cdots a_{m_{1}}^{1}\cdots a_{1}^{n}\cdots a_{m_{n}}^{n}\right]\nonumber \\
 & \times\sum_{\underset{{\rm with\,conditions}\,\beta}{k_{i}^{I}>0,\,I\in\mathcal{C},\,1\leq i\leq q_{I}}}\,\,\,\prod_{I\in\mathcal{C}}\frac{1}{q_{I}!}k_{1}^{I}\cdots k_{q_{I}}^{I}\nonumber \\
 & \times\frac{1}{m_{1}!}P_{d_{1}-1,g_{1},l_{1}}\left(a_{1}^{1},\dots,a_{m_{1}}^{1},k_{1}^{\left\{ 1,2\right\} },\dots,k_{q_{\left\{ 1,2\right\} }}^{\left\{ 1,2\right\} },\dots,k_{1}^{\left\{ 1,n\right\} },\dots,k_{q_{\left\{ 1,n\right\} }}^{\left\{ 1,n\right\} },-\sum A_{i}\right)\label{eq: Formula correlator in terms of P}\\
 & \times\prod_{i=2}^{n}\frac{1}{m_{i}!}P_{d_{i},g_{i},l_{i}}\left(a_{1}^{i},\dots,a_{m_{1}}^{i},-k_{1}^{\left\{ 1,i\right\} },\dots,-k_{q_{\left\{ 1,i\right\} }}^{\left\{ 1,i\right\} },\dots,k_{1}^{\left\{ i,n\right\} },\dots,k_{q_{\left\{ i,n\right\} }}^{\left\{ i,n\right\} },0\right),\nonumber 
\end{align}
where
\begin{itemize}
\item $\mathcal{C}$ is the set of pairs (2-element subsets) of $\left\{ 1,\dots,n\right\} $;
we also denote by $\mathcal{C}_{i}\subset\mathcal{C}$ the subset
of pairs that contain $i$,
\item the conditions $\alpha$ are
\[
g_{1}+\left(g_{1}-l_{1}\right)+m_{1}+\sum_{I\in\mathcal{C}_{1}}q_{I}=d_{1}+1
\]
and
\[
g_{i}+\left(g_{i}-l_{i}\right)+m_{i}+\sum_{I\in\mathcal{C}_{i}}q_{I}=d_{i}+2,\,\,{\rm when}\,\,i\geq2,
\]
\item the conditions $\beta$ are
\[
A_{i}-\sum_{j=1}^{i-1}K^{\left\{ j,i\right\} }+\sum_{j=i+1}^{n}K^{\left\{ i,j\right\} }=0,\,\,\,\,\,2\leq i\leq n,
\]
where $K^{I}=\sum_{i=1}^{q_{I}}k_{i}^{I}$,
\item the conditions $\gamma$ are
\[
\sum_{j=1}^{i-1}q_{\left\{ i,j\right\} }\geq1,\,{\rm for}\,2\leq i\leq n,
\]
\item the polynomial $P_{d,g,l}\left(x_{1},\dots,x_{m+1}\right)$, with
$d,g,l,m\geq0$, is of degree $2g$ and defined by
\[
P_{d,g,l}\left(x_{1},\dots,x_{m+1}\right)=\int_{{\rm DR}_{g}\left(0,x_{1},\dots,x_{m+1}\right)}\lambda_{l}\psi_{1}^{d+1},
\]
where $\sum_{i=1}^{m+1}x_{i}=0$ and the $\psi$-classes sit on the
marked point with weight $0$ of the double ramification cycle.
\end{itemize}
\end{prop}
\begin{proof}
To obtain this formula we use the definition of the correlators
\[
\langle\tau_{0}\tau_{d_{1}}\dots\tau_{d_{n}}\rangle_{l,g-l}=\left[\epsilon^{2l}\hbar^{g-l}\right]\frac{\sqrt{-1}^{g-l}}{\hbar^{n-1}}\left[\ldots\left[H_{d_{1-1}},\overline{H}_{d_{2}}\right],\dots,\overline{H}_{d_{n}}\right]\Big\vert_{u_{i}=\delta_{i,1}}
\]
and proceed as in Section~\ref{subsec:Less-trivial-case}.

We first need an expression of $H_{d_{1}}\star\overline{H}_{d_{2}}\star\cdots\star\overline{H}_{d_{n}}$.
Start from the following expression of the Hamiltonians
\begin{equation}
H_{d_{i}}=\sum_{g_{i},m_{i},l_{i}\geq0}\frac{\epsilon^{2l_{i}}\left(\sqrt{-1}\hbar\right)^{g_{i}-l_{i}}}{m_{i}!}\sum_{a_{1}^{i},\dots,a_{m_{i}}^{i}\in\mathbb{Z}}P_{d_{i},g_{i},l_{i}}\left(a_{1}^{i},\dots,a_{m_{i}}^{i},-\sum_{j=1}^{m_{i}}a_{j}^{i}\right)p_{a_{1}^{i}}\cdots p_{a_{m_{i}}^{i}}e^{x\sqrt{-1}\sum_{j=1}^{m_{i}}a_{j}^{i}},\label{eq: notation hamiltoniens avec P}
\end{equation}
where the first summation satisfies $g_{i}+\left(g_{i}-l_{i}\right)+m_{i}=d_{i}+2\,\,{\rm and}\,\,l_{i}\leq g_{i}.$
We proceed as in the proof of Proposition~\ref{prop: multi star product}:
we use the associativity of the star product and its developed expression
to obtain Eq.~(\ref{eq: 1st formula multiple star products}). We
then plug the expression of the Hamiltonian given by Eq~(\ref{eq: notation hamiltoniens avec P})
in Eq.~(\ref{eq: 1st formula multiple star products}). When the
$\sum_{j=1,j\neq i}^{n}q_{\left\{ i,j\right\} }$ derivatives act
on $\sum_{a_{1}^{i}+\cdots+a_{m_{i}}^{i}=0}P_{d_{i},g_{i},l_{i}}\left(a_{1}^{i},\dots,a_{m_{i}}^{i},0\right)p_{a_{1}^{i}}\cdots p_{a_{m_{i}}^{i}}$
in $\overline{H}_{d_{i}}$, it remains $\tilde{m}_{i}:=m_{i}-\sum_{j=1,j\neq i}^{n}q_{\left\{ i,j\right\} }$
variables $p$. Similarly, it remains $\tilde{m}_{1}:=m_{1}-\sum_{j=2}^{n}q_{\left\{ 1,j\right\} }$
variables $p$ when the derivatives act on $H_{d_{1}-1}$. We then
obtain
\begin{align}
 & H_{d_{1}}\star\overline{H}_{d_{2}}\star\cdots\star\overline{H}_{d_{n}}\nonumber \\
 & =\sum_{q_{I}\geq0,\,I\in\mathcal{C}}\,\,\,\sum_{g_{1},\dots,g_{n}\geq0}\,\,\,\sum_{\underset{{\rm with\,conditions\,\alpha}}{\tilde{m}_{1},\dots,\tilde{m}_{n}\geq0}}\,\,\,\sum_{l_{1},\dots,l_{n}\geq0}\,\,\,\sum_{\underset{{\rm with\,conditions}\,\beta}{k_{i}^{I}>0,\,I\in\mathcal{C},\,1\leq i\leq q_{I}}}\nonumber \\
 & \times\prod_{I\in\mathcal{C}}\frac{\left(\sqrt{-1}\hbar\right)^{q_{I}}}{q_{I}!}k_{1}^{I}\cdots k_{q_{I}}^{I}\nonumber \\
 & \times\frac{\epsilon^{2l_{1}}\left(\sqrt{-1}\hbar\right)^{g_{1}-l_{1}}}{\tilde{m}_{1}!}P_{d_{1}-1,g_{1},l_{1}}\left(a_{1}^{1},\dots,a_{\tilde{m}_{1}}^{1},k_{1}^{\left\{ 1,2\right\} },\dots,k_{q_{\left\{ 1,2\right\} }}^{\left\{ 1,2\right\} },\dots,k_{1}^{\left\{ 1,n\right\} },\dots,k_{q_{\left\{ 1,n\right\} }}^{\left\{ 1,n\right\} },-\sum_{i=1}^{n}\tilde{A}_{i}\right)p_{a_{1}^{1}}\cdots p_{a_{\tilde{m}_{1}}^{1}}e^{x\sqrt{-1}\sum_{i=1}^{n}\tilde{A}_{i}}\label{eq: Produit n hamiltonniens en fonction de P}\\
 & \times\prod_{i=2}^{n}\frac{\epsilon^{2l_{i}}\left(\sqrt{-1}\hbar\right)^{g_{i-l_{i}}}}{\tilde{m}_{i}!}P_{d_{i},g_{i},l_{i}}\left(a_{1}^{i},\dots,a_{\tilde{m}_{1}}^{i},-k_{1}^{\left\{ 1,i\right\} },\dots,-k_{q_{\left\{ 1,i\right\} }}^{\left\{ 1,i\right\} },\dots,k_{1}^{\left\{ i,n\right\} },\dots,k_{q_{\left\{ i,n\right\} }}^{\left\{ i,n\right\} },0\right)p_{a_{1}^{i}}\cdots p_{a_{\tilde{m}_{i}}^{i}},\nonumber 
\end{align}
where
\begin{itemize}
\item $\tilde{A}_{i}=\sum_{j=1}^{\tilde{m}_{i}}a_{j}^{i}$, with $1\leq i\leq n$,
\item the conditions $\alpha$ are
\[
g_{1}+\left(g_{1}-l_{1}\right)+\tilde{m}_{1}+\sum_{I\in\mathcal{C}_{1}}q_{I}=d_{1}+1
\]
and
\[
g_{i}+\left(g_{i}-l_{i}\right)+\tilde{m}_{i}+\sum_{I\in\mathcal{C}_{i}}q_{I}=d_{i}+2,\,\,{\rm when}\,\,i\geq2,
\]
\item the conditions $\beta$ are
\[
\tilde{A}_{i}-\sum_{j=1}^{i-1}K^{\left\{ j,i\right\} }+\sum_{j=i+1}^{n}K^{\left\{ i,j\right\} }=0,\,\,\,\,\,2\leq i\leq n.
\]
\end{itemize}
We now modify the notations by removing the tildes, i.e. we set $m_{i}:=\tilde{m}_{i}$
and $A_{i}:=\tilde{A}_{i}$ for any $1\leq i\leq n$.

In Step $1.2$ of Section~\ref{subsec:Less-trivial-case} we explained
why such an expression of $H_{d_{1}}\star\overline{H}_{d_{2}}\star\cdots\star\overline{H}_{d_{n}}$
is enough to get an expression for $\langle\tau_{0}\tau_{d_{1}}\dots\tau_{d_{n}}\rangle_{l,g-l}$.
We also explained that the commutators offer some simplifications
given by the conditions $\gamma$.

We now proceed as in Step $2$ of Section~\ref{subsec:Less-trivial-case}:
we extract the coefficient of $\epsilon^{2l}\hbar^{g-l}$ in the expression
of $\frac{\sqrt{-1}^{g-l}}{\hbar^{n-1}}H_{d_{1}}\star\overline{H}_{d_{2}}\star\cdots\star\overline{H}_{d_{n}}$
given by Eq.~(\ref{eq: Produit n hamiltonniens en fonction de P}),
then we use the Lemma~\ref{lem:evaluation } to substitute $u_{i}=\delta_{i,1}$
in this coefficient. We get Eq.~(\ref{eq: Formula correlator in terms of P}).
\end{proof}
We now prove that the correlator $\langle\tau_{0}\tau_{d_{1}}\dots\tau_{d_{n}}\rangle_{l,g-l}$
vanishes when $\sum d_{i}>4g-2+n-l$ or when $\sum d_{i}$ has the
parity of $n+1-l$. According to Section~\ref{subsec: String equation level structure},
this proves Proposition~\ref{prop: top vanishing correlators }. 
\begin{proof}
[Proof of Proposition \ref{prop: top vanishing correlators }]The
three last lines of Eq.~(\ref{eq: Formula correlator in terms of P})
form a polynomial the indeterminates $a_{1}^{1},\dots,a_{m_{1}}^{1},\dots,$
$a_{1}^{n},\dots,a_{m_{n}}^{n}$ depending of the set of parameters
$\mathcal{S}=\left\{ d_{i},g_{i},l_{i},m_{i}\vert1\leq i\leq n\right\} $,
we denote this polynomial by $Q_{\mathcal{S}}$. The parity and the
degree of this polynomial is described by Lemma~(\ref{lem: multi Ehrhart polynomials}).
Since the polynomial $P_{d,g,l}$ is even and of degree $2g$, we
deduce that $Q_{\mathcal{S}}$ is a polynomial of degree
\[
2\sum_{I\in\mathcal{C}}q_{I}-\left(n-1\right)+\sum_{i=1}^{n}2g_{i}=2g+n-1
\]
 and has the parity of $n-1$. We used the constraint $\sum_{I\in\mathcal{C}}q_{I}+\sum_{i=1}^{n}g_{i}=g+n-1$
of the first summation in formula~(\ref{eq: Formula correlator in terms of P})
to obtain this equality.

We then extract the coefficient of $a_{1}^{1}\cdots a_{m_{1}}^{1}\cdots a_{1}^{n}\cdots a_{m_{n}}^{n}$
in $\mathcal{Q}_{S}$. This coefficient vanishes if 
\[
\sum_{i=1}^{n}m_{i}>2g+n-1
\]
or if $\sum_{i=1}^{n}m_{i}$ has the parity of $n$. However conditions
$\alpha$ give $\sum_{i=1}^{n}m_{i}=\sum_{i=1}^{n}d_{i}-2g+l+1$.
Hence, this coefficient vanishes if 
\[
\sum d_{i}>4g-2+n-l
\]
or if $\sum d_{i}$ has the parity of $n-l+1$. This proves the proposition.
\end{proof}

\appendix

\section{\label{Appendix A}Proofs of the quantum integrability and the tau
symmetry}

The Hamiltonian density $H_{d}$ is the same than the one given in
\cite{Buryak_2019}, although it is introduced differently.
Let us verify that the two definitions lead to the same object. In
particular, this will prove Propositions~\ref{prop: quantum integrability}~and~\ref{prop: tau symmetry}
which are proved in \cite{Buryak_2019} with their definition
of $H_{d}$. 

In \cite{Buryak_2019}, the authors defined
\[
H_{d}^{BDGR}:=\sum_{s\geq0}\left(-\partial_{x}\right)^{s}\frac{\partial G_{d+1}}{\partial u_{s}},
\]
where
\begin{align*}
G_{d}: & =\sum_{\underset{2g-1+n>0}{g\geq0,n\geq0}}\frac{\left(i\hbar\right)^{g}}{n!}\sum_{a_{1},\dots,a_{n}\in\mathbb{Z}}\left(\int_{{\rm DR}_{g}\left(-\sum a_{i},a_{1},\dots,a_{n}\right)}\psi_{1}^{d}\Lambda\left(\frac{-\epsilon^{2}}{i\hbar}\right)\right)p_{a_{1}}\cdots p_{a_{n}}e^{ix\sum a_{i}}.
\end{align*}

\begin{lem}
Let $\phi\in\tilde{\mathcal{A}}$. We have
\begin{equation}
\sum_{s\geq0}\left(-\partial_{x}\right)^{s}\frac{\partial\phi}{\partial u_{s}}=\sum_{b\in\mathbb{Z}}e^{-ibx}\frac{\partial\overline{\phi}}{\partial p_{b}}.\label{eq: Derivee Fonctionnelle}
\end{equation}
\end{lem}
\begin{proof}
Write $\phi$ as
\begin{align*}
\phi\left(x\right) & =\sum_{k=0}^{d}\sum_{a_{1},...,a_{k}\in\mathbb{Z}}\phi_{k}\left(a_{1},\dots,a_{k}\right)p_{a_{1}}\cdots p_{a_{k}}e^{ix\sum a_{i}}
\end{align*}
where $\phi_{k}\left(a_{1},\dots,a_{k}\right)\in\mathbb{C}\left[a_{1},\dots,a_{k}\right]\left[\left[\epsilon,\hbar\right]\right]$
is a symmetric polynomial in its $k$ indeterminates $a_{1},\dots,a_{k}$
for $0\leq k\leq d$.

We start from the RHS of Eq.~(\ref{eq: Derivee Fonctionnelle}).
Using this expression of $\phi$, we get
\begin{align*}
\sum_{b\in\mathbb{Z}}e^{-ibx}\frac{\partial\overline{\phi}}{\partial p_{b}} & =\sum_{k\geq1}^{d}k\sum_{\underset{\sum_{i=1}^{k-1}a_{i}+b=0}{a_{1},\dots,a_{k-1},b\in\mathbb{Z}}}\phi_{k}\Big(a_{1},\dots,a_{k-1},b\big)p_{a_{1}}\cdots p_{a_{k-1}}e^{-ixb}\\
 & =\sum_{k=1}^{d}k\sum_{a_{1},...,a_{k-1}\in\mathbb{Z}}\phi_{k}\Big(a_{1},,\dots,a_{k-1},-\sum_{i=1}^{k-1}a_{i}\Big)p_{a_{1}}\cdots p_{a_{k-1}}e^{ix\sum_{i=1}^{k-1}a_{i}}.
\end{align*}

We will obtain this same expression from the LHS. Note that we can
rewrite $\phi$ as 
\[
\phi\left(x\right)=\sum_{k=0}^{d}\sum_{s_{1},\dots,s_{k}\geq0}\left(-i\right)^{\sum_{i=1}^{k}s_{i}}u_{s_{1}}\cdots u_{s_{k}}\left[a_{1}^{s_{1}}\cdots a_{k}^{s_{k}}\right]\phi_{k}\left(a_{1},\dots,a_{k}\right).
\]
Hence we find
\begin{align*}
\frac{\partial\phi}{\partial u_{s}} & =\sum_{k=1}^{d}\sum_{j=1}^{k}\sum_{a_{1},\dots,\hat{a}_{j},\dots,a_{k}\in\mathbb{Z}}\left(-i\right)^{s}\left[a_{j}^{s}\right]\phi_{k}\left(a_{1},\dots,a_{k}\right)p_{a_{1}}\cdots\hat{p}_{a_{j}}\cdots p_{a_{k}}e^{ix\sum_{\underset{i\neq j}{i=1}}^{k}a_{i}}\\
 & =\sum_{k=1}^{d}k\sum_{a_{1},\dots,a_{k-1}\in\mathbb{Z}}\left(-i\right)^{s}\left[a_{k}^{s}\right]\phi_{k}\left(a_{1},\dots,a_{k}\right)p_{a_{1}}\cdots p_{a_{k-1}}e^{ix\sum_{i=1}^{k-1}a_{i}}.
\end{align*}
Then we get
\[
\sum_{s\geq0}\left(-\partial_{x}\right)^{s}\frac{\partial\phi}{\partial u_{s}}=\sum_{k=1}^{d}\sum_{a_{1},\dots,a_{k-1}\in\mathbb{Z}}\sum_{s\geq0}\left(-\sum_{i=1}^{k-1}a_{i}\right)^{s}\left[a_{k}^{s}\right]\phi_{k}\left(a_{1},\dots,a_{k}\right)p_{a_{1}}\cdots p_{a_{k-1}}e^{ix\sum_{i=1}^{k-1}a_{i}}.
\]
Moreover we have $\sum_{s\geq0}\left(-\sum_{i=1}^{k-1}a_{i}\right)^{s}\left[a_{k}^{s}\right]\phi_{k}\left(a_{1},\dots,a_{k}\right)=\phi_{k}\left(a_{1},\dots,-\sum_{i=1}^{k-1}a_{i}\right)$.
We then obtain the expression of the RHS. 
\end{proof}
\begin{prop}
Fix $d\geq0$. We have
\[
H_{d}=H_{d}^{BDGR}.
\]
\end{prop}
\begin{proof}
Using the preceding lemma and the expression of $G_{d+1}$, we find
\[
H_{d}^{BDGR}=\sum_{b\in\mathbb{Z}}e^{-ibx}\frac{\partial\overline{G}_{d+1}}{\partial p_{b}}=H_{d}.
\]
\end{proof}

\bibliographystyle{alpha}
\bibliography{Bibliographie}

\begin{thebibliography}{BDGR19}

\bibitem[BDGR18]{buryak2018tau}
Alexandr Buryak, Boris Dubrovin, J{\'e}r{\'e}my Gu{\'e}r{\'e}, and Paolo Rossi.
\newblock Tau-structure for the double ramification hierarchies.
\newblock {\em Communications in Mathematical Physics}, 363(1):191--260, 2018.

\bibitem[BDGR19]{Buryak_2019}
Alexandr Buryak, Boris Dubrovin, J{\'{e}}r{\'{e}}my Gu{\'{e}}r{\'{e}}, and
  Paolo Rossi.
\newblock Integrable systems of double ramification type.
\newblock {\em International Mathematics Research Notices},
  2020(24):10381--10446, feb 2019.

\bibitem[BGR19]{BGR_2019}
Alexandr Buryak, J{\'{e}}r{\'{e}}my Gu{\'{e}}r{\'{e}}, and Paolo Rossi.
\newblock {DR}/{DZ} equivalence conjecture and tautological relations.
\newblock {\em Geometry {\&} Topology}, 23(7):3537--3600, dec 2019.

\bibitem[BR16]{BuryakRossi2016}
Alexandr Buryak and Paolo Rossi.
\newblock Double ramification cycles and quantum integrable systems.
\newblock {\em Letters in mathematical physics}, 106(3):289--317, 2016.

\bibitem[BSSZ15]{buryak2015integrals}
Alexandr Buryak, Sergey Shadrin, Loek Spitz, and Dimitri Zvonkine.
\newblock Integrals of $\psi$-classes over double ramification cycles.
\newblock {\em American Journal of Mathematics}, 137(3):699--737, 2015.

\bibitem[Bur15]{buryak2015double}
Alexandr Buryak.
\newblock Double ramification cycles and integrable hierarchies.
\newblock {\em Communications in Mathematical Physics}, 336(3):1085--1107,
  2015.

\bibitem[DZ01]{dubrovin2001normal}
Boris Dubrovin and Youjin Zhang.
\newblock Normal forms of hierarchies of integrable pdes, frobenius manifolds
  and gromov-witten invariants.
\newblock {\em arXiv preprint math/0108160}, 2001.

\bibitem[ELSV01]{Ekedahl_2001}
Torsten Ekedahl, Sergei Lando, Michael Shapiro, and Alek Vainshtein.
\newblock Hurwitz numbers and intersections on moduli spaces of curves.
\newblock {\em Inventiones Mathematicae}, 146(2):297--327, nov 2001.

\bibitem[FP05]{faber2005relative}
Carel Faber and Rahul Pandharipande.
\newblock Relative maps and tautological classes.
\newblock {\em Journal of the European Mathematical Society}, 7(1):13--49,
  2005.

\bibitem[FP15]{FaberPandha15_tautologica}
C.~{Faber} and R.~{Pandharipande}.
\newblock {Tautological and non-tautological cohomology of the moduli space of
  curves}.
\newblock In {\em {Handbook of moduli. Volume I}}, pages 293--330. Somerville,
  MA: International Press; Beijing: Higher Education Press, 2015.

\bibitem[FSZ10]{faber2010tautological}
Carel Faber, Sergey Shadrin, and Dimitri Zvonkine.
\newblock Tautological relations and the $ r $-spin witten conjecture.
\newblock In {\em Annales scientifiques de l'{\'E}cole Normale Sup{\'e}rieure},
  volume~43, pages 621--658, 2010.

\bibitem[GJV05]{goulden2005towards}
IP~Goulden, David~M Jackson, and Ravi Vakil.
\newblock Towards the geometry of double hurwitz numbers.
\newblock {\em Advances in Mathematics}, 198(1):43--92, 2005.

\bibitem[JPPZ17]{janda2017double}
Felix Janda, Rahul Pandharipande, Aaron Pixton, and Dimitri Zvonkine.
\newblock Double ramification cycles on the moduli spaces of curves.
\newblock {\em Publications math{\'e}matiques de l'IH{\'E}S}, 125(1):221--266,
  2017.

\bibitem[Kon92]{kontsevich1992intersection}
Maxim Kontsevich.
\newblock Intersection theory on the moduli space of curves and the matrix airy
  function.
\newblock {\em Communications in Mathematical Physics}, 147(1):1--23, 1992.

\bibitem[Oko00]{Okounkov_2000}
Andrei Okounkov.
\newblock Toda equations for hurwitz numbers.
\newblock {\em Mathematical Research Letters}, 7(4):447--453, 2000.

\bibitem[OP06]{Okounkov_2006}
Andrei Okounkov and Rahul Pandharipande.
\newblock Gromov{\textendash}witten theory, hurwitz theory, and completed
  cycles.
\newblock {\em Annals of Mathematics}, 163(2):517--560, mar 2006.

\bibitem[Pet15]{petersen2015eulerian}
T~Kyle Petersen.
\newblock Eulerian numbers.
\newblock In {\em Eulerian Numbers}, pages 3--18. Springer, 2015.

\bibitem[Wit90]{witten1990two}
Edward Witten.
\newblock Two-dimensional gravity and intersection theory on moduli space.
\newblock {\em Surveys in differential geometry}, 1(1):243--310, 1990.

\bibitem[Wit93]{witten1993algebraic}
Edward Witten.
\newblock Algebraic geometry associated with matrix models of two-dimensional
  gravity.
\newblock {\em Topological methods in modern mathematics (Stony Brook, NY,
  1991)}, 235, 1993.

\end{thebibliography}

\end{document}